\documentclass[11pt,english]{article}
\usepackage[T1]{fontenc}
\usepackage[latin9]{inputenc}
\usepackage{geometry}
\usepackage{color}
\usepackage{babel}
\usepackage{float}
\usepackage{amsmath}
\usepackage{amssymb}
\usepackage{graphicx}
\usepackage{setspace}
\usepackage[authoryear]{natbib}
\usepackage[unicode=true,pdfusetitle,
 bookmarks=true,bookmarksnumbered=false,bookmarksopen=false,
 breaklinks=false,pdfborder={0 0 0},pdfborderstyle={},backref=false,colorlinks=true]
 {hyperref}
\hypersetup{
 citecolor = blue}

\makeatletter

\providecommand{\tabularnewline}{\\}

\usepackage{amsthm}

\usepackage{amsfonts}%\usepackage{chicago}

\newcounter{Acnt}

\newcounter{Acntp}
\newtheorem{Ap-item}[Acntp]{ }

\newcounter{Acnta}
\newtheorem{Aa-item}[Acnta]{ }

\newcounter{Acntb}
\newtheorem{Ab-item}[Acntb]{ }

\newtheorem{assumption}[Acnt]{Assumption}

\newtheorem{theorem}{Theorem}[section]
\newtheorem{proposition}{Proposition}[section]

\newtheorem{corollary}[theorem]{Corollary}

\newtheorem{lemma}{Lemma}[section]

\numberwithin{Acnt}{section}
\makeatother

\begin{document}
\title{Partial Identification and Inference for \\
Conditional Distributions of Treatment Effects\thanks{This is a revised version of Chapter 3 of author's dissertation. I
am very grateful to Jason Abrevaya, Stephen Donald, Sukjin Han, Yu-Chin
Hsu, and participants in 2022 IAAE Annual Conference, 2022 KES Annual
Conference, and the seminar at Sungkyunkwan University for their valuable
comments and suggestions. Jiyong Jung provided excellent research
assistance. This paper was supported by the New Faculty Research Funding
provided by Sogang University. All errors are mine.}}
\author{Sungwon Lee\thanks{35 Baekbeom-ro Mapo-gu, Sogang University GN 710, Seoul 04107, South
Korea. }\\
Department of Economics\\
Sogang University\\
\href{mailto:sungwonlee@sogang.ac.kr}{sungwonlee@sogang.ac.kr}\\
}
\date{This draft: \today }
\maketitle
\begin{abstract}
This paper considers identification and inference for the distribution
of treatment effects conditional on observable covariates. Since the
conditional distribution of treatment effects is not point identified
without strong assumptions, we obtain bounds on the conditional distribution
of treatment effects by using the Makarov bounds. We also consider
the case where the treatment is endogenous and propose two stochastic
dominance assumptions to tighten the bounds. We develop a nonparametric
framework to estimate the bounds and establish the asymptotic theory
that is uniformly valid over the support of treatment effects. An
empirical example illustrates the usefulness of the methods. \textit{\emph{ }}\\
\textit{\emph{}}
\end{abstract}
\newpage{}

\section{\label{sec:Introduction} Introduction }

This paper considers identification and estimation of the conditional
distribution of treatment effects.\footnote{The distribution of treatment effects has received significant attention,
and its importance becomes greater when the benefit from the treatment
is non-transferrable (\citet{heckman2007econometric}).} While the unconditional distribution of treatment effects has been
considered extensively in the literature in both theoretical and applied
econometrics (e.g., \citet{heckman1997making,fan2010sharp,fan2012confidence,fan2017partial,firpo2019partial,frandsen2021partial}),
we focus on the conditional distribution of treatment effects to take
into account the heterogeneity caused by differences in the value
of covariates.\footnote{For example, the effect of a mother's smoking behavior on her child's
birth weight might differ across the mother's age (e.g., \citet{abrevaya2015estimating}).} Such heterogeneous treatment effects, if they exist, may help policy
makers develop more effective policies. We use the \emph{Makarov bounds},
which depend on the marginal distributions of potential outcomes while
the dependence structure between them is left unspecified, to partially
identify the conditional distribution of treatment effects (e.g.,
\citet{makarov1982estimates}, \citet{ruschendorf1982random}, and
\citet{williamson1990probabilistic}). We then provide estimation
and inference methods for the bounds. 

We start with the case where the assumption of the conditional independence
of the treatment and potential outcomes given covariates, which is
called an unconfoundedness condition, is satisfied. However, there
are many situations where such a conditional independence assumption
fails to hold. We present identification results of the distribution
of treatment effects when the treatment is endogenous \textit{without}
assuming the existence of an instrumental variable. We show that one
can still obtain Makarov-type bounds on the distribution of treatment
effects even with an endogenous treatment. When a treatment is endogenous,
the bounds on the distribution of the treatment effect may be too
wide to be informative.\footnote{This is also true for the case where we impose the unconfoundedness
assumption. We can expect that, since we are interested in the distribution
of treatment effects for some subgroup, the bounds on the conditional
distribution may be narrower than those on the unconditional distribution.
This is another motivation for considering the conditional distributions
of treatment effects in this paper.} Motivated by the previous studies in the literature (e.g., \citet{manski1997monotone,manski2000monotone,blundell2007changes}),
we consider two kinds of stochastic dominance between potential outcomes
in this paper to tighten the bounds. These stochastic dominance assumptions
are useful in many empirical situations as they are consistent with
many economic theories. The resulting bounds under the stochastic
dominance assumptions are easy to compute. 

We construct nonparametric kernel-type estimators of the bounds on
the conditional distribution of treatment effects and establish asymptotic
theory for them that is uniformly valid over the support of treatment
effects for some fixed subgroup of the population. They are useful
for comparing bounds between two subpopulations defined in terms of
different values of observable characteristics.\footnote{One can perform statistical testing for global hypotheses, such as
equality of or stochastic dominance between two bounds (e.g., \citet{barrett2003consistent,lee2009testing,seo2018tests}),
and this requires asymptotic theory uniformly valid over the support
of treatment effects. } We adapt the results on uniform inference for value functions that
were developed recently by \citet{firpo2021uniform} and provide a
set of conditions under which the estimated bounds are consistent
for the true bounds. We propose a bootstrap procedure to mimic the
asymptotic distributions of the estimated bounds and show its validity.
The bootstrap scheme in this paper is based on the novel bootstrap
procedure for Hadamard directionally differentiable functionals that
was developed by \citet{fang2019inference}.

The asymptotic theory developed in this paper relaxes some undesirable
assumptions imposed in several existing studies in the literature. Specifically,
the bounds on the distribution of treatment effects are defined as
some functionals of marginal distributions of potential outcomes.
These functionals involve the infimum and supremum over the supports
of potential outcomes. The uniqueness of the infimum and supremum
needs to be imposed to establish a pointwise asymptotic theory for
those bounds, as in \citet{fan2010sharp,fan2012confidence}.\footnote{\citet{fan2010sharp} impose Assumptions 3 and 4 to guarantee the
uniqueness of the maximizer and minimizer involved in the bounds on
the distribution of treatment effects. \citet{fan2012confidence}
impose Assumptions A3 and A4 for the same purpose. } However, the uniqueness assumption may not hold for some data generating
process (DGP), and the pointwise asymptotic results may not be valid
in such cases (cf. \citet{milgrom2002envelope}, \citet{firpo2021uniform}).
On the other hand, the inference results developed in this paper do
not require such uniqueness assumptions. 

One can use the novel approach developed by \citet{chernozhukov2013intersection}
that yields confidence bands uniformly valid over the joint support
of an outcome variable and covariates.\footnote{\citet{chernozhukov2013intersection} do not explicitly provide an
inference result for nonparametric conditional distribution estimators
that is uniformly valid over the joint support of an outcome variable
and covariates. However, one can easily adapt their results in the
literature to establish the strong Gaussian approximation of some
nonparametric estimator of a conditional distribution function. This
allows one to perform inference uniformly over the joint support of
the outcome variable of interest and conditioning variables (e.g.,
\citet{chernozhukov2014gaussian}).} Their inference methods for bounds defined by either supremum or
infimum of some function are based on the strong Gaussian approximation
using couplings. A key ingredient of their approach is the construction
of argsup and/or arginf sets, and it is allowed to avoid imposing
the uniqueness assumption on those sets. Our approach differs from
theirs in that we rely on weak convergence of nonparametric estimators
of the bounds to some fixed Gaussian process and the results on Hadamard
directionally differentiable functionals. In addition, while our approach
also requires to estimate argsup and arginf sets for the Hadamard
directional derivatives, the construction of these sets is computationally
easy in comparison to that of \citet{chernozhukov2013intersection}.
Therefore, the inference theory in this paper can be considered complementing
the existing inference methods for the Makarov bounds.

The Monte Carlo simulation results show that the inference methods
in this paper perform well in finite samples. We also provide an empirical
example to illustrate the practical relevance of the methods proposed
in this paper. We revisit the empirical question of the effect of
401(k) plans on net asset accumulations investigated by multiple papers
in the literature (e.g., \citet{Abadie2003,Chernozhukov2004,Wuethrich2019,SantAnna2022}).
We confirm that the identifying power of the stochastic dominance
assumptions is substantial from the empirical application. Furthermore,
we find evidence on heterogeneity in the treatment effect that is
consistent with the results in the literature. 

This paper contributes to several strands of the literature on treatment
effects. First, this paper mainly contributes to the literature on
identification and estimation of the distribution of treatment effects
by providing nonparametric estimation and inference methods for the
conditional distribution of treatment effects. The literature is too
vast to list all the papers relevant to this point. Identification
of the unconditional distribution of treatment effects has been studied
by, for example, \citet{heckman1997making}, \citet{fan2010sharp,fan2012confidence},
\citet{fan2017partial}, \citet{vuong2017counterfactual}, \citet{kim2018identifying},
and \citet{firpo2019partial}. 

This paper is closely related to \citet{kim2018partial} and \citet{frandsen2021partial}.
\citet{kim2018partial} provides identification results under several
distributional restrictions, together with an instrumental variable,
that help tighten the bounds on the distribution of treatment effects,
when the treatment is endogenous. The restrictions considered by \citet{kim2018partial}
are general and closely related to the stochastic dominance assumptions
in this paper that were proposed independently of the results in \citet{kim2018partial}.
This paper differs from \citet{kim2018partial} in that this paper
focuses on estimation and inference for the conditional distribution
of treatment effects with a motivation for treatment effect heterogeneity,
whereas \citet{kim2018partial} mainly considers identification of
the unconditional distribution of treatment effects under some restrictions
on the model. Moreover, it is much easier to compute and estimate
the bounds proposed in this paper than those in \citet{kim2018partial}.
Therefore, the bounds in this paper have great applicability to empirical
research. 

In comparison with \citet{frandsen2021partial}, we focus on estimation
and inference for the bounds on the conditional distribution of treatment
effects with a motivation for heterogeneity across subgroups. On the
other hand, \citet{frandsen2021partial} are mainly concerned with
identification of the unconditional distribution of treatment effects
under a novel restriction that is called the ``stochastic increasingness''
assumption. \citet{frandsen2021partial} also discuss how to incorporate
covariates and estimate the bounds on the unconditional distribution
of treatment effects. However, they do not develop the asymptotic
theory for the conditional distribution of treatment effects, especially
when the bounds on the conditional distribution are the Makarov bounds. 

This paper also contributes to the literature on heterogeneity in
treatment effects across subpopulations by considering the conditional
distribution of treatment effects (e.g., \citet{donald2012incorporating,abrevaya2015estimating,chang2015nonparametric,hsu2017consistent,shen2019estimation}).
Most existing studies focus on average/quantile treatment effects
or the marginal distributions of potential outcomes. The results of
this paper complement them by considering conditional distributions
of treatment effects. 

The rest of this paper is organized as follows. Section \ref{sec:ID}
presents the model and identification results. Section \ref{sec:Estimation}
provides nonparametric estimators of the Makarov bounds on the conditional
distribution of treatment effects and the bootstrap procedure. Section
\ref{sec:Inference} develops the asymptotic theory for the estimated
bounds on the conditional distribution of treatment effects. Section
\ref{sec:MC} presents the Monte Carlo simulation study, and Section
\ref{sec:Empirical} provides the empirical example considering the
effect of 401(k) plan on net financial asset accumulations. We then
conclude with Section \ref{sec:Conclusion}. All mathematical proofs,
technical expressions, and additional results are presented in Appendix. 

Before proceeding, we introduce some notation that will be used throughout
this paper. Two random variables $A$ and $B$ that are independent
of each other are denoted by $A\perp B$, and $\mathbb{E}[\cdot]$
is the expectation operator. For a matrix $A$, $A^{t}$ is the transpose
of $A$. For a set $A$, $l^{\infty}(A)$ denotes the set of uniformly
bounded functions on $A$. For a sequence of random variables $(Z_{n})$
and a random variable $Z$, we denote the weak convergence of $Z_{n}$
to $Z$ by $Z_{n}\Rightarrow Z$.\footnote{A formal definition of weak convergence can be found in, for example,
\citet[pp.107-108]{kosorok2008}.} For a set $A$, $int(A)$ is the interior of $A$.

\section{\label{sec:ID} Model and Identification}

\subsection{Identification Under the Unconfoundedness Condition}

Let $(\Omega,\mathcal{S},\mathbf{P})$ be a probability space. Let
$D$ be a binary variable that indicates whether a person gets the
treatment or not, in other words, $D=1$ if the person gets the treatment,
and $D=0$ if the person does not get the treatment. For each $d\in\{0,1\}$,
let $Y_{d}$ denote the potential outcome when $D=d$, and the observed
outcome is defined as $Y\equiv DY_{1}+(1-D)Y_{0}$. We assume that
$Y_{d}$ is a continuous random variable for each $d\in\{0,1\}$.
Let $X\in\mathbb{R}^{d_{x}}$ be a set of covariates and denote its
support by $\mathcal{X}$. We can only observe $(Y,D,X^{t})^{t}$
from the data. We denote the conditional distribution of $Y_{d}$
on $X=x$ by $F_{d|X}(\cdot|x)$ for each $d\in\{0,1\}$, and let
$p_{0}(x)\equiv\Pr(D=1|X=x)$ for a given $x\in\mathcal{X}$. $F_{X}(\cdot)$
denotes the distribution function of $X$. 

We begin with the models under the conditional independence of the
treatment variable and impose the following assumptions. 

\begin{assumption}\label{assu:unconfounded} $(Y_{1},Y_{0})\perp D|X$.
\end{assumption}

\begin{assumption}\label{assu:bound_propensity} There exist $\text{\ensuremath{\underline{p},\bar{p}\in(0,1)} such that \ensuremath{p_{0}(x)\in[\underline{p},\bar{p}]} uniformly in \ensuremath{x\in\mathcal{X}}}$.
\end{assumption} 

Assumption \ref{assu:unconfounded} is the unconfoundedness assumption,
which means that the treatment $D$ is independent of the potential
outcomes conditional on $X$. Assumption \ref{assu:bound_propensity}
is an overlap condition that is commonly imposed in the relevant literature
(e.g., \citet{imbens2004nonparametric}). This assumption implies
that we can observe individuals with $Y=Y_{1}$ and those with $Y=Y_{0}$
for any value of $x\in\mathcal{X}$. Under these assumptions, we have
the following identification result: 

\begin{lemma}\label{lem:id_unconfounded} Let $x\in\mathcal{X}$
be given. Suppose that Assumptions \ref{assu:unconfounded} and \ref{assu:bound_propensity}
are satisfied. For any measurable function $G$ such that $\mathbb{E}[|G(Y_{1})||X=x]<\infty$
and $\mathbb{E}[|G(Y_{0})||X=x]<\infty$, we have 
\begin{equation}
\begin{aligned}\mathbb{E}[G(Y_{1})|X=x] & =\frac{\mathbb{E}[D\cdot G(Y)|X=x]}{\mathbb{E}[D|X=x]},\\
\mathbb{E}[G(Y_{0})|X=x] & =\frac{\mathbb{E}[(1-D)\cdot G(Y)|X=x]}{\mathbb{E}[(1-D)|X=x]}.
\end{aligned}
\label{eq:ID_Marginal_Unconfounded}
\end{equation}
 \end{lemma} 

Lemma \ref{lem:id_unconfounded} implies that we can identify the
conditional distributions of $Y_{1}$ and $Y_{0}$ on $X=x$ when
we choose $G(Y)\equiv\mathbf{1}(Y\leq y)$ for a given $y\in\mathbb{R}$.
This result is also closely related to the identification of the unconditional
distributions of $Y_{1}$ and $Y_{0}$ that is considered by \citet{donald2014estimation}. 

The treatment effect $\Delta$ is defined as the difference between
$Y_{1}$ and $Y_{0}$ (i.e., $\Delta\equiv Y_{1}-Y_{0}$). The parameter
of interest is the conditional distribution of treatment effects given
$X=x$ for a given value $x\in\mathcal{X}$, and this conditional
distribution is denoted by $F_{\Delta|X}(\cdot|x)$. 

The (unconditional) distribution of treatment effects has received
a considerable amount of attention from the literature (e.g., \citet{heckman1997making},
\citet{fan2010sharp,fan2012confidence}, \citet{fan2017partial},
\citet{firpo2019partial}). The distribution function is useful in
the context of program evaluation, for example, to determine the proportion
of people who benefit from being treated, namely, $\Pr(\Delta\geq0)=1-F_{\Delta}(0)$.
One can refer to \citet{abbring2007econometric} for more discussion
on the topic. The conditional distributions takes potential heterogeneity
across subpopulations into account and thus can provide much richer
information on treatment effects. Such heterogeneity, if it exists,
may deliver different policy implications for different subpopulations. 

It is worth noting that even in a randomized experiment where one
can point identify the conditional distributions of $Y_{1}$ and $Y_{0}$
given $X$, the distribution of treatment effects is not point identified
without additional structures of the model.\footnote{When the conditional distributions of $Y_{1}$ and $Y_{0}$ given
$X$ are point identified, a sufficient condition for point identification
of the joint distribution of $Y_{1}$ and $Y_{0}$ is the (conditional)
rank invariance; that is, $F_{1|X}(Y_{1}|X)=F_{0|X}(Y_{0}|X)$ almost
surely. } In particular, \citet{fan2010sharp} provide sharp bounds on the
distribution of treatment effects in randomized experiments that are
based on \citet{makarov1982estimates}, and \citet{firpo2019partial}
improve the bounds of \citet{fan2010sharp} in the uniform sense.
We do not focus on the uniform sharpness of the bounds but on pointwise
sharp bounds of the conditional distribution of treatment effects
and nonparametric estimation of the bounds. We recall the bounds on
the conditional distribution of treatment effects $F_{\Delta|X}(\cdot|x)$: 

\begin{proposition}\label{prop:ID_dist_TE_Unconfounded} (Lemma 2.1
in \citet{fan2010sharp}) Let $x\in\mathcal{X}$ be fixed. Then, for
a given $\delta\in\mathbb{R}$, define 
\begin{eqnarray}
F_{\Delta|X}^{L}(\delta|x) & = & \max\left(\sup_{y\in\mathbb{R}}\left\{ F_{1|X}(y|x)-F_{0|X}(y-\delta|x)\right\} ,0\right),\label{eq:ID_F_Delta_LB}\\
F_{\Delta|X}^{U}(\delta|x) & = & \min\left(\inf_{y\in\mathbb{R}}\left\{ F_{1|X}(y|x)-F_{0|X}(y-\delta|x)\right\} ,0\right)+1.\label{eq:ID_F_Delta_UB}
\end{eqnarray}
 If Assumptions \ref{assu:unconfounded} and \ref{assu:bound_propensity}
are satisfied, we then have 
\begin{equation}
F_{\Delta|X}(\delta|x)\in\left[F_{\Delta|X}^{L}(\delta|x),F_{\Delta|X}^{U}(\delta|x)\right].\label{eq:ID_F_Delta}
\end{equation}
 \end{proposition} 

Note that this result is essentially identical to the identification
result in \citet{fan2010sharp}. Without additional structures on
the model, these bounds are sharp in the pointwise sense for given
$x\in\mathcal{X}$. It is also worth noting that one can derive bounds
on the unconditional distribution of treatment effects by averaging
the bounds on the conditional distribution. Specifically, if the potential
outcomes are correlated with $X$, then the bounds on the unconditional
distribution obtained from the conditional distributions of the potential
outcomes are tighter than the those obtained from the unconditional
distributions of the potential outcomes (cf. \citet{firpo2019partial}). 

\subsection{Identification with an Endogenous Treatment}

In this section, we discuss identification of the distribution of
treatment effects with an endogenous binary treatment. When the treatment
is endogenously determined, the conditional distributions of $Y_{1}$
and $Y_{0}$ given $X$ are generally partially identified without
additional assumptions. The following theorem shows that even if the
conditional distributions of $Y_{1}$ and $Y_{0}$ given $X$ are
partially identified, we can still bound $F_{\Delta|X}(\delta)$. 

\begin{theorem}\label{thm:ID_Endo_FH_Bound} Let $x\in\mathcal{X}$
be fixed. Suppose that, for all $y\in\mathbb{R}$, the identified
sets of $F_{1|X}(y|x)$ and $F_{0|X}(y|x)$ are given by $\left[F_{1|X}^{L}(y|x),F_{1|X}^{U}(y|x)\right]$,
and $\left[F_{0|X}^{L}(y|x),F_{0|X}^{U}(y|x)\right]$, respectively. For a given $\delta\in Supp(\Delta|X=x)$, define 
\begin{eqnarray}
F_{\Delta|X}^{e,L}(\delta|x) & = & \max\left(\sup_{y}\left\{ F_{1|X}^{L}(y|x)-F_{0|X}^{U}(y-\delta|x)\right\} ,0\right),\label{eq:trt1}\\
F_{\Delta|X}^{e,U}(\delta|x) & = & \min\left(\inf_{y}\left\{ F_{1|X}^{U}(y|x)-F_{0|X}^{L}(y-\delta|x)\right\} ,0\right)+1.\label{eq:trt2}
\end{eqnarray}
 Then, 
\[
F_{\Delta|X}^{e,L}(\delta|x)\leq F_{\Delta|X}(\delta|x)\leq F_{\Delta|X}^{e,U}(\delta|x).
\]
 \end{theorem}

One example of the identified sets of $F_{1|X}$ and $F_{0|X}$ is
Manski's bounds (\citet{manski1990nonparametric,manski1994selection})
defined as follows: 
\begin{equation}
\begin{aligned}F_{1|X}^{L}(y|x) & \equiv\Pr(Y\leq y|D=1,X=x)\Pr(D=1|X=x),\\
F_{1|X}^{U}(y|x) & \equiv\Pr(Y\leq y|D=1,X=x)\Pr(D=1|X=x)+\Pr(D=0|X=x),\\
F_{0|X}^{L}(y|x) & \equiv\Pr(Y\leq y|D=0,X=x)\Pr(D=0|X=x),\\
F_{0|X}^{U}(y|x) & \equiv\Pr(Y\leq y|D=0,X=x)\Pr(D=0|X=x)+\Pr(D=1|X=x).
\end{aligned}
\label{eq:manski_bound}
\end{equation}
 The bounds presented in Theorem \ref{thm:ID_Endo_FH_Bound} based
on Manski's bounds are often too wide to be informative. Some model
restrictions have been imposed to tighten bounds on the parameter
of interest in the literature (e.g., \citet{manski1997monotone,manski2000monotone,blundell2007changes}).
Motivated by these studies, we propose two stochastic dominance assumptions
that may help tighten the bounds and be consistent with many economic
theories in the absence of instrumental variables. The bounds on the
distribution of treatment effects in Theorem \ref{thm:ID_Endo_FH_Bound}
are based on the bounds on the conditional distributions of the potential
outcomes, and therefore, we can tighten the bounds on $F_{\Delta|X}$
once we have tighter bounds on the conditional distributions of the
potential outcomes than (\ref{eq:manski_bound}). 

Let $F_{d|d^{'}X}(y|x)\equiv\Pr\left(Y_{d}\leq y|D=d^{'},X=x\right)$
for given $d,d^{'}\in\{0,1\}$. The following assumption states that
the potential outcome conditional on $D=1$ first-order stochastically
dominates the potential outcome conditional on $D=0$. 

\begin{assumption}\label{assu:fsd1} Let $x\in\mathcal{X}$ be given.
For all $d\in\{0,1\}$ and $y\in\mathbb{R}$, $F_{d|1,X}(y|x)\leq F_{d|0,X}(y|x)$.
\end{assumption}

Assumption \ref{assu:fsd1} is identical to the stochastic dominance
assumption of \citet{blundell2007changes} that is used to tighten
the bounds on the wage distribution of the whole population in the
presence of sample selection. \citet{okumura2014concave} generalize
this stochastic dominance assumption to the case where $D$ is not
binary. Since Assumption \ref{assu:fsd1} implies that $E[Y_{d}|D=1,X=x]\geq E[Y_{d}|D=0,X=x]$,
this assumption is a sufficient condition for the monotone treatment
selection (MTS) assumption in \citet{manski2000monotone}.

Assumption \ref{assu:fsd1} is consistent with some economic theories
in empirical studies. As mentioned earlier, \citet{blundell2007changes}
utilize this assumption in conjunction with positive selection of
labor force participation from the standard labor supply model that
wage and probability of labor force participation are positively correlated.
As another example, we can consider sorting models for return to education,
where potential employees signal their ability to employers by using
their education level. Considering return to college education (i.e.,
$Y$ is wage, and $D$ is college entrance), it is likely that the
more capable people are, the more likely it is for them to complete
a college education. As a result, we may anticipate that people with
college degrees are more likely to have higher learning ability than
those who did not complete a college education (e.g., \citet{bedard2001human}).
Many studies in the labor economics literature consider learning ability
as an important factor affecting wage (\citet{lang1986human}), and
thus, we can assume that people with college degrees have higher wages
than those without degrees. As a consequence, the sorting hypothesis
is consistent with Assumption \ref{assu:fsd1}. 

The following theorem gives identification results for the conditional
distributions of the potential outcomes under Assumption \ref{assu:fsd1}. 

\begin{theorem}\label{thm:marginal_ID_FSD1} Let $x\in\mathcal{X}$
be fixed. Suppose that Assumption \ref{assu:fsd1} holds. For a given
$y\in\mathbb{R}$, define 
\begin{eqnarray*}
F_{1|X}^{L,FSD1}(y|x) & \equiv & \Pr(Y\leq y|D=1,X=x),\\
F_{1|X}^{U,FSD1}(y|x) & \equiv & \Pr(Y\leq y|D=1,X=x)\Pr(D=1|X=x)+\Pr(D=0|X=x),\\
F_{0|X}^{L,FSD1}(y|x) & \equiv & \Pr(Y\leq y|D=0,X=x)\Pr(D=0|X=x),\\
F_{0|X}^{U,FSD1}(y|x) & \equiv & \Pr(Y\leq y|D=0,X=x).
\end{eqnarray*}
Then, 
\begin{eqnarray}
F_{1|X}(y|x) & \in & \left[F_{1|X}^{L,FSD1}(y|x),F_{1|X}^{U,FSD1}(y|x)\right],\label{eq:fsd11}\\
F_{0|X}(y|x) & \in & \left[F_{0|X}^{L,FSD1}(y|x),F_{0|X}^{U,FSD1}(y|x)\right].\label{eq:fsd12}
\end{eqnarray}
 \end{theorem}

Comparing the bounds on the conditional distributions of the potential
outcomes in Theorem \ref{thm:marginal_ID_FSD1} with those in equation
(\ref{eq:manski_bound}), we can find that the upper bound on $F_{1|X}(y|x)$
and lower bound on $F_{0|X}(y|x)$ in Theorem \ref{thm:marginal_ID_FSD1}
are identical to their counterparts in (\ref{eq:manski_bound}). Since
Assumption \ref{assu:fsd1} designates only one direction of the monotonicity
of the distribution functions, it is impossible to improve the lower
bound on $F_{0|X}(y|x)$ and the upper bound on $F_{1|X}(y|x)$. Nevertheless,
the bounds on the conditional distributions of the potential outcomes
provided in Theorem \ref{thm:marginal_ID_FSD1} are narrower than
those in (\ref{eq:manski_bound}). 

We now introduce another stochastic dominance assumption. The following
stochastic dominance assumption may be regarded as an assumption corresponding
to the monotone treatment response (MTR) assumption in \citet{manski1997monotone}. 

\begin{assumption}\label{assu:fsd2} Let $x\in\mathcal{X}$ be given.
For all $d\in\{0,1\}$ and $y\in\mathbb{R}$, $F_{1|d,X}(y|x)\leq F_{0|d,X}(y|x)$.
\end{assumption}

Note that Assumption \ref{assu:fsd2} implies that $\mathbb{E}[Y_{1}|X=x]\geq\mathbb{E}[Y_{0}|X=x]$.
To see this, recall that $F_{1|X}(y|x)=F_{1|1,X}(y|x)p_{0}(x)+F_{1|0,X}(y|x)(1-p_{0}(x))$.
Under Assumption \ref{assu:fsd2}, we have $F_{1|1,X}(y|x)\leq F_{0|1,X}(y|x)$
and $F_{1|0,X}(y|x)\leq F_{0|0,X}(y|x)$, and therefore, we have $F_{1|X}(y|x)\leq F_{0|X}(y|x)$
and $\mathbb{E}[Y_{1}|X=x]\geq\mathbb{E}[Y_{0}|X=x]$. 

Assumption \ref{assu:fsd2} can be regarded as a distributional generalization
of, but is weaker than, the MTR assumption of \citet{manski1997monotone}.
This assumption is also applicable to many empirical studies as it
is compatible with some economic theories. Recalling the example of
return to college education, Assumption \ref{assu:fsd2} may be consistent
with the human capital theory in which education increases productivity
through a human capital production function and return to education
reflects the increased productivity (e.g., \citet{mincer1974schooling}).
For example, Assumption \ref{assu:fsd2} can be translated into that
people who have completed their college degrees are likely to be paid
higher wages than when they have not, and this argument can be justified
by the human capital theory in the labor economics. 

It is worth noting that Assumption \ref{assu:fsd2} is a special case
of conditional negative quadrant dependence, which is a dependence
concept considered by \citet{lehmann1966some}. \citet{kim2018partial}
proposes a concept of conditional quadrant dependence, as well as
a negative stochastic dominance, in the presence instrumental variables.
He provides tighter bounds on the distribution of treatment effects
than bounds using (\ref{eq:manski_bound}) under these assumptions
with the existence of an instrumental variable. A difference
between \citet{kim2018partial} and this paper is that this paper
does not impose additional assumptions on the data generating process,
such as monotonicity in structural functions for $Y_{1}$ and/or $Y_{0}$
and the availability of instrumental variables. Furthermore, this
paper focuses on estimation of the bounds on conditional distribution
of treatment effects, whereas \citet{kim2018partial} focuses on identification
of the unconditional distribution of treatment effects. 

The following theorem shows that we can tighten the bounds on the
conditional distributions of $Y_{1}$ and $Y_{0}$ given $X$ under
Assumption \ref{assu:fsd2}: 

\begin{theorem}\label{thm:marginal_ID_FSD2} Let $x\in\mathcal{X}$
be fixed. Suppose that Assumption \ref{assu:fsd2} holds. For a given
$y\in\mathbb{R}$, define 
\begin{eqnarray*}
F_{1|X}^{L,FSD2}(y|x) & = & \Pr(Y\leq y|D=1,X=x)\Pr(D=1|X=x),\\
F_{1|X}^{U,FSD2}(y|x) & = & \Pr(Y\leq y|X=x),\\
F_{0|X}^{L,FSD2}(y|x) & = & F_{1|X}^{U,FSD2}(y|x),\\
F_{0|X}^{U,FSD2}(y|x) & = & \Pr(Y\leq y|D=0,X=x)\Pr(D=0|X=x)+\Pr(D=1|X=x).
\end{eqnarray*}
Then, 
\begin{eqnarray}
F_{1|X}(y|x) & \in & \left[F_{1|X}^{L,FSD2}(y|x),F_{1|X}^{U,FSD2}(y|x)\right],\label{eq:fsd21}\\
F_{0|X}(y|x) & \in & \left[F_{0|X}^{L,FSD2}(y|x),F_{0|X}^{U,FSD2}(y|x)\right].\label{eq:fsd22}
\end{eqnarray}
\end{theorem}

As Assumption \ref{assu:fsd1} is not enough to improve $F_{1|X}^{U}(y|x)$
and $F_{0|X}^{L}(y|x)$ in (\ref{eq:manski_bound}), we can see that
the lower bound on $F_{1|X}(y|x)$ and the upper bound on $F_{0|X}(y|x)$
in Theorem \ref{thm:marginal_ID_FSD2} remain the same as those in
(\ref{eq:manski_bound}). 

If both Assumptions \ref{assu:fsd1} and \ref{assu:fsd2} hold, one can further improve the bounds on the conditional
distributions of the potential outcomes. This result is summarized
in the following corollary: 

\begin{corollary}\label{coro:marginal_ID_FSD12} Let $x\in\mathcal{X}$
be fixed. Suppose that Assumptions \ref{assu:fsd1} and \ref{assu:fsd2}
hold. For a given $y\in\mathbb{R}$, 
\begin{eqnarray*}
F_{1|X}(y|x) & \in & \left[F_{1|X}^{L,FSD1}(y|x),F_{1|X}^{U,FSD2}(y|x)\right],\\
F_{0|X}(y|x) & \in & \left[F_{0|X}^{L,FSD2}(y|x),F_{0|X}^{U,FSD1}(y|x)\right].
\end{eqnarray*}
\end{corollary}

It is straightforward to see that the identified sets for conditional
distribution functions of $Y_{1}$ and $Y_{0}$ on $X=x$ presented
in Corollary \ref{coro:marginal_ID_FSD12} are connected and that
the intersection of these sets is the boundary of each set, that is
$F_{1|X}^{U,FSD2}(y|x)=F_{0|X}^{L,FSD2}(y|x)$ for all $y\in\mathbb{R}$.

It is worth mentioning that the stochastic dominance relationships
in Assumptions \ref{assu:fsd1} and \ref{assu:fsd2} can be reversed,
depending on the empirical context. For example, one can impose a
condition that for all $d\in\{0,1\}$ and $y\in\mathbb{R}$, $F_{d|1,X}(y|x)\geq F_{d|0,X}(y|x)$
instead of Assumption \ref{assu:fsd1}. Similarly, one can consider
a stochastic dominance relationship that for all $d\in\{0,1\}$ and
$y\in\mathbb{R}$, $F_{1|d,X}(y|x)\geq F_{0|d,X}(y|x)$ instead of
Assumption \ref{assu:fsd2}. These stochastic dominance relationships
can be motivated by, for example, the empirical study of \citet{Angrist1990},
who considers the long-term effect of veteran status on earnings.
As pointed out by \citet{Angrist1990}, it would be likely that men
with fewer civilian opportunities served in the armed force, and this
may yield a negative selection in the sense that the potential earnings
of veterans are likely to be smaller than those of non-veterans (i.e.,
$F_{d|1,X}(y|x)\geq F_{d|0,X}(y|x)$). In addition, since military
service may hinder human capital accumulation, the potential earnings
they would have received when they served in the armed force are likely
to be smaller than the potential earnings they would have received
when they did not (i.e., $F_{1|d,X}(y|x)\geq F_{0|d,X}(y|x)$). 

There are other ways to tighten the bounds on the distribution of
treatment effects. For example, one can utilize support restrictions
(e.g., \citet{kim2018identifying}), consider a weaker condition than
the conditional independence, such as $c$-independence proposed by
\citet{masten2018identification}, or impose some dependence structure
as in \citet{frandsen2021partial}. In addition, although we do not
assume the availability of (monotone) instrumental variables in this
paper, one can consider the monotone instrumental variable (MIV) assumption
when there exists such a variable, as in \citet{manski2000monotone}
and \citet{kim2018partial}. A related but stronger assumption on
instrumental variables is the stochastically MIV (SMIV) assumption
considered by \citet{mourifie2020sharp}. The SMIV assumption imposes
a stochastic dominance ordering on the joint distribution of $Y_{0}$
and $Y_{1}$ across values of an instrumental variable, and the conditional
joint distribution of the potential outcomes given a value of the
instrumental variable, say $z$, (first-order) stochastically dominates
the conditional joint distribution given a value of the instrumental
variable that is smaller than $z$. These restrictions help tighten
the bounds on the conditional distributions of the potential outcomes
and/or the bounds on $F_{\Delta|X}$, but estimation and (uniform)
inference for such bounds is not trivial. Therefore, we leave it for
future study. 

\section{\label{sec:Estimation} Estimation of Bounds and Confidence Bands}

\subsection{Estimation and Bootstrap Procedure }

We consider estimation of the bounds and construction of confidence
bands for them in this section. To this end, we assume that the conditional
distributions of the potential outcomes are identified as follows:
for all $y\in\mathbb{R}$ and $x\in\mathcal{X}$, 
\begin{equation}
\begin{array}{cc}
F_{1|X}(y|x) & \in\left[LB_{1|X}(y|x),UB_{1|X}(y|x)\right],\\
F_{0|X}(y|x) & \in\left[LB_{0|X}(y|x),UB_{0|X}(y|x)\right].
\end{array}\label{eq:est_bound}
\end{equation}
 We denote nonparametric estimators of $LB_{1|X}$, $UB_{1|X}$, $LB_{0|X}$,
and $UB_{0|X}$ by $\hat{LB}_{1|X,n}$, $\hat{UB}_{1|X,n}$, $\hat{LB}_{0|X,n}$,
and $\hat{UB}_{0|X,n}$, respectively. In this paper, we consider
kernel-type estimators of the bounds. Let 
\begin{align*}
\mathbb{F}_{\mathbf{Y}|X}((y_{1l},y_{1u},y_{0l},y_{0u})|x) & \equiv\left(LB_{1|X}(y_{1l}|x),UB_{1|X}(y_{1u}|x),LB_{0|X}(y_{0l}|x),UB_{0|X}(y_{0u}|x)\right)^{t},\\
\hat{\mathbb{F}}_{\mathbf{Y}|X,n}((y_{1l},y_{1u},y_{0l},y_{0u})|x) & \equiv\left(\hat{LB}_{1|X,n}(y_{1l}|x),\hat{UB}_{1|X,n}(y_{1u}|x),\hat{LB}_{0|X,n}(y_{0l}|x),\hat{UB}_{0|X,n}(y_{0u}|x)\right)^{t}.
\end{align*}
 Suppose that there exists a positive real sequence $\left(r_{n}\right)_{n}$
such that $r_{n}\rightarrow\infty$ and that 
\[
r_{n}\left(\hat{\mathbb{F}}_{\mathbf{Y}|X,n}(\cdot|x)-\mathbb{F}_{\mathbf{Y}|X}(\cdot|x)\right)\Rightarrow\mathbb{G}\left(\cdot\right)\text{ in }\left(l^{\infty}(\mathcal{Y})\right)^{4},
\]
 where $\mathbb{G}(\cdot)$ is a four-dimensional centered Gaussian
process. This weak convergence result can be derived under a set of
mild regularity conditions when using standard nonparametric estimators.
When we use kernel-type estimators, we have $r_{n}=\sqrt{nh_{n}^{d_{x}}}$,
where $h_{n}$ is a bandwidth. 

Define 
\begin{align*}
\Pi_{L}\left(\mathbb{F}_{\mathbf{Y}|X}\right)(y,\delta|x) & \equiv LB_{1|X}(y|x)-UB_{0|X}(y-\delta|x),\\
\Pi_{U}\left(\mathbb{F}_{\mathbf{Y}|X}\right)(y,\delta|x) & \equiv UB_{1|X}(y|x)-LB_{0|X}(y-\delta|x),
\end{align*}
 and 
\begin{align}
\phi_{L}\left(\mathbb{F}_{\mathbf{Y}|X}\right) & \equiv\sup_{\delta\in\mathcal{D}}\Big|\sup_{y\in\mathcal{Y}}\Pi_{L}(\mathbb{F}_{\mathbf{Y}|X})(y,\delta|x)-F_{\Delta|X}^{L,0}(\delta|x)\Big|=\sup_{\delta}\Big|F_{\Delta|X}^{L}(\delta|x)-F_{\Delta|X}^{L,0}(\delta|x)\Big|,\label{eq:lower_functional}\\
\phi_{U}\left(\mathbb{F}_{\mathbf{Y}|X}\right) & \sup_{\delta\in\mathcal{D}}\Big|\inf_{y\in\mathcal{Y}}\{\Pi_{U}(\mathbb{F}_{\mathbf{Y}|X})(y,\delta|x)+1\}-F_{\Delta|X}^{U,0}(\delta|x)\Big|=\sup_{\delta}\Big|F_{\Delta|X}^{U}(\delta|x)-F_{\Delta|X}^{U,0}(\delta|x)\Big|,\label{eq:upper_functional}
\end{align}
 where $F_{\Delta|X}^{L,0}$ and $F_{\Delta|X}^{U,0}$ are some fixed
bounds that are true under the hypotheses $F_{\Delta|X}^{L}(\cdot|x)=F_{\Delta|X}^{L,0}(\cdot|x)$
and $F_{\Delta|X}^{U}(\cdot|x)=F_{\Delta|X}^{U,0}(\cdot|x)$. The
functionals in (\ref{eq:lower_functional}) and (\ref{eq:upper_functional})
can be considered as the Kolmogorov-Smirnov (KS) type statistics,
and they are useful to construct two-sided uniform confidence bands
for the lower and upper bounds on the conditional distribution of
treatment effects.\footnote{One can consider one-sided test statistics to construct one-sided
uniform confidence bands. However, we do not examine such test statistics
in this paper as the main purpose is to obtain uniform confidence
sets for the identified region and the two-sided uniform confidence
bands are enough to provide confidence sets for the identified region. } 

Let $\mathbb{H}\equiv(h_{1},h_{2},h_{3},h_{4})^{t}$ be a four-dimensional
vector-valued function and $f$ be a real-valued function. We denote
the Hadamard directional derivatives of $\phi_{L}$ and $\phi_{U}$
at $f$ in direction $\mathbb{H}$ by $\phi_{L}^{'}(f;\mathbb{H})$
and $\phi_{U}^{'}(f;\mathbb{H})$, respectively. We also let $\widehat{\phi_{L}^{'}}\left(f;\mathbb{H},a_{n}\right)$
and $\widehat{\phi_{U}^{'}}\left(f;\mathbb{H},a_{n}\right)$ denote
estimators of $\phi_{L}^{'}(f;\mathbb{H})$ and $\phi_{U}^{'}(f;\mathbb{H})$,
respectively, where $(a_{n})$ is a positive real sequence decreasing
to zero. We provide the forms of the Hadamard directional derivatives
and their estimators in Appendix \ref{sec:HD-Derivatives}. The main
goal is to establish the limiting distributions of $r_{n}\left(\phi_{L}\left(\hat{\mathbb{F}}_{\mathbf{Y}|X,n}\right)-\phi_{L}\left(\mathbb{F}_{\mathbf{Y}|X}\right)\right)$
and $r_{n}\left(\phi_{U}\left(\hat{\mathbb{F}}_{\mathbf{Y}|X,n}\right)-\phi_{U}\left(\mathbb{F}_{\mathbf{Y}|X}\right)\right)$.
The limiting distributions allow one to construct confidence bands
for the bounds on the conditional distribution of treatment effects. 

As will be shown below, the limiting distributions are nonstandard.
We propose to use a multiplier bootstrap to mimic the asymptotic distributions
of $r_{n}\left(\phi_{L}\left(\hat{\mathbb{F}}_{\mathbf{Y}|X,n}\right)-\phi_{L}\left(\mathbb{F}_{\mathbf{Y}|X}\right)\right)$
and $r_{n}\left(\phi_{U}\left(\hat{\mathbb{F}}_{\mathbf{Y}|X,n}\right)-\phi_{U}\left(\mathbb{F}_{\mathbf{Y}|X}\right)\right)$.
Specifically, let 
\[
\hat{\Psi}_{i}((y_{1l},y_{1u},y_{0l},y_{0u})|x)\equiv\left(\hat{\psi}_{1l,i|X,n}(y_{1l}|x),\hat{\psi}_{1u,i|X,n}(y_{1u}|x),\hat{\psi}_{0l,i|X,n}(y_{0l}|x),\hat{\psi}_{0u,i|X,n}(y_{0u}|x)\right)^{t}
\]
 be an estimated influence function for the $i$-th observation of
$\hat{\mathbb{F}}_{\mathbf{Y}|X,n}((y_{1l},y_{1u},y_{0l},y_{0u})|x)$
and $\{B_{i}:i=1,2,...n\}$ be $n$ random draws from $B$ that is
independent of the data and satisfies some moment conditions (cf.
Assumption \ref{assu:boot_weight}). Define 
\begin{align*}
 & \hat{\mathbb{F}}_{\mathbf{Y}|X,n}^{*}((y_{1l},y_{1u},y_{0l},y_{0u})|x)\\
\equiv & \left(\sum_{i}B_{i}\hat{\psi}_{1l,i|X,n}(y_{1l}|x),\sum_{i}B_{i}\hat{\psi}_{1u,i|X,n}(y_{1u}|x),\sum_{i}B_{i}\hat{\psi}_{0l,i|X,n}(y_{0l}|x),\sum_{i}B_{i}\hat{\psi}_{0u,i|X,n}(y_{0u}|x)\right)^{t}.
\end{align*}

Then, one can implement the following bootstrap procedure to obtain
confidence bands for $F_{\Delta|X}^{L}(\cdot|x)$ and $F_{\Delta|X}^{U}(\cdot|x)$. 

\paragraph{Bootstrap procedure }
\begin{enumerate}
\item Estimate the conditional distribution functions of the potential outcomes
or bounds on them, and construct 
\begin{align*}
\Pi_{L}(\hat{\mathbb{F}}_{\mathbf{Y}|X,n})(y,\delta|x) & =\hat{LB}_{1|X,n}(y|x)-\hat{UB}_{0|X,n}(y-\delta|x),\\
\Pi_{U}(\hat{\mathbb{F}}_{\mathbf{Y}|X,n})(y,\delta|x) & =\hat{UB}_{1|X,n}(y|x)-\hat{LB}_{0|X,n}(y-\delta|x).
\end{align*}
 
\item Repeat the following procedure for $m$ times, where $m$ denotes
the number of bootstrap iterations. 
\begin{enumerate}
\item Generate the bootstrap weights $\{B_{i}:i=1,2,...n\}$ from the random
variable $B$ in Assumption \ref{assu:boot_weight}. 
\item Using the estimated conditional distribution functions of the potential
outcomes or bounds on them, compute $r_{n}\hat{\mathbb{F}}_{\mathbf{Y}|X,n}^{*}(\cdot|x)$. 
\item For each iteration $b=1,2,...,m$, compute 
\begin{align*}
 & \widehat{\phi_{L}^{'}}^{(b)}\left(\Pi_{L}(\hat{\mathbb{F}}_{\mathbf{Y}|X,n})(y,\delta|x);r_{n}\hat{\mathbb{F}}_{\mathbf{Y}|X,n}^{*}(\cdot|x),a_{n}\right),\\
 & \widehat{\phi_{U}^{'}}^{(b)}\left(\Pi_{U}(\hat{\mathbb{F}}_{\mathbf{Y}|X,n})(y,\delta|x)+1;r_{n}\hat{\mathbb{F}}_{\mathbf{Y}|X,n}^{*}(\cdot|x),a_{n}\right).
\end{align*}
\end{enumerate}
\item For a given significance level $\alpha\in(0,1)$, find the $(1-\frac{\alpha}{2})$-th
quantiles of 
\begin{align*}
 & \left\{ \widehat{\phi_{L}^{'}}^{(b)}\left(\Pi_{L}(\hat{\mathbb{F}}_{\mathbf{Y}|X,n})(y,\delta|x);r_{n}\hat{\mathbb{F}}_{\mathbf{Y}|X,n}^{*}(\cdot|x),a_{n}\right):b=1,2,...,m\right\} ,\\
 & \left\{ \widehat{\phi_{U}^{'}}^{(b)}\left(\Pi_{U}(\hat{\mathbb{F}}_{\mathbf{Y}|X,n})(y,\delta|x)+1;r_{n}\hat{\mathbb{F}}_{\mathbf{Y}|X,n}^{*}(\cdot|x),a_{n}\right):b=1,2,...m\right\} .
\end{align*}
We denote the quantiles by $c_{1-\frac{\alpha}{2}}^{L}$ and $c_{1-\frac{\alpha}{2}}^{U}$,
respectively. 
\item Let 
\begin{align*}
\tilde{CI}\left(F_{\Delta|X}^{L}(\delta;x),1-\alpha\right) & \equiv\left[\hat{F}_{\Delta|X,n}^{L}(\delta|x)-c_{1-\frac{\alpha}{2}}^{L},\hat{F}_{\Delta|X,n}^{L}(\delta|x)+c_{1-\frac{\alpha}{2}}^{L}\right],\\
\tilde{CI}\left(F_{\Delta|X}^{U}(\delta;x),1-\alpha\right) & \equiv\left[\hat{F}_{\Delta|X,n}^{U}(\delta|x)-c_{1-\frac{\alpha}{2}}^{U},\hat{F}_{\Delta|X,n}^{U}(\delta|x)+c_{1-\frac{\alpha}{2}}^{U}\right].
\end{align*}
 The $(1-\alpha)\times100\%$ confidence bands of $F_{\Delta|X}^{L}(\cdot|x)$
and $F_{\Delta|X}^{U}(\cdot|x)$ can be constructed as 
\begin{align*}
CI\left(F_{\Delta|X}^{L}(\delta;x),1-\alpha\right) & \equiv\left[\max\left\{ \hat{F}_{\Delta|X,n}^{L}(\delta|x)-c_{1-\frac{\alpha}{2}}^{L},0\right\} ,\min\left\{ \hat{F}_{\Delta|X,n}^{L}(\delta|x)+c_{1-\frac{\alpha}{2}}^{L},1\right\} \right],\\
CI\left(F_{\Delta|X}^{U}(\delta;x),1-\alpha\right) & \equiv\left[\max\left\{ \hat{F}_{\Delta|X,n}^{U}(\delta|x)-c_{1-\frac{\alpha}{2}}^{U},0\right\} ,\min\left\{ \hat{F}_{\Delta|X,n}^{U}(\delta|x)+c_{1-\frac{\alpha}{2}}^{U},1\right\} \right],
\end{align*}
 respectively. 
\end{enumerate}
The difference between $\tilde{CI}$ and $CI$ in the above bootstrap
procedure is that $CI$'s are confidence bands obtained to impose
the logical bounds when necessary. \citet{chen2021shape} show that
one advantage of applying such operators to the original confidence
bands is that the resulting confidence bands have greater coverage
than the original confidence bands. In addition, this procedure is
very simple and easy to implement in practice. 

The above bootstrap procedure allows us to construct the pointwise
confidence set of the identified set, where the identified set is
$[F_{\Delta|X}^{L}(\delta|x),F_{\Delta|X}^{U}(\delta|x)]$ for given
$\delta\in Supp(\Delta|X=x)$. Specifically, we can show that 
\[
\underset{n\rightarrow\infty}{\lim\inf}\Pr\left(F_{\Delta|X}(\delta|x)\in\left[\hat{F}_{\Delta|X,n}^{L}(\delta|x)-c_{1-\frac{\alpha}{2}}^{L},\hat{F}_{\Delta|X,n}^{U}(\delta|x)+c_{1-\frac{\alpha}{2}}^{U}\right]\right)\geq1-\alpha,
\]
 and this can be used as a $(1-\alpha)\times100\%$ confidence set
for the identified set. This confidence set does not require the uniqueness
of the infimum and supremum in the Makarov bounds. As pointed out
by \citet{firpo2021uniform}, however, this confidence set is likely
to be conservative. 

\subsection{Nonparametric Estimators of the Bounds on Conditional Distributions
of Potential Outcomes }

We propose to use kernel-type estimators of bounds on the conditional
distributions of the potential outcomes in (\ref{eq:est_bound}).
Let $K(\cdot):\mathbb{R}^{d_{x}}\rightarrow\mathbb{R}$ be a $d_{x}$-dimensional
kernel function and $h_{n}$ be a bandwidth such that $h_{n}\rightarrow0$,
$nh_{n}^{d_{x}}\rightarrow\infty$, and $nh_{n}^{d_{x}+4}\rightarrow0$
as $n\rightarrow\infty$. We start with the case where Assumptions
\ref{assu:unconfounded} and \ref{assu:bound_propensity} hold. In
this case, the conditional distributions of the potential outcomes
are point identified (i.e., $LB_{j|X}(y|x)=UB_{j|X}(y|x)=F_{j|X}(y|x)$
for each $j\in\{0,1\}$), and we can use 
\begin{align*}
\hat{F}_{1|X,n}(y_{1}|x) & \equiv\frac{\sum_{i}^{n}\mathbf{1}(Y_{i}\leq y_{1})D_{i}K(\frac{X_{i}-x}{h_{n}})}{\sum_{i}^{n}D_{i}K(\frac{X_{i}-x}{h_{n}})},\\
\hat{F}_{0|X,n}(y_{0}|x) & \equiv\frac{\sum_{i}^{n}\mathbf{1}(Y_{i}\leq y_{0})(1-D_{i})K(\frac{X_{i}-x}{h_{n}})}{\sum_{i}^{n}(1-D_{i})K(\frac{X_{i}-x}{h_{n}})}
\end{align*}
 as estimators of $F_{1|X}(y_{1}|x)$ and $F_{0|X}(y_{0}|x)$, respectively.
We then define 
\begin{equation}
\hat{\mathbb{F}}_{\mathbf{Y}|X,n}((y_{1l},y_{1u},y_{0l},y_{0u})|x)\equiv\left(\hat{F}_{1|X,n}(y_{1l}|x),\hat{F}_{1|X,n}(y_{1u}|x),\hat{F}_{0|X,n}(y_{0l}|x),\hat{F}_{0|X,n}(y_{0u}|x)\right)^{t}.\label{eq:kernel_unconfounded}
\end{equation}
 Note that the kernel function and bandwidth used to estimate $F_{1|X}$
can be different from those used to estimate $F_{0|X}$. In addition,
we have $\Pi_{L}(\mathbb{F}_{\mathbf{Y}|X,n})(y,\delta|x)=\Pi_{U}(\mathbb{F}_{\mathbf{Y}|X,n})(y,\delta|x)=F_{1|X}(y|x)-F_{0|X}(y-\delta|x)$
and can use the bootstrap procedure with $r_{n}=\sqrt{nh_{n}^{d_{x}}}$
under a set of regularity conditions. 

We now consider the case where the treatment is endogenous and Assumptions \ref{assu:fsd1} and \ref{assu:fsd2} hold. Recall that
$F_{d|d^{'}X}(y|x)=\Pr(Y_{d}\leq y|D=d^{'},X=x)$ for given $d,d^{'}\in\{0,1\}$,
and the conditional distribution of $Y$ given $X=x$ is denoted by
$F_{Y|X}(\cdot|x)$. By corollary \ref{coro:marginal_ID_FSD12}, we
can construct kernel estimators of these objects as 
\begin{align*}
\hat{F}_{1|1X,n}(y_{1}|x) & \equiv\frac{\sum_{i}^{n}\mathbf{1}(Y_{i}\leq y_{1})D_{i}K(\frac{X_{i}-x}{h_{n}})}{\sum_{i}^{n}D_{i}K(\frac{X_{i}-x}{h_{n}})},\\
\hat{F}_{0|0X,n}(y_{0}|x) & \equiv\frac{\sum_{i}^{n}\mathbf{1}(Y_{i}\leq y_{0})(1-D_{i})K(\frac{X_{i}-x}{h_{n}})}{\sum_{i}^{n}(1-D_{i})K(\frac{X_{i}-x}{h_{n}})},\\
\hat{F}_{Y|X,n}(y|x) & \equiv\frac{\sum_{i}^{n}\mathbf{1}(Y_{i}\leq y)K(\frac{X_{i}-x}{h_{n}})}{\sum_{i}^{n}K(\frac{X_{i}-x}{h_{n}})},
\end{align*}
 and define 
\begin{equation}
\hat{\mathbb{F}}_{\mathbf{Y}|X,n}((y_{1l},y_{1u},y_{0l},y_{0u})|x)\equiv\left(\hat{F}_{1|1X,n}(y_{1l}|x),\hat{F}_{Y|X,n}(y_{1u}|x),\hat{F}_{Y|X,n}(y_{0l}|x),\hat{F}_{0|0X,n}(y_{0u}|x)\right)^{t}.\label{eq:kernel_endo}
\end{equation}
  Then, we have 
\begin{align*}
\Pi_{L}(\hat{\mathbb{F}}_{\mathbf{Y}|X,n})(y,\delta|x) & =\hat{F}_{1|1X,n}(y|x)-\hat{F}_{0|0X,n}(y-\delta|x),\\
\Pi_{U}(\hat{\mathbb{F}}_{\mathbf{Y}|X,n})(y,\delta|x) & =\hat{F}_{Y|X,n}(y|x)-\hat{F}_{Y|X,n}(y-\delta|x).
\end{align*}
We provide a description on how to estimate the influence functions
in Appendix \ref{sec:Influence}. 

\section{\label{sec:Inference} Asymptotic Theory }

\subsection{Inference under the Unconfoundedness Assumption }

We first develop the asymptotic theory for the kernel estimators in
(\ref{eq:kernel_unconfounded}). 

\begin{assumption}\label{assu:iid} $\{W_{i}\equiv(Y_{i},D_{i},X_{i}^{t})^{t}:i=1,2,...,n\}$
is a random sample. \end{assumption} 

\begin{assumption}\label{assu:dist_X} (i) The support of $X$, $\mathcal{X}$,
is a compact subset of $\mathbb{R}^{d_{x}}$; (ii) the distribution
of $X$ admits its density $f_{X}(\cdot)$ on $\mathcal{X}$ such
that $0<\inf_{x\in\mathcal{X}}f_{X}(x)<\sup_{x\in\mathcal{X}}f_{X}(x)<\infty$.
The density function $f_{X}(\cdot)$ is twice continuously differentiable
and $\sup_{x\in\mathcal{X}}|f_{X}^{(1)}(x)|$ and $\sup_{x\in\mathcal{X}}|f_{X}^{(2)}(x)|$
are bounded. \end{assumption}

\begin{assumption}\label{assu:smoothness} (i) The propensity score
function $p_{0}(x)$ is twice continuously differentiable and all
derivatives are uniformly bounded over $\mathcal{X}$; (ii) For each
$j\in\{0,1\}$, $F_{j|X}(y|x)$ is twice continuously differentiable
with respect to $x$ and all derivatives are uniformly bounded. \end{assumption}

\begin{assumption}\label{assu:kernel} The kernel function $K(\cdot)$
is a product of a univariate bounded kernel function $k(\cdot):\mathbb{R}\rightarrow\mathbb{R}_{+}$
such that $\int k(u)du=1$, $\int uk(u)du=0$, and $\int u^{2}k(u)du\equiv k_{2}<\infty$.
The support of the univariate kernel function is compact. \end{assumption}

\begin{assumption}\label{assu:bandwidth} The bandwidth $h_{n}$
satisfies the following conditions: (i) $h_{n}\rightarrow0$; (ii)
$nh_{n}^{d_{x}}\rightarrow\infty$; (iii) $nh_{n}^{d_{x}+4}\rightarrow0$
as $n\rightarrow\infty$. \end{assumption} 

Assumption \ref{assu:iid} is an i.i.d assumption on the sample. This
can be relaxed at a cost of more complicated proofs of the theoretical
results. Assumption \ref{assu:dist_X} imposes some degree of smoothness
on the distribution of $X$. This assumption is standard in the literature
on kernel estimation. Assumption \ref{assu:smoothness} requires that
the conditional distribution functions of $Y_{1}$ and $Y_{0}$ and
the propensity score functions be smooth enough. This condition, together
with Assumptions \ref{assu:kernel} and \ref{assu:bandwidth}, allows
to eliminate the bias of estimators of conditional distribution functions
of $Y_{1}$ and $Y_{0}$. Assumption \ref{assu:kernel} is also standard
in the literature, and there are several kernel functions that satisfy
this assumption (e.g., Epanechnikov, biweight, triweight kernels).
Assumption \ref{assu:bandwidth} restricts the rate of bandwidth.
The last condition $nh_{n}^{d_{x}+4}\rightarrow0$ is required to
eliminate the asymptotic bias of the kernel estimators. Note that
when the dimension of $X$ is large, one can use a higher-order kernel
to handle the bias term of kernel estimators. 

We first establish the weak convergence of the kernel estimators under
the unconfoundedness assumption. For given $x\in\mathcal{X}$, define
$\mathbb{F}_{\mathbf{Y}|X}(\cdot|x)\equiv\left(F_{1|X}(\cdot|x),F_{1|X}(\cdot|x),F_{0|X}(\cdot|x),F_{0|X}(\cdot|x)\right)^{t}$
and $\hat{\mathbb{F}}_{\mathbf{Y}|X,n}(\cdot|x)\equiv\left(\hat{F}_{1|X,n}(\cdot|x),\hat{F}_{1|X,n}(\cdot|x),\hat{F}_{0|X,n}(\cdot|x),\hat{F}_{0|X,n}(\cdot|x)\right)^{t}$. 

\begin{theorem}\label{thm:asymp_marginals} Suppose that Assumptions
\ref{assu:unconfounded} and \ref{assu:bound_propensity} hold. Let
$x\in int(\mathcal{X})$ be given. If Assumptions \ref{assu:iid}--\ref{assu:bandwidth}
hold, then, 
\[
\sqrt{nh_{n}^{d_{x}}}\left(\hat{\mathbb{F}}_{\mathbf{Y}|X,n}(\cdot|x)-\mathbb{F}_{\mathbf{Y}|X}(\cdot|x)\right)\Rightarrow\mathbb{G}_{x}\left(\cdot\right)\equiv\left(\mathbb{G}_{1,x}(\cdot),\mathbb{G}_{1,x}(\cdot),\mathbb{G}_{0,x}(\cdot),\mathbb{G}_{0,x}(\cdot)\right)^{t}\text{ in }\left(l^{\infty}(\mathcal{Y})\right)^{4},
\]
 where $\mathbb{G}_{1,x}(\cdot)$ and $\mathbb{G}_{0,x}(\cdot)$ are
mean zero Gaussian processes with covariance kernels $H_{1,x}(s,t)$
and $H_{0,x}(s,t)$, respectively, whose the forms are given in Appendix
\ref{sec:Proofs}, and $\mathbb{G}_{x}((y_{1l},y_{1u},y_{0l},y_{0u}))=\left(\mathbb{G}_{1,x}(y_{1l}),\mathbb{G}_{1,x}(y_{1u}),\mathbb{G}_{0,x}(y_{0l}),\mathbb{G}_{0,x}(y_{0u})\right)^{t}$.
\end{theorem} 

Theorem \ref{thm:asymp_marginals} is useful not only for deriving
the asymptotic distribution of the estimated bounds on $F_{\Delta|X}(\cdot|x)$,
but also for conducting uniform inference for some policy-relevant
parameter (e.g., quantile treatment effects). 

We now establish the asymptotic theory for confidence bands for the
identified set of $F_{\Delta|X}(\cdot|x)$. To this end, we consider
the following two hypotheses: $F_{\Delta|X}^{L}(\delta|x)=F_{\Delta|X}^{L,0}(\delta|x)$
and $F_{\Delta|X}^{U}(\delta|x)=F_{\Delta|X}^{U,0}(\delta|x)$, where
$F_{\Delta|X}^{L,0}(\delta|x)$ and $F_{\Delta|X}^{U,0}(\delta|x)$
are some fixed lower and upper bounds on the conditional distribution
of treatment effects, respectively. 

The main challenge to constructing uniform confidence bands, however,
is that the mappings $\phi_{L}$ and $\phi_{U}$ are not Hadamard
differentiable. To overcome this difficulty, we use Theorem 3.2 in
\citet{firpo2021uniform} that shows that these functionals are Hadamard
directionally differentiable. Based on this result, we employ the
inference method of \citet{fang2019inference} to establish the asymptotic
theory. 

Recall that $\Pi_{L}(\mathbb{F}_{\mathbf{Y}|X})(y,\delta|x)=\Pi_{U}(\mathbb{F}_{\mathbf{Y}|X})(y,\delta|x)=F_{1|X}(y|x)-F_{0|X}(y-\delta|x)$.
The following theorem establishes the limiting distributions of $\phi_{L}\left(\hat{\mathbb{F}}_{\mathbf{Y}|X,n}(\cdot|x)\right)$
and $\phi_{U}\left(\hat{\mathbb{F}}_{\mathbf{Y}|X,n}(\cdot|x)\right)$. 

\begin{theorem}\label{thm:Asymptotic_Dist of Derivative} Suppose
that Assumptions \ref{assu:unconfounded} and \ref{assu:bound_propensity}
hold. Let $x\in int(\mathcal{X})$ be given. If Assumptions \ref{assu:iid}--\ref{assu:bandwidth}
hold, then, 
\begin{equation}
\begin{aligned}\sqrt{nh_{n}^{d_{x}}}\left(\phi_{L}\left(\hat{\mathbb{F}}_{\mathbf{Y}|X,n}(\cdot|x)\right)-\phi_{L}\left(\mathbb{F}_{\mathbf{Y}|X}(\cdot|x)\right)\right) & \Rightarrow\phi_{L}^{'}\left(\Pi_{L}(\mathbb{F}_{\mathbf{Y}|X})(y,\delta|x);\mathbb{G}_{x}(\cdot)\right)\ \text{in }l^{\infty}\left(\mathcal{Y}^{4}\right),\\
\sqrt{nh_{n}^{d_{x}}}\left(\phi_{U}\left(\hat{\mathbb{F}}_{\mathbf{Y}|X,n}(\cdot|x)\right)-\phi_{U}\left(\mathbb{F}_{\mathbf{Y}|X}(\cdot|x)\right)\right) & \Rightarrow\phi_{U}^{'}\left(\Pi_{U}(\mathbb{F}_{\mathbf{Y}|X})(y,\delta|x)+1;\mathbb{G}_{x}(\cdot)\right)\ \text{in }l^{\infty}\left(\mathcal{Y}^{4}\right),
\end{aligned}
\label{eq:te_asymptotic}
\end{equation}
 where and $\mathbb{G}_{x}(\cdot)$ is the Gaussian process defined
in Theorem \ref{thm:asymp_marginals}. The forms of $\phi_{L}^{'}$
and $\phi_{U}^{'}$ are provided in Appendix \ref{sec:HD-Derivatives}.
\end{theorem}

It is worth pointing out that the development of weak convergence
of the estimated bounds does not rely on some pointwise asymptotic
theory when establishing the asymptotic theory for the estimated Makarov
bounds. Instead, we first derive the weak convergence of the standard
kernel estimators for a fixed conditioning value and consider the
double-supremum as an operator to establish the weak convergence of
the estimated bounds, as in \citet{firpo2021uniform}. In doing so,
we can avoid imposing the uniqueness of $argsup$ and $arginf$ that
is needed for pointwise inference. 

Theorem \ref{thm:Asymptotic_Dist of Derivative} is a direct consequence
of Theorem 2.1 of \citet{fang2019inference}. The limiting distribution
presented in Theorem \ref{thm:Asymptotic_Dist of Derivative} is non-standard,
and therefore we need to rely on some resampling method to mimic the
limiting distribution and conduct inference. Since the functionals
$\phi_{L}$ and $\phi_{U}$ are not Hadamard differentiable but only
Hadamard directionally differentiable, the standard bootstrap fails
(cf. Theorem 3.1 of \citet{fang2019inference}). To resolve this problem,
we utilize the result of bootstrap validity that was proposed by \citet{firpo2021uniform}.
The next assumption imposes conditions on the bootstrap weight $B$: 

\begin{assumption}\label{assu:boot_weight} Let $B$ be a random
variable that is independent of the data $\mathcal{W}$ such that
$\mathbb{E}[B]=0$, $Var(B)=1$, and $\int_{0}^{\infty}\sqrt{\Pr(|B|>x)}dx<\infty$.
\end{assumption} 

The last condition in Assumption \ref{assu:boot_weight} is satisfied
if $\mathbb{E}\left[|B|^{2+\epsilon}\right]<\infty$ for some $\epsilon>0$.
One can use a standard normal random variable as the bootstrap weight
$B$. 

We employ the approach of \citet{fang2019inference} to approximate
the limiting distribution presented in Theorem \ref{thm:Asymptotic_Dist of Derivative},
which was also considered by \citet{firpo2021uniform}. Note that
it is required to consistently estimate the Hadamard directional derivatives
in (\ref{eq:phi_derivative}) and that the Hadamard directional derivatives
are defined in terms of the limit operator. To this end, we consider
a tuning sequence $(a_{n})$ that satisfies the conditions in the
following assumption: 

\begin{assumption}\label{assu:est_der_tuning}Let $(a_{n})$ be a
sequence of positive real numbers such that $a_{n}\downarrow0$ and
$a_{n}\sqrt{nh_{n}^{d_{x}}}\rightarrow\infty$.\end{assumption} 

Since the limiting distributions in Theorem \ref{thm:Asymptotic_Dist of Derivative}
are nonstandard, we use a bootstrap to mimic the limiting distributions.
The next theorem demonstrates that one can use the bootstrap procedure
in Section \ref{sec:Estimation} to approximate the asymptotic distributions
of $\phi_{L}^{'}\left(\Pi_{L}(\mathbb{F}_{\mathbf{Y}|X})(y,\delta|x);\mathbb{G}_{x}(\cdot)\right)$
and $\phi_{U}^{'}\left(\Pi_{U}(\mathbb{F}_{\mathbf{Y}|X})(y,\delta|x)+1;\mathbb{G}_{x}(\cdot)\right)$.
Let $\mathbb{\hat{F}}_{\mathbf{Y}|X,n}^{*}(\cdot|x)$ denote the vector
of simulated stochastic processes using estimated influence functions
in (\ref{eq:influence_exo}) in Appendix \ref{sec:Influence}. 

\begin{theorem}\label{thm:bootstrap} Suppose that Assumptions \ref{assu:unconfounded}
and \ref{assu:bound_propensity} hold. Let $x\in int(\mathcal{X})$
be given and Assumptions \ref{assu:iid}--\ref{assu:bandwidth},
\ref{assu:boot_weight}, and \ref{assu:est_der_tuning} hold. Then,
we have 
\begin{align*}
\widehat{\phi_{L}^{'}}\left(\Pi_{L}(\hat{\mathbb{F}}_{\mathbf{Y}|X,n})(y,\delta|x);\sqrt{nh_{n}^{d_{x}}}\mathbb{\hat{F}}_{\mathbf{Y}|X,n}^{*}(\cdot|x),a_{n}\right) & \Rightarrow\phi_{L}^{'}\left(\Pi_{L}(\mathbb{F}_{\mathbf{Y}|X})(y,\delta|x);\mathbb{G}_{x}(\cdot)\right),\\
\widehat{\phi_{U}^{'}}\left(\Pi_{U}(\hat{\mathbb{F}}_{\mathbf{Y}|X,n})(y,\delta|x)+1;\sqrt{nh_{n}^{d_{x}}}\mathbb{\hat{F}}_{\mathbf{Y}|X,n}^{*}(\cdot|x),a_{n}\right) & \Rightarrow\phi_{U}^{'}\left(\Pi_{U}(\mathbb{F}_{\mathbf{Y}|X})(y,\delta|x)+1;\mathbb{G}_{x}(\cdot)\right),
\end{align*}
 in $l^{\infty}\left(\mathcal{Y}^{4}\right)$, conditional on data.
The forms of $\widehat{\phi_{L}^{'}}$ and $\widehat{\phi_{U}^{'}}$
are provided in Appendix \ref{sec:HD-Derivatives}. \end{theorem} 

\subsection{Inference with an Endogenous Treatment}

We now develop the asymptotic theory when the treatment is endogenous.
Our asymptotic theory for an endogenous treatment focuses on the situation
where Assumptions \ref{assu:fsd1} and \ref{assu:fsd2} hold, and
thus, one can use the kernel estimators in Section \ref{sec:Estimation}.
To be concrete, define, for given $x\in\mathcal{X}$, $\mathbb{F}_{\mathbf{Y}|X}(\cdot|x)\equiv\left(F_{1|1X}(\cdot|x),F_{Y|X}(\cdot|x),F_{Y|X}(\cdot|x),F_{0|0X}(\cdot|x)\right)^{t}$
and $\hat{\mathbb{F}}_{\mathbf{Y}|X,n}(\cdot|x)\equiv\left(\hat{F}_{1|1X,n}(\cdot|x),\hat{F}_{Y|X,n}(\cdot|x),\hat{F}_{Y|X,n}(\cdot|x),\hat{F}_{0|0X,n}(\cdot|x)\right)^{t}$,
where each component of $\hat{\mathbb{F}}_{\mathbf{Y}|X,n}(\cdot|x)$
is given in (\ref{eq:kernel_endo}). Let $\Pi_{L}(\mathbb{F}_{\mathbf{Y}|X,n})(y,\delta|x)\equiv F_{1|1X}(y|x)-F_{0|0X}(y-\delta|x)$
and $\Pi_{U}(\mathbb{F}_{\mathbf{Y}|X,n})(y,\delta|x)\equiv F_{Y|X}(y|x)-F_{Y|X}(y-\delta|x)$. 

\begin{assumption}\label{assu:smooth_endo} $F_{1|1,X}(y|x)$, $F_{0|0,X}(y|x)$,
and $F_{Y|X}(y|x)$ are twice continuously differentiable with respect
to $x$ and all derivatives are uniformly bounded. \end{assumption}

\begin{theorem}\label{thm:endo_marginal_weak_limit} Suppose that
Assumptions \ref{assu:fsd1} and \ref{assu:fsd2} hold. Let $x\in int(\mathcal{X})$
be given and Assumptions \ref{assu:iid}, \ref{assu:dist_X}, \ref{assu:kernel},
\ref{assu:bandwidth}, and \ref{assu:smooth_endo} hold. Then, 
\[
\sqrt{nh_{n}^{d_{x}}}\left(\hat{\mathbb{F}}_{\mathbf{Y}|X,n}(\cdot|x)-\mathbb{F}_{\mathbf{Y}|X}(\cdot|x)\right)\Rightarrow\mathbb{G}_{x}^{e}(\cdot)\ \text{in }\left(l^{\infty}(\mathcal{Y})\right)^{4},
\]
 where $\mathbb{G}_{x}^{e}((y_{1l},y_{1u},y_{0l},y_{0u}))\equiv(\mathbb{G}_{1,x}^{e}(y_{1l}),\mathbb{G}_{Y,x}^{e}(y_{1u}),\mathbb{G}_{Y,x}^{e}(y_{0l}),\mathbb{G}_{0,x}^{e}(y_{0u}))^{t}$
and $\mathbb{G}_{1,x}^{e}$, $\mathbb{G}_{0,x}^{e}$, and $\mathbb{G}_{Y,x}^{e}$
are Gaussian processes with mean zero and covariance kernels being
$H_{1,x}^{e}(\cdot,\cdot)$, $H_{0,x}^{e}(\cdot,\cdot)$, and $H_{Y,x}^{e}(\cdot,\cdot)$,
respectively. The forms of the covariance kernels are given in Appendix
\ref{sec:Proofs}. 

In addition, 
\begin{equation}
\begin{aligned}\sqrt{nh_{n}^{d_{x}}}\left(\phi_{L}\left(\hat{\mathbb{F}}_{\mathbf{Y}|X,n}(\cdot|x)\right)-\phi_{L}\left(\mathbb{F}_{\mathbf{Y}|X}(\cdot|x)\right)\right) & \Rightarrow\phi_{L}^{'}\left(\Pi_{L}(\mathbb{F}_{\mathbf{Y}|X})(y,\delta|x);\mathbb{G}_{x}^{e}(\cdot)\right),\\
\sqrt{nh_{n}^{d_{x}}}\left(\phi_{U}(\hat{\mathbb{F}}_{\mathbf{Y}|X,n}(\cdot|x))-\phi_{U}^{e}(\mathbb{F}_{\mathbf{Y}|X}(\cdot|x))\right) & \Rightarrow\phi_{U}^{'}\left(\Pi_{U}(\mathbb{F}_{\mathbf{Y}|X})(y,\delta|x)+1;\mathbb{G}_{x}^{e}(\cdot)\right),
\end{aligned}
\label{eq:te_asymptotic_endo}
\end{equation}
 in $l^{\infty}\left(\mathcal{Y}^{4}\right)$. The forms of $\phi_{L}^{'}$
and $\phi_{U}^{'}$ are provided in Appendix \ref{sec:HD-Derivatives}.
\end{theorem}

We simulate the limiting distribution provided in Theorem \ref{thm:endo_marginal_weak_limit}
in a similar way to the exogenous treatment case. Let $\mathbb{\hat{F}}_{\mathbf{Y}|X,n}^{*}(\cdot|x)$
denote the vector of simulated stochastic processes using estimated
influence functions in (\ref{eq:Influence_endo}) in Appendix \ref{sec:Influence}.
The following theorem shows the validity of the bootstrap procedure: 

\begin{theorem}\label{thm:endo_bootstrap} Suppose that conditions
in Theorem \ref{thm:endo_marginal_weak_limit} are satisfied. In addition,
if Assumptions \ref{assu:boot_weight} and \ref{assu:est_der_tuning}
also hold, then, we have 
\begin{align*}
\widehat{\phi_{L}^{'}}\left(\Pi_{L}(\hat{\mathbb{F}}_{\mathbf{Y}|X,n})(y,\delta|x);\sqrt{nh_{n}^{d_{x}}}\mathbb{\hat{F}}_{\mathbf{Y}|X,n}^{*}(\cdot|x),a_{n}\right) & \Rightarrow\phi_{L}^{'}\left(\Pi_{L}(\mathbb{F}_{\mathbf{Y}|X})(y,\delta|x);\mathbb{G}_{x}^{e}(\cdot)\right),\\
\widehat{\phi_{U}^{'}}\left(\Pi_{U}(\hat{\mathbb{F}}_{\mathbf{Y}|X,n})(y,\delta|x)+1;\sqrt{nh_{n}^{d_{x}}}\mathbb{\hat{F}}_{\mathbf{Y}|X,n}^{*}(\cdot|x),a_{n}\right) & \Rightarrow\phi_{U}^{'}\left(\Pi_{U}(\mathbb{F}_{\mathbf{Y}|X})(y,\delta|x)+1;\mathbb{G}_{x}^{e}(\cdot)\right),
\end{align*}
 in $l^{\infty}\left(\mathcal{Y}^{4}\right)$, conditional on data.
The forms of $\widehat{\phi_{L}^{'}}$ and $\widehat{\phi_{U}^{'}}$
are provided in Appendix \ref{sec:HD-Derivatives}. \end{theorem} 

We provide several extensions of the main results in this section
in Appendix. We extend the results of \citet{abrevaya2015estimating}
who consider estimating conditional average treatment effects on a
subset of covariates to the conditional distribution of treatment
effects (Appendix \ref{sec:semipara}). This result may be practically
relevant when the number of covariates is large. We also discuss how
to adopt the results to conduct global hypothesis testing (Appendix
\ref{sec:hypothesis_testing}). 

It is worth noting that \citet{firpo2019partial} provide formal definitions
of pointwise and uniform sharpness of bounds on distribution functions
and show that the Makarov bounds are not uniformly sharp. The inference
methods developed in this paper are based on pointwise sharp bounds
on the conditional distribution of treatment effects and are valid
uniformly over the support of treatment effects. Inference for uniformly
sharp bounds on the distribution of treatment effects is left as future
research. 

\section{\label{sec:MC} Monte Carlo Simulation }

We conduct a set of simulations to investigate the performance of
the bootstrap in finite samples. To this end, the following DGP is
considered: 
\begin{align*}
Y_{1} & =\mu_{1}+X\beta_{1}+(\phi_{1}+X\gamma_{1})U_{1},\\
Y_{0} & =\mu_{0}+X\beta_{0}+(\phi_{0}+X\gamma_{0})U_{0},\\
D & =\mathbf{1}(X\alpha\geq V),
\end{align*}
 where $X=2\tilde{X}-1$ with $\tilde{X}$ being an uniform a random
variable on $[0,1]$, $U\equiv(U_{1},U_{0})^{t}\sim N\left(\begin{pmatrix}0\\
0
\end{pmatrix},\begin{pmatrix}1 & 0\\
0 & 1
\end{pmatrix}\right)$, and $V\sim N(0,1)$. The conditional treatment effect on $X=x$
is defined as $(\mu_{1}-\mu_{0})+x(\beta_{1}-\beta_{0})+((\phi_{1}+X\gamma_{1})U_{1}-(\phi_{0}+X\gamma_{0})U_{0})$.
The parameter values are set as follows: $\beta_{1}=\gamma_{1}=1$,
$\beta_{0}=\gamma_{0}=0.9$, $\phi_{1}=\phi_{0}=1$, and $\alpha=1$.
We consider various values for parameter vector $(\mu_{1},\mu_{0})^{t}$
to investigate the performance of a KS test statistic under the null
and alternative hypotheses. Specifically, to investigate the performance
of the KS test under null hypotheses, we consider the following hypothesis:
\[
H_{0}:F_{\Delta|X}^{L}(\delta|x=0)=\left(2\cdot\Phi\left(\frac{\delta}{2}\right)-1\right)\cdot\mathbf{1}(\delta\geq0)\text{ for all }\delta,
\]
 where $\Phi(\cdot)$ is the standard normal distribution function. 

The lower bound in the null hypothesis is one derived by \citet{frank1987best},
and the null hypothesis is true if and only if $\mu_{1}=\mu_{0}=0$.
The KS statistic for this null is constructed as follows: 
\[
KS_{n}\equiv\sup_{\delta}\sqrt{nh_{n}}\Bigg|F_{\Delta|X}^{L}(\delta|x=0)-\left(2\cdot\Phi\left(\frac{\delta}{2}\right)-1\right)\cdot\mathbf{1}(\delta\geq0)\Bigg|.
\]
 To investigate the performance of the KS test under some alternative
hypotheses, we consider the cases where $(\mu_{1},\mu_{0})^{t}=(\mu,0)^{t}$
for $\mu\in\{-1,1\}$. The bandwidth is chosen to be $h_{n}=1.06\times s.d(X)\times n^{-1/6}$,
where $s.d(X)$ denotes the (sample) standard deviation of $X$.

It is worth emphasizing that the asymptotic distribution of the KS
test statistic is nonstandard and may not be uniformly valid with
respect to the DGP. For this reason, it is important to investigate
whether the KS statistic performs well with various choices for $(a_{n})$.
Specifically, we consider several rates at which $a_{n}$ grows to
the infinity: (i) $a_{n}=c\times\log\left(\log\left(nh_{n}\right)\right)/\sqrt{nh_{n}}$,
(ii) $a_{n}=c\times\sqrt{\log\left(nh_{n}\right)}/\sqrt{nh_{n}}$,
and (iii) $a_{n}=c\times\left(nh_{n}\right)^{1/6}/\sqrt{nh_{n}}$
for some $c>0$. These choices of $(a_{n})$ satisfy Assumption \ref{assu:est_der_tuning}.
We consider various values of $c$ ranging from $0.1$ to $0.5$ to
see whether the finite-sample performance of the KS test is sensitive
to the choice of $c$.\footnote{\citet{firpo2021uniform} suggest using $c=0.2$ with $a_{n}=c\times\log\left(\log\left(nh_{n}\right)\right)/\sqrt{nh_{n}}$
in our context. } 

The sample size $n$ is set to be 500, and the number of bootstrap
iterations is 500. The bootstrap weight $B$ is drawn from the standard
normal distribution, and all simulation results are obtained from
500 iterations. The nominal level is set to be 0.05. 

Table \ref{tab:sim_rp} presents the simulation results for rejection
probabilities. We find that the rejection probability under $H_{0}$
tends to decrease as $c$ increases, except for the case of $\sqrt{nh_{n}}a_{n}\propto\left(nh_{n}\right)^{1/6}$.\footnote{We also considered larger values of $c$ than $0.5$, but they are
not reported here. The simulation results with those values of $c$
suggest that the resulting confidence bands would be too conservative
(i.e., the rejection probabilities under the null hypothesis are relatively
smaller than the nominal rate). As a result, we do \textit{not }recommend
using a too large value of $c$, based on our simulation results.
It would be an interesting question how to choose $c$ or the rate
of $(a_{n})$ in a data-dependent way, but this is far beyond the
scope of this paper. We therefore leave this interesting and important
question for future research.} The KS statistic performs well in finite samples for various values
of $(a_{n})$ in the sense that the rejection probability under $H_{0}$
is close to the nominal probability and that the rejection probabilities
under $H_{1}$ are large in general. 

\begin{table}[h]
\caption{\label{tab:sim_rp} Rejection Probabilities, $n=500$, $B=500$, $x=Q_{X}(0.5)$}

\bigskip{}

\begin{centering}
\begin{tabular}{cccccccccc}
\hline 
$c$\textbackslash$a_{n}$ & \multicolumn{3}{c}{$c\cdot\log\left(\log\left(nh_{n}\right)\right)/\sqrt{nh_{n}}$} & \multicolumn{3}{c}{$c\sqrt{\log\left(nh_{n}\right)}/\sqrt{nh_{n}}$} & \multicolumn{3}{c}{$c\left(nh_{n}\right)^{1/6}/\sqrt{nh_{n}}$}\tabularnewline
\cline{2-10} \cline{3-10} \cline{4-10} \cline{5-10} \cline{6-10} \cline{7-10} \cline{8-10} \cline{9-10} \cline{10-10} 
 & $\mu=0$ & $\mu=-1$ & $\mu_{1}=1$ & $\mu=0$ & $\mu=-1$ & $\mu_{1}=1$ & $\mu=0$ & $\mu=-1$ & $\mu_{1}=1$\tabularnewline
\hline 
0.1 & 0.078 & 0.998 & 0.824 & 0.068 & 0.996 & 0.816 & 0.044 & 0.996 & 0.816\tabularnewline
0.2 & 0.066 & 0.986 & 0.806 & 0.060 & 0.992 & 0.782 & 0.064 & 0.992 & 0.782\tabularnewline
0.3 & 0.060 & 0.994 & 0.778 & 0.058 & 0.986 & 0.756 & 0.050 & 0.992 & 0.754\tabularnewline
0.4 & 0.058 & 0.994 & 0.734 & 0.044 & 0.992 & 0.748 & 0.056 & 0.988 & 0.746\tabularnewline
0.5 & 0.050 & 0.988 & 0.772 & 0.040 & 0.998 & 0.704 & 0.058 & 0.984 & 0.738\tabularnewline
\hline 
\end{tabular}
\par\end{centering}
\bigskip{}

Note: The nominal level is set to be 0.05. The case of $\mu=0$ is
where the null hypothesis is true. 
\end{table}

\section{\label{sec:Empirical} Empirical Application: The Effect of 401(k)
Plans on Net Financial Assets}

In this section, we provide an empirical example to illustrate the
usefulness of the methods proposed in this paper. We revisit the empirical
question on the effect of participation in 401(k) plans on net financial
assets investigated by many studies in the literature (e.g., \citet{Abadie2003,Chernozhukov2004,Wuethrich2019,SantAnna2022}).
Our main goals in this empirical application are twofold. First, we
empirically show that the stochastic dominance assumptions (Assumptions
\ref{assu:fsd1} and \ref{assu:fsd2}) can provide considerable identifying
power. Second, we complement the existing results that there is substantial
heterogeneity in treatment effects across income groups by estimating
the bounds on the conditional distribution of treatment effects on
income levels without an instrumental variable. In doing so, we complement
the empirical results on the effect of 401(k) plans on net financial
assets documented in the literature by providing estimation results
on the distribution of the treatment effect. Since our focus is on
verifying the identifying power of the stochastic dominance assumptions
and potential heterogeneity in the treatment effect across different
subpopulations, we do not report confidence bands for clear illustration.

The U.S. introduced several tax-deferred retirement plans, including
401(k) plans and individual retirement accounts (IRAs), in the early
1980s. These retirement plans can be used as a way to accumulate individual
assets. Many papers in the literature have considered how those tax-deferred
retirement plans affect asset accumulation or savings. The main challenge
with identifying and estimating the causal effect of 401(k) participation
on assets is that participation in 401(k) plans is endogenously determined.
Furthermore, the effect of 401(k) plans on net financial assets is
heterogeneous across income levels, as shown by \citet{Chernozhukov2004}.
Motivated by the empirical results of \citet{Chernozhukov2004}, we
focus on the distribution of treatment effects conditional on an individual's
income. Moreover, while it is common to use the eligibility for 401(k)
as an instrumental variable to point identify some distributional
effects (e.g., quantile treatment effects), our identification and
estimation strategies do not rely on such an instrumental variable. 

We use the data from \citet{SantAnna2022} for this empirical analysis.
The original dataset contains 9,910 households from the 1991 Survey
of Income and Program Participation. The dependent variable is the
amount of net financial assets measured in ten thousand dollars. We
exclude observations with a value of the dependent variable higher
than the 0.99 sample quantile and lower than the 0.01 sample quantile
of the net financial assets. This results in a sample of 9,712 households.
The treatment variable is a binary variable indicating whether a household
participates in 401(k) plans. To investigate the potential heterogeneity
across income levels, we use the income variable as the covariate
of interest. For stability of estimation, we standardize the original
variable of income and consider the sample mean and various quantiles
of income. Table \ref{tab:emp2_summary} reports the summary statistics
of the data. 

\begin{table}[H]
\caption{\label{tab:emp2_summary} Summary Statistics of the Data}

\bigskip{}

\begin{centering}
\begin{tabular}{ccccccc}
\hline 
Variables & Mean & Median & S.D. & Min & Max & Obs.\tabularnewline
\hline 
Net financial assets & 1.4506 & 0.1499 & 3.1991 & -2.350 & 21.995 & 9,712\tabularnewline
Treatment & 0.26 & 0 & 0.4387 & 0 & 1 & 9,712\tabularnewline
Income & 3.661 & 3.120 & 2.379 & 0.003 & 19.299 & 9,712\tabularnewline
Age & 40.969 & 40 & 10.311 & 25 & 64 & 9,712\tabularnewline
\hline 
\end{tabular}
\par\end{centering}
\bigskip{}

Note: The net financial assets and income are measured in \$10,000. 
\end{table}

We first discuss the validity of Assumptions \ref{assu:fsd1} and
\ref{assu:fsd2} in this empirical example. It is well known that
the preference for saving is heterogeneous in that some people have
a stronger preference for saving than others. This leads to the nonrandom
selection into participation in 401(k) plans or other tax-deferred
retirement plans (e.g., \citet{Chernozhukov2004}). Based on this
observation, it is likely that people who participate in 401(k) plans
would have a stronger preference for saving than those who do not
participate. Therefore, the net financial assets of people with a
strong preference for saving tend to be larger than those of people
with a weak preference for saving, suggesting that it may be plausible
to impose Assumption \ref{assu:fsd1} on the model. 

Assumption \ref{assu:fsd2} is consistent with the empirical example
as most tax-deferred retirement plans, including 401(k) plans, by
themselves increase the amount of assets that an individual possess.
Therefore, regardless of whether an individual participates in 401(k)
plans or not, the potential net financial assets that one would have
had if she participated in 401(k) plans are likely to be larger than
those she would have had if she did not participate in 401(k). As
a result, it is plausible to impose Assumption \ref{assu:fsd2} on
the model. 

For estimation, we set $h_{n}=1.06\times n^{-1/(5+d_{x})}$ and $a_{n}=0.2\times\log\left(\log\left(nh_{n}^{d_{x}}\right)\right)/\sqrt{nh_{n}^{d_{x}}}$
with $d_{x}=1$. We denote the $\tau$-th (sample) quantile of income
by $Q_{income}(\tau)$. 

We note that when Assumptions \ref{assu:fsd1} and \ref{assu:fsd2}
are not imposed, the estimated bounds are the logical ones, regardless
of the conditioning value of the income.\footnote{By the logical bounds, we mean that the lower and upper bounds are
equal to 0 and 1, respectively, for all $\delta\in Supp(\Delta|X=x)$. } On the other hand, the estimated bounds under Assumptions \ref{assu:fsd1}
and \ref{assu:fsd2} are informative in the sense that they are not
the logical bounds. This indicates that the stochastic dominance assumptions
have considerable identifying power. 

We find substantial heterogeneity in the distribution of treatment
effects across different values of the income. Figure \ref{fig:endo_treat}
compares the estimated bounds on the conditional distribution of treatment
effects at the 0.2 and 0.8 quantiles of the income. Specifically,
the star-marked lines are the bounds on the conditional distribution
of treatment effects given the 0.2 quantile of the income. The circle-marked
lines are the bounds on the conditional distribution of treatment
effects given the 0.8 quantile of income. We find that the lower bound
conditional on the 0.2 quantile of the income is larger than the lower
bound conditional on the 0.8 quantile of the income. Conversely, the
upper bound conditional on the 0.8 quantile of the income is larger
than the upper bound conditional on the 0.2 quantile of the income
for all $\delta\in[-24,24]$ (i.e., from -\$240,000 to \$240,000).
When considering the bounds on $\Pr\left(Y_{1}-Y_{0}\leq1|X=x\right)$
for $x\in\left\{ Q_{income}(0.2),Q_{income}(0.8)\right\} $, the estimation
results show that $\Pr\left(Y_{1}-Y_{0}\leq1|X=Q_{income}(0.2)\right)\in\left[0.5804,1\right]$
and that $\Pr\left(Y_{1}-Y_{0}\leq1|X=Q_{income}(0.8)\right)\in\left[0.0638,1\right]$.
These estimated bounds suggest that the proportion of individuals
who experience a positive treatment effect larger than \$10,000 among
those with the income being equal to the 0.2 sample quantile is at
most 41.96\%. The proportion among those with the income being equal
to the 0.8 sample quantile is at most 93.62\%. As a result, there
may be a possibility that the proportion of individuals who experience
a certain level of positive treatment effect is larger when considering
a higher level of income. This argument is in part consistent with
the finding of \citet{Chernozhukov2004} that the quantile treatment
effects of 401(k) plans on net financial wealth tend to increase as
the income level increases (see Figure 2 in \citet{Chernozhukov2004}). 

Figure \ref{fig:comparison_all_quantiles} compares the estimated
bounds at a specific quantile level with those at the mean of income
under Assumptions \ref{assu:fsd1} and \ref{assu:fsd2}. When considering
a low quantile level, e.g., $\tau\in\{0.1,0.2,0.3\}$, we find that
the lower bound conditional on $X=Q_{income}(\tau)$ is larger than
that conditional on the mean of income. The upper bound conditional
on $X=Q_{income}(\tau)$ is smaller than that conditional on the mean
of income over the potential support of the treatment effect. However,
when considering a high quantile level, e.g., $\tau\in\{0.7,0.8,0.9\}$,
the estimation results are the opposite. These estimation results
indicate that the treatment effect of 401(k) plans on net financial
assets is likely to be heterogeneous across income levels, which is
consistent with the finding of \citet{Chernozhukov2004}. 

\section{\label{sec:Conclusion} Conclusion }

This paper considers identification and estimation of bounds on the
conditional distribution of treatment effects. The conditional distribution
may provide evidence on potential heterogeneity in treatment effects
across subpopulations that are defined in terms of values of covariates,
and therefore, is of practical importance in many empirical studies.
We show that when the treatment is endogenously determined, one can
tighten the bounds by imposing stochastic dominance assumptions. These
assumptions are consistent with many economic theories and easy to
interpret, and the resulting bounds on the distribution of treatment
effects are easy to compute. We propose nonparametric estimators of
the bounds and establish the uniform asymptotic theory based on the
novel approach of \citet{fang2019inference} and \citet{firpo2021uniform}.
The asymptotic theory in this paper is useful for constructing uniform
confidence bands and conducting statistical tests for global hypotheses.
We then provide an empirical application of the methodology proposed
in this paper to illustrate its relevance to empirical research. 

There are several interesting directions for future research. First,
one can consider inference that is uniformly valid regardless of whether
the (conditional) distributions of the potential outcomes are point
identified or partially identified, as in, for example, \citet{Imbens2004},
\citet{Stoye2009}, and \citet{andrews2010inference}. Second, one
can consider testing global hypotheses, such as stochastic dominance.
Although we cannot directly test stochastic dominance between two
conditional distribution of treatment effects as they are not point
identified, one can provide weak evidence by using similar arguments
to those in \citet{firpo2021uniform}. Third, it would be fruitful
to develop inference methods that are uniformly valid in values of
covariates. In work in progress, we consider a semiparametric approach
to uniform inference over the support of treatment effects and covariates.
It is expected to help resolve many interesting questions in empirical
analysis that cannot be answered by the framework proposed in this
paper. Lastly, it is worth considering inference in the presence of
instrumental variables. 

\bibliographystyle{chicago}
\bibliography{dist_te}

\newpage{} 

\appendix 
\counterwithin{figure}{section} 
\counterwithin{table}{section}
\counterwithin{equation}{section}

\section{\label{sec:HD-Derivatives} Hadamard Directional Derivatives }

To derive the Hadamard directional derivatives of $\phi_{L}$ and
$\phi_{U}$, we introduce some additional notations that are used
in \citet{firpo2021uniform}. For a set-valued map (or correspondence)
from $\mathcal{A}$ to the collection of subsets of $\mathcal{B}$,
$S$, $gr(S)$ is the graph of $S$ in $\mathcal{A}\times\mathcal{B}$.
Recall that for 
\[
\mathbb{F}_{\mathbf{Y}|X}((y_{1l},y_{1u},y_{0l},y_{0u})|x)\equiv\begin{pmatrix}LB_{1|X}(y_{1l}|x)\\
UB_{1|X}(y_{1u}|x)\\
LB_{0|X}(y_{0l}|x)\\
UB_{0|X}(y_{0u}|x)
\end{pmatrix},
\]
 we have defined 
\begin{align*}
\Pi_{L}\left(\mathbb{F}_{\mathbf{Y}|X}\right)(y,\delta|x) & \equiv LB_{1|X}(y|x)-UB_{0|X}(y-\delta|x),\\
\Pi_{U}\left(\mathbb{F}_{\mathbf{Y}|X}\right)(y,\delta|x) & \equiv UB_{1|X}(y|x)-LB_{0|X}(y-\delta|x).
\end{align*}

Let $S^{+}:\mathcal{D}\rightrightarrows\mathcal{Y}$ be a set-valued
map such that for each $\delta\in\mathcal{D}$, 
\[
S^{+}(\delta)\equiv\left\{ y\in\mathcal{Y}:y=\arg\sup\Pi_{L}(\mathbb{F}_{\mathbf{Y}|X})(y,\delta|x)\right\} .
\]
 Similarly, we define a set-valued map $S^{-}:\mathcal{D}\rightrightarrows\mathcal{Y}$
such that for each $\delta\in\mathcal{D}$, 
\[
S^{-}(\delta)\equiv\left\{ y\in\mathcal{Y}:y=\arg\inf(\Pi_{U}(\mathbb{F}_{\mathbf{Y}|X})(y,\delta|x)+1)\right\} .
\]
 For given $\epsilon>0$, $\delta\in\mathcal{D}$, and a set-valued
correspondence $S:\mathcal{D}\rightrightarrows\mathcal{Y}$, we define
the set of $\epsilon$-maximizers as 
\[
\Lambda_{f}(\delta,\epsilon;S)\equiv\left\{ y\in S(\delta):f(\delta,y)\geq\sup_{y\in\mathcal{Y}}f(\delta,y)-\epsilon\right\} .
\]

Let $\mathbb{H}\equiv(h_{1},h_{2},h_{3},h_{4})^{t}$ be a four-dimensional
vector-valued function. After some algebra, one can show that the
Hadamard directional derivatives of $\phi_{L}$ and $\phi_{U}$ are
\begin{equation}
\begin{aligned}\phi_{L}^{'}(f;\mathbb{H}) & \equiv\max\left\{ \lim_{\epsilon\rightarrow0^{+}}\sup_{\delta\in\mathcal{D}}\sup_{y\in\Lambda_{f}(\delta,\epsilon;S^{+})}(h_{1}(y)-h_{4}(y-\delta)),\right.\\
 & \ \ \ \ \ \ \ \ \ \ \ \left.\lim_{\epsilon\rightarrow0^{+}}\sup_{\delta\in\mathcal{D}}\inf_{y\in\Lambda_{f}(\delta,\epsilon;S^{+})}-(h_{2}(y)-h_{3}(y-\delta))\right\} ,\\
\phi_{U}^{'}(f;\mathbb{H}) & \equiv\max\left\{ \lim_{\epsilon\rightarrow0^{+}}\sup_{\delta\in\mathcal{D}}\sup_{y\in\Lambda_{-f}(\mathcal{\delta},\epsilon;S^{-})}-(h_{1}(y)-h_{4}(y-\delta)+1),\right.\\
 & \ \ \ \ \ \ \ \ \ \ \ \left.\lim_{\epsilon\rightarrow0^{+}}\sup_{\delta\in\mathcal{D}}\inf_{y\in\Lambda_{-f}(\delta,\epsilon;S^{-})}(h_{2}(y)-h_{3}(y-\delta)+1)\right\} .
\end{aligned}
\label{eq:phi_derivative}
\end{equation}
 The form of $\phi_{L}^{'}(f;\mathbb{H})$ can be found in the second
part of Theorem 3.2. of \citet{firpo2021uniform}. Here we derive
the Hadamard directional derivative of $\phi_{U}$. Recall that 
\begin{align*}
\phi_{U}(f) & =\sup_{\delta\in\mathcal{D}}\Big|\inf_{y\in\mathcal{Y}}f(\delta,y)\Big|\\
 & =\sup_{\delta\in\mathcal{D}}\Big|\sup_{y\in\mathcal{Y}}(-f(\delta,y))\Big|
\end{align*}
 since $\inf f=-\sup(-f)$. 
\begin{align*}
\phi_{U}(f+t_{n}h)-\phi_{U}(f) & =\sup_{\delta\in\mathcal{D}}\Big|\sup_{y\in\mathcal{Y}}(-f(\delta,y)-t_{n}h(\delta,y))\Big|-\sup_{\delta\in\mathcal{D}}\Big|\sup_{y\in\mathcal{Y}}(-f(\delta,y))\Big|\\
 & =\sup_{\delta\in\mathcal{D}}\Big|\sup_{y\in\mathcal{Y}}(-f(\delta,y)+t_{n}(-h(\delta,y)))\Big|-\sup_{\delta\in\mathcal{D}}\Big|\sup_{y\in\mathcal{Y}}(-f(\delta,y))\Big|.
\end{align*}
 By using the same argument of the proof of Theorem 3.2. of \citet{firpo2021uniform},
we can obtain that 
\[
\phi_{U}^{'}(f;h)\equiv\max\left\{ \lim_{\epsilon\rightarrow0^{+}}\sup_{\delta\in\mathcal{D}}\sup_{y\in\Lambda_{-f}(\mathcal{\delta},\epsilon)}-h(\delta,y),\lim_{\epsilon\rightarrow0^{+}}\sup_{\delta\in\mathcal{D}}\inf_{y\in\Lambda_{-f}(\delta,\epsilon)}h(\delta,y)\right\} .
\]

Let $a_{n}$ be a positive real sequence satisfying the conditions
in Assumption \ref{assu:est_der_tuning}. Then, one can estimate $\phi_{L}^{'}(f;\mathbb{H})$
and $\phi_{U}^{'}(f;\mathbb{H})$ as follows: 

\begin{equation}
\begin{aligned}\widehat{\phi_{L}^{'}}\left(f;\mathbb{H},a_{n}\right) & \equiv\max\left\{ \sup_{\delta\in\mathcal{D}}\sup_{y\in\Lambda_{f}(\delta,a_{n};S^{+})}(h_{1}(y)-h_{4}(y-\delta)),\right.\\
 & \ \ \ \ \ \ \ \ \ \ \ \left.\sup_{\delta\in\mathcal{D}}\inf_{y\in\Lambda_{f}(\delta,a_{n};S^{+})}-(h_{2}(y)-h_{3}(y-\delta))\right\} ,\\
\widehat{\phi_{U}^{'}}\left(f;\mathbb{H},a_{n}\right) & \equiv\max\left\{ \sup_{\delta\in\mathcal{D}}\sup_{y\in\Lambda_{-f}(\mathcal{\delta},a_{n};S^{-})}-(h_{1}(y)-h_{4}(y-\delta)+1),\right.\\
 & \ \ \ \ \ \ \ \ \ \ \ \left.\sup_{\delta\in\mathcal{D}}\inf_{y\in\Lambda_{-f}(\delta,a_{n};S^{-})}(h_{2}(y)-h_{3}(y-\delta)+1)\right\} .
\end{aligned}
\label{eq:est_phi_derivative}
\end{equation}

\section{\label{sec:Influence} Estimation of Influence Functions}

When the unconfoundedness and overlap assumptions hold, the conditional
distributions of the potential outcomes are point identified and they
can be estimated as 
\begin{align*}
\hat{F}_{1|X,n}(y_{1}|x) & \equiv\frac{\sum_{i}^{n}\mathbf{1}(Y_{i}\leq y_{1})D_{i}K(\frac{X_{i}-x}{h_{n}})}{\sum_{i}^{n}D_{i}K(\frac{X_{i}-x}{h_{n}})},\\
\hat{F}_{0|X,n}(y_{0}|x) & \equiv\frac{\sum_{i}^{n}\mathbf{1}(Y_{i}\leq y_{0})(1-D_{i})K(\frac{X_{i}-x}{h_{n}})}{\sum_{i}^{n}(1-D_{i})K(\frac{X_{i}-x}{h_{n}})}.
\end{align*}
 Then, one can estimate the influence functions of $\hat{F}_{1|X,n}(y|x)$
and $\hat{F}_{0|X,n}(y|x)$ for the $i-$th observation by 

\begin{equation}
\begin{aligned}\hat{\psi}_{1,i}(y|x) & \equiv\frac{\{\mathbf{1}(Y_{i}\leq y)-\hat{F}_{1|X,n}(y|x)\}\cdot D_{i}\cdot K(\frac{X_{i}-x}{h_{n}})}{\sum_{j}^{n}D_{j}\cdot K(\frac{X_{j}-x}{h_{n}})},\\
\hat{\psi}_{0,i}(y|x) & \equiv\frac{\{\mathbf{1}(Y_{i}\leq y)-\hat{F}_{0|X,n}(y|x)\}\cdot(1-D_{i})\cdot K(\frac{X_{i}-x}{h_{n}})}{\sum_{j}^{n}(1-D_{j})\cdot K(\frac{X_{j}-x}{h_{n}})},
\end{aligned}
\label{eq:influence_exo}
\end{equation}
 respectively. 

We now assume that Assumptions \ref{assu:fsd1} and \ref{assu:fsd2}
hold. Then, the bounds on the conditional distributions of the potential
outcomes can be estimated by using 
\begin{align*}
\hat{F}_{1|1X,n}(y_{1}|x) & \equiv\frac{\sum_{i}^{n}\mathbf{1}(Y_{i}\leq y_{1})D_{i}K(\frac{X_{i}-x}{h_{n}})}{\sum_{i}^{n}D_{i}K(\frac{X_{i}-x}{h_{n}})},\\
\hat{F}_{0|0X,n}(y_{0}|x) & \equiv\frac{\sum_{i}^{n}\mathbf{1}(Y_{i}\leq y_{0})(1-D_{i})K(\frac{X_{i}-x}{h_{n}})}{\sum_{i}^{n}(1-D_{i})K(\frac{X_{i}-x}{h_{n}})},\\
\hat{F}_{Y|X,n}(y|x) & \equiv\frac{\sum_{i}^{n}\mathbf{1}(Y_{i}\leq y)K(\frac{X_{i}-x}{h_{n}})}{\sum_{i}^{n}K(\frac{X_{i}-x}{h_{n}})},
\end{align*}

Then, one can estimate the influence functions of $\hat{F}_{1|1X,n}(y|x)$,
$\hat{F}_{0|0X,n}(y|x)$, and $\hat{F}_{Y|X,n}(y|x)$ for the $i$-th
observation by 
\begin{equation}
\begin{aligned}\hat{\psi}_{11,i}(y|x) & \equiv\frac{\{\mathbf{1}(Y_{i}\leq y)-\hat{F}_{1|1X,n}(y|x)\}\cdot D_{i}\cdot K(\frac{X_{i}-x}{h_{n}})}{\sum_{j}^{n}D_{j}K(\frac{X_{j}-x}{h_{n}})},\\
\hat{\psi}_{00,i}(y|x) & \equiv\frac{\{\mathbf{1}(Y_{i}\leq y)-\hat{F}_{0|0X,n}(y|x)\}\cdot(1-D_{i})\cdot K(\frac{X_{i}-x}{h_{n}})}{\sum_{j}^{n}(1-D_{j})K(\frac{X_{j}-x}{h_{n}})},\\
\hat{\psi}_{Y,i}(y|x) & \equiv\frac{\{\mathbf{1}(Y_{i}\leq y)-\hat{F}_{Y|X,n}(y|x)\}K(\frac{X_{i}-x}{h_{n}})}{\sum_{j}^{n}K(\frac{X_{j}-x}{h_{n}})},
\end{aligned}
\label{eq:Influence_endo}
\end{equation}
 respectively.

\section{\label{sec:semipara} Semiparametric Estimation of Conditional Distributions
of Treatment Effects on a Subset of $X$ }

We may be interested in the conditional distribution of treatment
effects within some subpopulation characterized by a subset of $X$.
For example, suppose we are interested in the effect of smoking on
birth weight. There are many potential factors that affect the birth
weight, such as mother's age and education level, family income, and
baby's gender. Among those factors, the focus may be on the heterogeneity
in the treatment effect across mother's age (we call this variable
$X_{1}$). Note that Assumption \ref{assu:unconfounded} does not
necessarily imply that $(Y_{1},Y_{0})\perp D|X_{1}$, and \citet{abrevaya2015estimating}
develop approaches to estimating conditional average treatment effects
to capture heterogeneity in some subpopulation. We complement \citet{abrevaya2015estimating}
by providing a way to identify and estimate the distribution of treatment
effects of some subpopulation.

Let $x_{1}\in\mathcal{X}_{1}\equiv Supp(X_{1})\subseteq\mathbb{R}^{d_{1}}$
be given. Then, from Lemma \ref{lem:id_unconfounded} and the law
of iterated expectations, it is straightforward to see that 
\begin{equation}
\begin{alignedat}{1}F_{1|X_{1}}(y|x_{1}) & =\mathbb{E}\left[\frac{D\mathbf{1}(Y\leq y)}{p_{0}(X)}\Big|X_{1}=x_{1}\right],\\
F_{0|X_{1}}(y|x_{1}) & =\mathbb{E}\left[\frac{(1-D)\mathbf{1}(Y\leq y)}{1-p_{0}(X)}\Big|X_{1}=x_{1}\right],
\end{alignedat}
\label{eq:id_cond_dist-1}
\end{equation}
 where $F_{d|X_{1}}(y|x_{1})\equiv\Pr(Y_{d}\leq y|X_{1}=x_{1})$ for
each $d\in\{0,1\}$. We focus on a semiparametric approach that uses
parametric estimation for the propensity score but nonparametric kernel-smoothing
estimation for the second-step, which was proposed by \citet{abrevaya2015estimating}.
We use parametric estimators of the propensity score mainly for practical
reasons. First, we can easily incorporate discrete regressors when
estimating the propensity score. Second, we can avoid the curse of
dimensionality when the dimension of $X$ is very large. Lastly, the
semiparametric approach is expected to be less sensitive to tuning
parameters as the number of tuning parameters required for estimation
is fewer for the semiparametric approach than fully nonparametric
approaches. While the semiparametric approach using a parametric specification
for the propensity has some advantages over the fully nonparametric
approach, it can lead to model misspecification. To mitigate the issues
about potential model misspecification, one may employ the fully nonparametric
approach in \citet{abrevaya2015estimating}.

We assume that the propensity score function is parameterized by a
finite-dimensional parameter $\theta_{0}$: $p_{0}(x)=p(x;\theta_{0})$
for all $x\in\mathcal{X}$. Let $\hat{\theta}_{n}$ be an estimator
of $\theta_{0}$, then the propensity score function can be estimated
by $p(x;\hat{\theta}_{n})$. Let $K_{1}(\cdot):\mathbb{R}^{d_{1}}\rightarrow\mathbb{R}$
be a kernel function that is symmetric around zero. The identification
results in equation (\ref{eq:id_cond_dist-1}) suggest the following
semiparametric estimators of the conditional distributions of $Y_{1}$
and $Y_{0}$ on $X_{1}=x_{1}$: 
\begin{align}
\hat{F}_{1|X_{1},n}(y|x_{1}) & \equiv\sum_{i}^{n}\frac{D_{i}\cdot\mathbf{1}(Y_{i}\leq y)}{p(X_{i};\hat{\theta}_{n})}\cdot K_{1}(\frac{X_{1i}-x_{1}}{h_{1n}})\Big/\sum_{i}K_{1}(\frac{X_{1i}-x_{1}}{h_{1n}}),\label{eq:est_F1-1}\\
\hat{F}_{0|X_{1},n}(y|x_{1}) & \equiv\sum_{i}^{n}\frac{(1-D_{i})\cdot\mathbf{1}(Y_{i}\leq y)}{(1-p(X_{i};\hat{\theta}_{n}))}\cdot K_{1}(\frac{X_{1i}-x_{1}}{h_{1n}})\Big/\sum_{i}K_{1}(\frac{X_{1i}-x_{1}}{h_{1n}}),\label{eq:est_F0-1}
\end{align}
 where $h_{1n}$ is a bandwidth. Let $\hat{\mathbb{F}}_{\mathbf{Y}|X_{1},n}((y_{1},y_{0})|x_{1})\equiv\left(\hat{F}_{1|X_{1},n}(y_{1}|x_{1}),\hat{F}_{0|X_{1},n}(y_{0}|x_{1})\right)^{t}$.
With these estimators of $F_{1|X_{1}}$ and $F_{0|X_{1}}$, the bounds
in Lemma \ref{prop:ID_dist_TE_Unconfounded} can be estimated as follows:
\begin{eqnarray}
\hat{F}_{\Delta|X_{1},n}^{L}(\delta|x_{1}) & \equiv & \max\left(\sup_{y}\left\{ \hat{F}_{1|X_{1},n}(y|x_{1})-\hat{F}_{0|X_{1},n}(y-\delta|x_{1})\right\} ,0\right),\label{eq:est_F_Delta_LB-1}\\
\hat{F}_{\Delta|X_{1},n}^{U}(\delta|x_{1}) & = & \min\left(\inf_{y}\left\{ \hat{F}_{1|X_{1},n}(y|x)-\hat{F}_{0|X_{1},n}(y-\delta|x_{1})\right\} ,0\right)+1.\label{eq:est_F_Delta_UB-1}
\end{eqnarray}

One can use the same bootstrap procedure to approximate the distributions
of the estimated bounds on $F_{\Delta|X_{1}}(\cdot|x_{1})$. Define
\[
\tilde{\mathbb{F}}_{\mathbf{Y}|X_{1},n}^{*}(y|x;B_{1})=\begin{pmatrix}F_{1|X_{1},n}^{*}(y|x_{1};B_{1})\\
F_{0|X_{1},n}^{*}(y|x_{1};B_{1})
\end{pmatrix}\equiv\begin{pmatrix}\sum_{i}B_{1i}\hat{\psi}_{11,i}(y|x_{1})\\
\sum_{i}B_{1i}\hat{\psi}_{10,i}(y|x_{1})
\end{pmatrix},
\]
 where 
\begin{align*}
\hat{\psi}_{11,i}(y|x_{1}) & \equiv\left(\frac{D_{i}\cdot(\mathbf{1}(Y_{i}\leq y)-\hat{F}_{1|X_{1},n}(y|x_{1}))}{\hat{p}_{n}(X_{i};\hat{\theta}_{n})}\right)\cdot K_{1}(\frac{X_{1i}-x_{1}}{h_{1n}})\Big/\sum_{j}K_{1}(\frac{X_{1j}-x_{1}}{h_{1n}}),\\
\hat{\psi}_{10,i}(y|x_{1}) & \equiv\left(\frac{(1-D_{i})\cdot(\mathbf{1}(Y_{i}\leq y)-\hat{F}_{0|X_{1},n}(y|x_{1}))}{1-\hat{p}_{n}(X_{i};\hat{\theta}_{n})}\right)\cdot K_{1}(\frac{X_{1i}-x_{1}}{h_{1n}})\Big/\sum_{j}K_{1}(\frac{X_{1j}-x_{1}}{h_{1n}}),
\end{align*}
 and $B_{1}$ is a random variable independent of the data. We consider
the following set of assumptions to prove the validity of the bootstrap
in this case.

\begin{assumption}\label{assu:fx1} (i) The support of $X_{1}$,
$\mathcal{X}_{1}$, is a compact subset of $\mathbb{R}^{d_{1}}$;
(ii) the distribution of $X_{1}$ admits its density $f_{X_{1}}(\cdot)$
on $\mathcal{X}_{1}$ such that $0<\inf_{x_{1}\in\mathcal{X}_{1}}f_{X_{1}}(x_{1})<\sup_{x_{1}\in\mathcal{X}_{1}}f_{X_{1}}(x_{1})<\infty$.
The density function $f_{X_{1}}(\cdot)$ is twice continuously differentiable
and $\sup_{x_{1}\in\mathcal{X}_{1}}|f_{X_{1}}^{(1)}(x_{1})|$ and
$\sup_{x_{1}\in\mathcal{X}_{1}}|f_{X_{1}}^{(2)}(x_{1})|$ are bounded.
\end{assumption}

\begin{assumption}\label{assu:x1_smooth} For given $j\in\{0,1\}$
and $x_{1}\in\mathcal{X}_{1}$, $F_{j|X_{1}}(y|x_{1})$ is continuously
differentiable with respect to $x$ and the derivative is uniformly
bounded. \end{assumption}

\begin{assumption}\label{assu:propensity_para} Let $\hat{\theta}_{n}$
be an estimator of $\theta_{0}\in\Theta\subseteq\mathbb{R}^{d_{\theta}}$,
where propensity score function $p(x;\theta_{0})$ is the propensity
score function. $\hat{\theta}_{n}$ satisfies $\sup_{x\in\mathcal{X}}|p(x;\hat{\theta}_{n})-p(x;\theta_{0})|=O_{p}(n^{-1/2})$.
\end{assumption}

\begin{assumption}\label{assu:kernel1} The kernel function $K_{1}(\cdot)$
is a $d_{1}$-dimensional product kernel with a univariate bounded
kernel function $k_{1}(\cdot):\mathbb{R}\rightarrow\mathbb{R}_{+}$
such that $\int k_{1}(u)du=1$, $\int uk_{1}(u)du=0$, and $\int u^{2}k_{1}(u)du<\infty$.
The support of the univariate kernel function is compact.\end{assumption}

\begin{assumption}\label{assu:bandwidth1} (i) $h_{1n}\rightarrow0$
; (ii) $nh_{1n}^{d_{1}}\rightarrow\infty$; and (iii) $nh_{1n}^{d_{1}+4}\rightarrow0$.
\end{assumption}

\begin{assumption}\label{assu:boot_weight1} Let $B_{1}$ be a random
variable that is independent of the data $\mathcal{W}$ such that
$\mathbb{E}[B_{1}]=0$, $Var(B_{1})=1$, and $\int_{0}^{\infty}\sqrt{\Pr(|B_{1}|>x)}dx<\infty$.
\end{assumption}

\begin{assumption}\label{assu:tuning_1} Let $a_{1n}$ be a sequence
of positive real numbers such that $a_{1n}\downarrow0$ and $a_{1n}\sqrt{nh_{1n}^{d_{1}}}\rightarrow\infty$
\end{assumption}

Define $||K_{1}||_{2}^{2}\equiv\int K_{1}^{2}(u)du$, $G_{1|X_{1}}(y|x_{1})\equiv\mathbb{E}\left[\frac{F_{1|X}(y|X)}{p_{0}(X)}\Big|X_{1}=x_{1}\right]$,
and $G_{0|X_{1}}(y|x_{1})\equiv\mathbb{E}\left[\frac{F_{0|X}(y|X)}{1-p_{0}(X)}\Big|X_{1}=x_{1}\right]$.
The following theorem establishes the asymptotic distributions of
the estimated bounds on the conditional distribution of treatment
effects and the bootstrap validity: $\Pi_{L}(\mathbb{F}_{\mathbf{Y}|X_{1},n})(y,\delta|x_{1})=\Pi_{U}(\mathbb{F}_{\mathbf{Y}|X_{1}})(y,\delta|x_{1})=F_{1|X_{1}}(y|x_{1})-F_{0|X_{1}}(y-\delta|x_{1})$.

\begin{theorem}\label{thm:bootstrap-1} Let $x_{1}\in int(\mathcal{X}_{1})$
be given. Suppose that Assumptions \ref{assu:unconfounded}, \ref{assu:bound_propensity},
\ref{assu:iid}, and \ref{assu:fx1}--\ref{assu:bandwidth1} hold.
Then, 
\begin{equation}
\begin{aligned}\sqrt{nh_{1n}^{d_{1}}}\left(\phi_{L}\left(\hat{\mathbb{F}}_{\mathbf{Y}|X_{1},n}(\cdot|x_{1})\right)-\phi_{L}\left(\mathbb{F}_{\mathbf{Y}|X_{1}}(\cdot|x_{1})\right)\right) & \Rightarrow\phi_{L}^{'}\left(\Pi_{L}(\mathbb{F}_{\mathbf{Y}|X_{1}})(y,\delta|x_{1});\tilde{\mathbb{G}}(\cdot)\right)\ \text{in }l^{\infty}(\mathcal{Y}^{2}),\\
\sqrt{nh_{1n}^{d_{1}}}\left(\phi_{U}\left(\hat{\mathbb{F}}_{\mathbf{Y}|X_{1},n}(\cdot|x_{1})\right)-\phi_{U}\left(\mathbb{F}_{\mathbf{Y}|X_{1}}(\cdot|x_{1})\right)\right) & \Rightarrow\phi_{U}^{'}\left(\Pi_{U}(\mathbb{F}_{\mathbf{Y}|X_{1}})(y,\delta|x_{1})+1;\tilde{\mathbb{G}}(\cdot)\right)\ \text{in }l^{\infty}(\mathcal{Y}^{2}),
\end{aligned}
\label{eq:te_asymptotic-1}
\end{equation}
 where $\mathbb{F}_{\mathbf{Y}|X_{1}}((y_{1},y_{0})|x_{1})\equiv\left(F_{1|X_{1}}(y_{1}|x_{1}),F_{0|X_{1}}(y_{0}|x_{0})\right)^{t}$
and $\tilde{\mathbb{G}}((\cdot,\cdot))\equiv(\tilde{\mathbb{G}}_{1}(\cdot),\tilde{\mathbb{G}}_{0}(\cdot))^{t}$
is a two-dimensional Gaussian process with mean zero and covariance
kernels 
\[
\tilde{H}_{1}(y_{1},y_{2})\equiv\left\{ G_{1|X_{1}}(\min(y_{1},y_{2})|x_{1})-F_{1|X_{1}}(y_{1}|x_{1})F_{1|X_{1}}(y_{2}|x_{1})\right\} \frac{||K_{1}||_{2}^{2}}{f_{X_{1}}(x_{1})}
\]
 and 
\[
\tilde{H}_{0}(y_{1},y_{2})\equiv\left\{ G_{0|X_{1}}(\min(y_{1},y_{2})|x_{1})-F_{0|X_{1}}(y_{1}|x_{1})F_{0|X_{1}}(y_{2}|x_{1})\right\} \frac{||K_{1}||_{2}^{2}}{f_{X_{1}}(x_{1})},
\]
 respectively.

If, in addition, Assumptions \ref{assu:boot_weight1}--\ref{assu:tuning_1}
hold, then, conditional on data, 
\begin{align*}
\widehat{\phi_{L}^{'}}\left(\Pi_{L}(\hat{\mathbb{F}}_{\mathbf{Y}|X_{1},n})(y,\delta|x_{1});\sqrt{nh_{1n}^{d_{1}}}\tilde{\mathbb{F}}_{\mathbf{Y}|X_{1},n}^{*}(\cdot|x_{1};B_{1}),a_{1n}\right) & \Rightarrow\phi_{L}^{'}\left(\Pi_{L}(\mathbb{F}_{\mathbf{Y}|X_{1}})(y,\delta|x_{1});\tilde{\mathbb{G}}(\cdot)\right),\\
\widehat{\phi_{U}^{'}}\left(\Pi_{U}(\hat{\mathbb{F}}_{\mathbf{Y}|X_{1},n})(y,\delta|x_{1})+1;\sqrt{nh_{1n}^{d_{1}}}\mathbb{\tilde{F}}_{\mathbf{Y}|X_{1},n}^{*}(\cdot|x_{1};B_{1}),a_{1n}\right) & \Rightarrow\phi_{U}^{'}\left(\Pi_{U}(\mathbb{F}_{\mathbf{Y}|X_{1}})(y,\delta|x_{1})+1;\mathbb{\tilde{G}}(\cdot)\right),
\end{align*}
 in $l^{\infty}(\mathcal{Y}^{2})$. \end{theorem} 

\section{\label{sec:hypothesis_testing} Global Hypotheses Testing }

In this section, we briefly discuss global hypotheses testing for
the bounds on the distribution of treatment effects to illustrate
the usefulness of our approach. We focus on the case where we compare
lower bounds between two groups, but the result can easily be generalized
to other cases. 

We first consider testing lower bounds on the distribution of treatment
effects between two groups, where each group is defined in terms of
the value of the covariate. Suppose that we consider two values $x_{A},x_{B}\in int(\mathcal{X})$.
The resulting lower bounds are denoted by $L_{A}(\delta)$ and $L_{B}(\delta)$,
respectively (i.e. $L_{A}(\delta)=\sup_{y}\Pi(\mathbb{F}_{\mathbf{Y}|X})(y,\delta|x_{A})$
and $L_{B}(\delta)=\sup_{y}\Pi(\mathbb{F}_{\mathbf{Y}|X})(y,\delta|x_{B})$),
and the null hypothesis is as follows: 
\[
H_{0}:L_{A}(\delta)=L_{B}(\delta),\ \forall\delta\in\mathcal{D}.
\]
 Let $\mu$ be the Lebesgue measure on $\mathcal{D}$ and define 
\[
\theta_{L,p,e}\equiv\left(\int_{\mathcal{D}}\Big|L_{A}(\delta)-L_{B}(\delta)\Big|^{p}d\mu(\delta)\right)^{1/p}
\]
 for some $1\leq p<\infty$. Replacing $L_{A}(\delta)$ and $L_{B}(\delta)$
with their estimators, we obtain a uniform test statistic $\hat{\theta}_{L,P,e,n}.$
Define $\rho(f)(\delta|x)\equiv\sup_{y\in\mathcal{Y}}f(y,\delta|x)$,
then the Hadamard directional derivative of $\rho(f)(\delta|x)$ for
directions $h$ at $f$ is 
\[
\rho_{f}^{'}(h)(\delta|x)=\lim_{\epsilon\rightarrow0}\sup_{y\in\Lambda_{f}(\delta,\epsilon;S^{+})}h(y,\delta|x)
\]
 (cf. \citet{firpo2021uniform}). We also define $\mathbb{G}_{(x_{A},x_{B})}((y_{1},y_{2},y_{3},y_{4}))\equiv\left(\mathbb{G}_{x_{A}}((y_{1},y_{2})),\mathbb{G}_{x_{B}}((y_{3},y_{4}))\right)^{t}$. 

\begin{theorem}\label{thm:test_equality} Suppose that Assumptions
\ref{assu:unconfounded} and \ref{assu:bound_propensity} hold and
that $\mu(\mathcal{D})<\infty$. Let $x_{A},x_{B}\in int(\mathcal{X})$
be given and Assumptions \ref{assu:iid}--\ref{assu:bandwidth} hold.
Then, under $H_{0}$, 
\[
\sqrt{nh_{n}^{d_{x}}}\hat{\theta}_{L,P,e,n}\Rightarrow\theta_{L,P,e}^{'}\left(\mathbb{G}_{(x_{A},x_{B})}((\cdot,\cdot,\cdot,\cdot))\right),
\]
 where, for $\mathbb{H}_{A}(\cdot,\cdot)=(h_{1}(\cdot|x_{A}),h_{2}(\cdot|x_{A}))^{t}$,
$\mathbb{H}_{B}(\cdot,\cdot)=(h_{3}(\cdot|x_{B}),h_{4}(\cdot|x_{B}))^{t}$,
\[
\theta_{L,P,e}^{'}(\mathbb{H}_{A},\mathbb{H}_{B})\equiv\left(\int_{\mathcal{D}}|\rho_{\Pi(\mathbb{F}_{\mathbf{Y}|X})(y,\delta|x_{A})}^{'}(\Pi(\mathbb{H}_{A})(y,\delta|x_{A}))-\rho_{\Pi(\mathbb{F}_{\mathbf{Y}|X})(y,\delta|x_{B})}^{'}(\Pi(\mathbb{H}_{B})(y,\delta|x_{B}))|^{p}d\mu\right)^{1/p}.
\]
 If Assumptions \ref{assu:boot_weight} and \ref{assu:est_der_tuning}
additionally hold, then, conditional on data, 
\[
\hat{\theta}_{L,P,e}^{'}\left(\mathbb{F}_{\mathbf{Y}|X,n}^{*}(\cdot|x_{A}),\mathbb{F}_{\mathbf{Y}|X,n}^{*}(\cdot|x_{B})\right)\Rightarrow\theta_{L,P,e}^{'}\left(\mathbb{G}_{(x_{A},x_{B})}((\cdot,\cdot,\cdot,\cdot))\right),
\]
 where $\hat{\rho}_{f}^{'}(h)(\delta|x)\equiv\sup_{y\in\Lambda_{f}(\delta,a_{n};S^{+})}h(y,\delta|x)$
and 
\[
\hat{\theta}_{L,P,e}^{'}\left(\mathbb{H}_{A},\mathbb{H}_{B}\right)\equiv\left(\int_{\mathcal{D}}|\hat{\rho}_{\Pi(\hat{\mathbb{F}}_{\mathbf{Y}|X,n})(y,\delta|x_{A})}^{'}(\Pi(\mathbb{H}_{A})(y,\delta|x_{A}))-\hat{\rho}_{\Pi(\hat{\mathbb{F}}_{\mathbf{Y}|X,n})(y,\delta|x_{B})}^{'}(\Pi(\mathbb{H}_{B})(y,\delta|x_{B}))|^{p}d\mu\right)^{1/p}.
\]
\end{theorem} 

Theorem \ref{thm:test_equality} establishes the limiting distribution
of the normalized test statistic and the validity of the bootstrap.
Although it is interesting to investigate the performance of this
uniform test, it is beyond the scope of this paper. Therefore, we
leave this interesting topic for future research. 

\section{\label{sec:Proofs} Mathematical Proofs }

\subsection{Proof of Lemma \ref{lem:id_unconfounded} }

\begin{proof} We only prove the first result, as the second result
can be proven in a similar way. Observe that 
\begin{align*}
\mathbb{E}\left[D\cdot G(Y)|X=x\right] & =\mathbb{E}\left[D\cdot G(Y)|X=x,D=1\right]\Pr(D=1|X=x)\\
 & =\mathbb{E}\left[G(Y_{1})|X=x,D=1\right]\Pr(D=1|X=x)\\
 & =\mathbb{E}[G(Y_{1})|X=x]\cdot\Pr(D=1|X=x).
\end{align*}
 where the last holds by Assumption \ref{assu:unconfounded}. In addition,
Assumption \ref{assu:bound_propensity} implies that $\Pr(D=1|X=x)>0$.
Therefore, 
\[
\frac{\mathbb{E}[D\cdot G(Y)|X=x]}{\mathbb{E}[D|X=x]}=\frac{\mathbb{E}[G(Y_{1})|X=x]\cdot\Pr(D=1|X=x)}{\Pr(D=1|X=x)}=\mathbb{E}[G(Y_{1})|X=x],
\]
 and this completes the proof. \end{proof}

\subsection{Proof of Theorem \ref{thm:ID_Endo_FH_Bound}}

\begin{proof} To see this, recall that from Proposition \ref{prop:ID_dist_TE_Unconfounded},
we have 
\[
\max\left(\sup_{y}\left\{ F_{1|X}(y|x)-F_{0|X}(y-\delta|x)\right\} ,0\right)\leq F_{\Delta|X}(\delta|x)\leq\min\left(\inf_{y}\left\{ F_{1|X}(y|x)-F_{0|X}(y-\delta|x)\right\} ,0\right)+1.
\]
 Since the conditional distributions of the potential outcomes given
$X=x$ are partially identified by the hypothesis and the functions
$\max[\cdot,\cdot]$ and $\min[\cdot,\cdot]$ are non-decreasing,
we obtain that 
\[
F_{\Delta|X}^{e,L}(\delta|x)\leq\max\left(\sup_{y}\left\{ F_{1|X}(y|x)-F_{0|X}(y-\delta|x)\right\} ,0\right)
\]
and that 
\[
F_{\Delta|X}^{e,U}(\delta|x)\geq\min\left(\inf_{y}\left\{ F_{1|X}(y|x)-F_{0|X}(y-\delta|x)\right\} ,0\right)+1.
\]
 Replacing the bounds in Theorem \ref{prop:ID_dist_TE_Unconfounded}
with $F_{\Delta|X}^{e,L}(\delta|x)$ and $F_{\Delta|X}^{e,U}(\delta|x)$
yields the result. \end{proof}

\subsection{Proof of Theorem \ref{thm:marginal_ID_FSD1} }

\begin{proof} To prove the results in Theorem \ref{thm:marginal_ID_FSD1},
recall that $F_{1|X}(y|x)=F_{1|1X}(y|x)\Pr(D=1|X=x)+F_{1|0X}(y|x)\Pr(D=0|X=x)$.
Since $F_{1|0X}(y|x)\geq F_{1|1X}(y|x)$ for all $y$, replacing $F_{1|0X}(y|x)$
with $F_{1|1X}(y|x)$ in the above decomposition of $F_{1|X}(y|x)$
results in that $F_{1|X}^{L,FSD1}(y|x)=\Pr(Y\leq y|D=1,X=x)=\Pr(Y_{1}\leq y|D=1,X=x)$
and that $F_{1|X}(y|x)=\Pr(Y_{1}\leq y|X=x)\geq F_{1|X}^{L,FSD1}(y|x)$.
Note that $F_{1|X}^{U,FSD1}(y|x)$ and $F_{0|X}^{L,FSD1}(y|x)$ are
identical to $F_{1|X}^{U}(y|x)$ and $F_{0|X}^{L}(y|x)$ in (\ref{eq:manski_bound})
and hence, they are valid. Since $F_{0|X}(y|x)=F_{0|X}^{L,FSD1}(y|x)+\Pr(D=1|X=x)\cdot F_{0|1X}(y|x)$,
$\Pr(Y_{0}\leq y|D=1,X=x)\leq\Pr(Y_{0}\leq y|D=0,X=x)$ and $\Pr(Y\leq y|D=0,X=x)=\Pr(Y_{0}\leq y|D=0,X=x)$,
we obtain that $F_{0|X}(y|x)\leq\Pr(Y_{0}\leq y|D=0,X=x)=\Pr(Y\leq y|D=0,X=x)=F_{0|X}^{U,FSD1}(y|x)$.
\end{proof}

\subsection{Proof of Theorem \ref{thm:marginal_ID_FSD2} }

\begin{proof} The proof proceeds as follows. Since we have, for all
$y\in\mathbb{R}$, $\Pr(Y_{1}\leq y|D=0,X=x)\leq\Pr(Y_{0}\leq y|D=0,X=x)$,
we obtain the upper bound on $F_{1|X}(y|x)$ by replacing $\Pr(Y_{1}\leq y|D=0,X=x)$
with $\Pr(Y_{0}\leq y|D=0,X=x)$. Similarly, we have $\Pr(Y_{1}\leq y|D=1,X=x)\leq\Pr(Y_{0}\leq y|D=1,X=x)$,
and this inequality is used to establish the lower bound on $F_{0|X}(y|x)$.
Note that $\Pr(Y\leq y|D=d,X=x)=\Pr(Y_{d}\leq y|D=d,X=x)$ for $d\in\{0,1\}$
and that $\Pr(Y\leq y|D=1,X=x)\Pr(D=1|X=x)+\Pr(Y\leq y|D=0,X=x)\Pr(D=0|X=x)=\Pr(Y\leq y|X=x)$,
and this completes the proof. \end{proof}

\subsection{Proof of Corollary \ref{coro:marginal_ID_FSD12}}

\begin{proof} The results in Corollary \ref{coro:marginal_ID_FSD12}
are directly implied by Theorems \ref{thm:marginal_ID_FSD1} and \ref{thm:marginal_ID_FSD2}.
Recall that 
\begin{eqnarray*}
F_{1|X}^{U,FSD2}(y|x) & = & F_{1|1X}(y|x)\Pr(D=1|X=x)+F_{0|0X}(y|x)\Pr(D=0|X=x)\\
 & \leq & F_{1|1X}(y|x)\Pr(D=1|X=x)+\Pr(D=0|X=x)=F_{1|X}^{U,FSD1}(y|x)
\end{eqnarray*}
and 
\[
F_{1|X}^{L,FSD1}(y|x)=\Pr(Y\leq y|D=1,X=x)\geq F_{1|1X}(y|x)\Pr(D=1|X=x)=F_{1|X}^{L,FSD2}(y|x),
\]
 and therefore these establish the bounds on $F_{1|X}(y|x)$. One
can use a similar argument for the bounds on $F_{0|X}(y|x)$. \end{proof}

For any real sequences $(a_{n})$ and $(b_{n})$, $a_{n}\lesssim b_{n}$
means that there is a constant $C$, not depending on $n$, such that
$|a_{n}|\leq C\cdot|b_{n}|$ for all $n\in\mathbb{N}$. For a set
$A\subseteq\mathbb{R}$, $l^{\infty}(A)$ denotes the set of uniformly
bounded functions on the set $A$. Let $(Z_{n})$ be a sequence of
random variables and $Z$ be a random variable. $Z_{n}\overset{p}{\rightarrow}Z$
means that $Z_{n}$ converges in probability to $Z$. We also denote
the weak convergence of $Z_{n}$ to $Z$ by $Z_{n}\Rightarrow Z$.
We abbreviate Vapnik-\v{C}ervonenkis to VC. For a class of functions,
$\mathcal{F}$, a probability measure $Q$ and $\epsilon>0$, $N(\epsilon,\mathcal{F},L_{r}(Q))$
denotes the covering number which is the minimum number of $L_{r}(Q)$
$\epsilon$-balls that cover $\mathcal{F}$, where $r\geq1$. For
a pseudo-metric space $(T,\rho)$, the diameter of $T$ is $\sup_{s,t\in T}\rho(s,t)$.
Throughout this section, I sometimes use the same notation but for
some possibly different object. 

\subsection{Proof of Theorem \ref{thm:asymp_marginals} }

We first provide the forms of covariance kernels: 
\begin{align*}
H_{1,x}(s,t) & \equiv\int K^{2}(u)du\left\{ F_{1|X}(\min(s,t)|x)-F_{1|X}(s|x)\cdot F_{1|X}(t|x)\right\} \frac{1}{p_{0}(x)\cdot f_{X}(x)},\\
H_{0,x}(s,t) & \equiv\int K^{2}(u)du\left\{ F_{0|X}(\min(s,t)|x)-F_{0|X}(s|x)\cdot F_{0|X}(t|x)\right\} \frac{1}{(1-p_{0}(x))\cdot f_{X}(x)}.
\end{align*}

\begin{proof} I only consider the weak convergence of $\sqrt{nh_{n}^{d_{x}}}(\hat{F}_{1|X,n}(\cdot|x)-F_{1|X}(\cdot|x))$,
and for the other objects one can use a similar argument. Note that 

\begin{align*}
\hat{F}_{1|X,n}(y|x)-F_{1|X}(y|x) & =\frac{\sum_{i}^{n}\mathbf{1}(Y_{i}\leq y)\cdot D_{i}\cdot K(\frac{X_{i}-x}{h_{n}})}{\sum_{i}^{n}D_{i}K(\frac{X_{i}-x}{h_{n}})}-F_{1|X}(y|x)\\
 & =\frac{1}{p_{0}(x)\cdot f_{X}(x)+o_{p}(1)}\cdot\frac{1}{nh_{n}^{d_{x}}}\sum_{i}^{n}\left[\{\mathbf{1}(Y_{i}\leq y)-F_{1|X}(y|x)\}D_{i}\cdot K(\frac{X_{i}-x}{h_{n}})\right]\\
 & =\frac{1}{p_{0}(x)\cdot f_{X}(x)+o_{p}(1)}\cdot\frac{1}{nh_{n}^{d_{x}}}\sum_{i}^{n}\left[\{\mathbf{1}(Y_{i}\leq y)-F_{1|X}(y|X_{i})\}D_{i}\cdot K(\frac{X_{i}-x}{h_{n}})\right]\\
 & \ +\frac{1}{p_{0}(x)\cdot f_{X}(x)+o_{p}(1)}\cdot\frac{1}{nh_{n}^{d_{x}}}\sum_{i}^{n}\left[\{F_{1|X}(y|X_{i})-F_{1|X}(y|x)\}D_{i}\cdot K(\frac{X_{i}-x}{h_{n}})\right]
\end{align*}
 under Assumption \ref{assu:dist_X}, \ref{assu:smoothness}, \ref{assu:kernel},
and \ref{assu:bandwidth}. To show the weak convergence, I verify
the conditions for the functional central limit theorem in \citet[Theorem 10.6]{pollard1990empirical}.
We first consider the second term in the above equation. Note that
\begin{align*}
 & \frac{1}{nh_{n}^{d_{x}}}\sum_{i}^{n}\mathbb{E}\left[\{F_{1|X}(y|X_{i})-F_{1|X}(y|x)\}\cdot D_{i}\cdot K(\frac{X_{i}-x}{h_{n}})\right]\\
= & \frac{1}{h_{n}^{d_{x}}}\int\{F_{1|X}(y|t)-F_{1|X}(y|x)\}p_{0}(t)K(\frac{t-x}{h_{n}})f_{X}(t)dt\\
= & \int\left\{ F_{1|X}^{(1)}(y|x)h_{n}\sum_{j=1}^{d_{x}}u_{j}+\frac{1}{2}F_{1|X}^{(2)}(y|\dot{x})h_{n}^{2}\sum_{j=1}^{d_{x}}u_{j}^{2}\right\} \\
 & \times\left\{ p_{0}(x)+p_{0}^{(1)}(x)h_{n}\sum_{j=1}^{d_{x}}u_{j}+\frac{1}{2}p_{0}^{(2)}(\ddot{x})h_{n}^{2}\sum_{j=1}^{d_{x}}u_{j}\right\} \\
 & \times\left\{ f_{X}(x)+f_{X}^{(1)}(x)h_{n}\sum_{j=1}^{d_{x}}u_{j}+\frac{1}{2}f_{X}^{(2)}(\tilde{x})h_{n}^{2}\sum_{j=1}^{d_{x}}u_{j}\right\} K(u)du,
\end{align*}
 where $\dot{x}$, $\ddot{x}$, and $\tilde{x}$ are some values between
$x+th_{n}$ and $x$, under Assumptions \ref{assu:dist_X}, \ref{assu:smoothness}
and \ref{assu:kernel}. Since this bias term is $O\left(h_{n}^{2}\right)$
and $\sqrt{nh_{n}^{d_{x}+4}}\rightarrow0$ by Assumption \ref{assu:bandwidth},
this term is $o\left((nh_{n}^{d_{x}})^{-1/2}\right)$. 

Now we consider the term $\frac{1}{p_{0}(x)f_{X}(x)}\cdot\frac{1}{nh_{n}^{d_{x}}}\sum_{i}^{n}\left[\{\mathbf{1}(Y_{i}\leq y)-F_{1|X}(y|X_{i})\}D_{i}\cdot K(\frac{X_{i}-x}{h_{n}})\right]$.
For all $\omega\in\Omega$, define $f_{ni}(\omega,y;x)\equiv\frac{1}{\sqrt{nh_{n}^{d_{x}}}\cdot p_{0}(x)\cdot f_{X}(x)}\{\mathbf{1}(Y_{i}\leq y)-F_{1|X}(y|X_{i})\}D_{i}\cdot K(\frac{X_{i}-x}{h_{n}})$.
Then, we have $\mathbb{E}[f_{ni}(\cdot,y;x)]=0$ by the law of iterated
expectations. Define an envelope function $F_{ni}(x)\equiv\frac{1}{\underline{p}\underline{f}\sqrt{nh_{n}^{d_{x}}}}K(\frac{X_{i}-x}{h_{n}})$,
where $\underline{p}\equiv\inf_{x\in\mathcal{X}}p_{0}(x)>0$ and $\underline{f}\equiv\inf_{x\in\mathcal{X}}f_{X}(x)>0$. 

Let $\mathcal{C}=\left\{ (-\infty,y]:y\in\mathcal{Y}\right\} $ and
define 
\begin{align*}
\mathcal{F}_{1n} & \equiv\left\{ \frac{1}{p_{0}(x)\cdot f_{X}(x)\sqrt{nh_{n}^{d_{x}}}}K(\frac{X-x}{h_{n}})\cdot D\cdot\mathbf{1}_{C}:C\in\mathcal{C}\right\} ,\\
\mathcal{F}_{2n} & \equiv\left\{ -\frac{1}{p_{0}(x)\cdot f_{X}(x)\sqrt{nh_{n}^{d_{x}}}}K(\frac{X-x}{h_{n}})\cdot D\cdot F_{1|X}(y|X):y\in\mathcal{Y}\right\} ,
\end{align*}
 and $\mathcal{F}_{n}\equiv\{f_{ni}(\omega,y;x):y\in\mathcal{Y}\}=\mathcal{F}_{1n}+\mathcal{F}_{2n}$.
The VC index of $\mathcal{C}$ is 2, and thus it follows from Theorem
9.2 and Lemma 9.9 in \citet{kosorok2008} that $N(\epsilon,\mathcal{F}_{1n},L_{2}(\mathbb{P}))\lesssim\epsilon^{-2}$
for any probability measure $\mathbb{P}$. Since $F_{1|X}(\cdot|X)$
is a monotone increasing function almost all $X$, applying Lemmas
9.9 and 9.10 in \citet{kosorok2008} implies that the VC index of
$\mathcal{F}_{2n}$ is equal to 2, and therefore $N(\epsilon,\mathcal{F}_{2n},L_{2}(\mathbb{P}))\lesssim\epsilon^{-2}$
for any probability measure $\mathbb{P}$. By the proof of Theorem
3 in \citet{andrews1994empirical}, one obtains that 
\begin{align*}
\sup_{Q}\log N\left(\epsilon,\mathcal{F}_{n},L_{2}(Q)\right) & \leq\sup_{Q}\log N\left(\frac{\epsilon}{2},\mathcal{F}_{1n},L_{2}(Q)\right)+\sup_{Q}\log N\left(\frac{\epsilon}{2},\mathcal{F}_{2n},L_{2}(Q)\right)\lesssim\log\epsilon^{-2},
\end{align*}
 where the supremum is taken over all finite probability measures
$Q$. Therefore, 
\[
\int_{0}^{1}\sqrt{\sup_{Q}\log N(\epsilon,\mathcal{F}_{n},L_{2}(Q))}d\epsilon\lesssim\int_{0}^{1}\sqrt{\log(1/\epsilon^{2})}d\epsilon=\int_{0}^{\infty}u^{1/2}\exp(-\frac{u}{2})du<\infty,
\]
which implies that $\mathcal{F}_{n}$ satisfies Pollard's entropy
condition. Since Pollard's entropy condition implies manageability
(see \citet[p.2284]{andrews1994empirical}), condition (i) of Theorem
10.6 in \citet{pollard1990empirical} is satisfied. 

Let $\mathcal{Z}_{n}(\omega,t;x)\equiv\sum_{i}^{n}f_{ni}(\omega,t;x)$
and consider $\lim_{n\rightarrow\infty}\mathbb{E}[\mathcal{Z}_{n}(\omega,s;x)\mathcal{Z}_{n}(\omega,t;x)]$
for any $s,t\in\mathcal{Y}$. Then, 
\begin{align*}
 & \lim_{n\rightarrow\infty}\mathbb{E}[\mathcal{Z}_{n}(\omega,s;x)\mathcal{Z}_{n}(\omega,t;x)]\\
= & \lim_{n\rightarrow\infty}\mathbb{E}\left[\sum_{i}^{n}f_{ni}(\omega,s;x)\cdot\sum_{j}^{n}f_{nj}(\omega,t;x)\right]\\
= & \lim_{n\rightarrow\infty}\mathbb{E}\left[\sum_{i}^{n}f_{ni}(\omega,s;x)f_{ni}(\omega,t;x)\right]\ \text{(by i.i.d assumption)}\\
= & \lim_{n\rightarrow\infty}\mathbb{E}\left[\left\{ F_{1|X}(\min(s,t)|X)-F_{1|X}(s|X)\cdot F_{1|X}(t|X)\right\} K^{2}(\frac{X-x}{h_{n}})\frac{p_{0}(X)}{p_{0}(x)^{2}\cdot f_{X}(x)^{2}}\frac{1}{h_{n}^{d_{x}}}\right]\\
= & \lim_{n\rightarrow\infty}\left[\int K^{2}(u)du\left\{ F_{1|X}(\min(s,t)|x)-F_{1|X}(s|x)\cdot F_{1|X}(t|x)\right\} \frac{1}{p_{0}(x)\cdot f_{X}(x)}+O(h_{n}^{d_{x}})\right]\\
= & \int K^{2}(u)du\left\{ F_{1|X}(\min(s,t)|x)-F_{1|X}(s|x)\cdot F_{1|X}(t|x)\right\} \frac{1}{p_{0}(x)\cdot f_{X}(x)}\\
\equiv & H_{1,x}(s,t)
\end{align*}
 is well-defined under Assumptions \ref{assu:smoothness}, \ref{assu:kernel},
and \ref{assu:bandwidth}. Therefore, condition (ii) of Theorem 10.6
in \citet{pollard1990empirical} is satisfied. 

To verify condition (iii) of Theorem 10.6 in \citet{pollard1990empirical},
we can show that 

\begin{align*}
\mathbb{E}[F_{ni}^{2}] & \lesssim\frac{1}{nh_{n}^{d_{x}}}\int K^{2}(\frac{t-x}{h_{n}})f_{X}(t)dt\\
 & \lesssim\frac{1}{nh_{n}^{d_{x}}}\int K(\frac{t-x}{h_{n}})f_{X}(t)dt\\
 & =\frac{1}{n}\left\{ f_{X}(x)+O\left(h_{n}^{2d_{x}}\right)\right\} 
\end{align*}
 by using the standard arguments for kernel estimation. Note that
the second line holds due to the kernel function being bounded. Therefore,
$\sum_{i}\mathbb{E}[F_{ni}^{2}(x)]\lesssim f_{X}(x)+O\left(h_{n}^{2d_{x}}\right)<\infty$,
which implies condition (iii) of Theorem 10.6 in \citet{pollard1990empirical}. 

Take any $\epsilon>0$, then, 
\begin{align*}
\mathbb{E}\left[F_{ni}^{2+\epsilon}\right] & \lesssim\frac{1}{\{nh_{n}^{d_{x}}\}^{1+\epsilon/2}}\int K^{2+\epsilon}(\frac{t-x}{h_{n}})f_{X}(t)dt\\
 & \lesssim\frac{1}{n}\cdot\frac{1}{(nh_{n})^{\epsilon/2}}\cdot\int u^{2}K(u)du=\frac{1}{n}\cdot o(1)
\end{align*}
 under Assumptions \ref{assu:dist_X}, \ref{assu:kernel}, and \ref{assu:bandwidth}.
Therefore, one can show that for any $\epsilon_{0}>0$, we have 
\begin{align*}
\sum_{i}\mathbb{E}\left[F_{ni}^{2}\mathbf{1}(F_{ni}>\epsilon_{0})\right]= & \sum_{i}\mathbb{E}\left[F_{ni}^{2+\epsilon}F_{ni}^{-\epsilon}\mathbf{1}(F_{ni}>\epsilon_{0})\right]\\
\leq & \sum_{i}\frac{1}{\epsilon_{0}^{\epsilon}}\mathbb{E}\left[F_{ni}^{2+\epsilon}\right]\lesssim\frac{1}{\epsilon_{0}^{\epsilon}}\cdot o(1)=o(1),
\end{align*}
 and thus condition (iv) of Theorem 10.6 in \citet{pollard1990empirical}
is satisfied. Lastly, define $\rho_{n}^{2}(s,t)\equiv\sum_{i}\mathbb{E}|f_{ni}(\cdot,s;x)-f_{ni}(\cdot,t;x)|^{2}$
for any $s,t\in\mathcal{Y}$. Without loss of generality, we assume
that $s\geq t$. Then, 
\begin{align}
 & \sum_{i}\mathbb{E}|f_{ni}(\cdot,s;x)-f_{ni}(\cdot,t;x)|^{2}\nonumber \\
= & \frac{1}{nh_{n}^{d_{x}}p_{0}(x)^{2}\cdot f_{X}(x)^{2}}\sum_{i}\mathbb{E}\Bigg|K(\frac{X_{i}-x}{h_{n}})\cdot D_{i}\cdot[\{\mathbf{1}(Y_{i}\leq s)-\mathbf{1}(Y_{i}\leq t)\}-\{F_{1|X}(s|X_{i})-F_{1|X}(t|X_{i})\}]\Bigg|^{2}\nonumber \\
= & \frac{\sum_{i}\mathbb{E}\left[p_{0}(X_{i})\cdot K^{2}(\frac{X_{i}-x}{h_{n}})\cdot\left\{ F_{1|X}(s|X_{i})+F_{1|X}(t|X_{i})-2F_{1|X}(t|X_{i})-(F_{1|X}(s|X_{i})-F_{1|X}(t|X_{i}))^{2}\right\} \right]}{p_{0}(x)^{2}\cdot f_{X}(x)^{2}\cdot nh_{n}^{d_{x}}}\nonumber \\
= & \frac{\int p_{0}(z)\cdot K^{2}(\frac{z-x}{h_{n}})\cdot\left\{ F_{1|X}(s|z)+F_{1|X}(t|z)-2F_{1|X}(t|z)-(F_{1|X}(s|z)-F_{1|X}(t|z))^{2}\right\} dz}{p_{0}(x)^{2}\cdot f_{X}(x)^{2}\cdot h_{n}^{d_{x}}}\nonumber \\
= & \frac{\int K^{2}(u)du}{p_{0}(x)\cdot f_{X}(x)}\cdot\left\{ F_{1|X}(s|x)+F_{1|X}(t|x)-2F_{1|X}(t|x)-(F_{1|X}(s|x)-F_{1|X}(t|x))^{2}\right\} +O\left(h_{n}^{d_{x}}\right)\nonumber \\
= & \frac{\int K^{2}(u)du}{p_{0}(x)\cdot f_{X}(x)}\left\{ \left(F_{1|X}(s|x)-F_{1|X}(t|x)\right)\cdot\left(1-\left(F_{1|X}(s|x)-F_{1|X}(t|x)\right)\right)\right\} +O\left(h_{n}^{d_{x}}\right),\label{eq:FCLT_rho_n}
\end{align}
 where the term $O\left(h_{n}^{d_{x}}\right)$ holds uniformly in
$s$ and $t$. Letting $\rho(s,t)\equiv\lim_{n\rightarrow\infty}\rho_{n}(s,t)$,
we have 
\[
\rho^{2}(s,t)=\frac{\int K^{2}(u)du}{p_{0}(x)\cdot f_{X}(x)}\left\{ \left(F_{1|X}(s|x)-F_{1|X}(t|x)\right)\cdot\left(1-\left(F_{1|X}(s|x)-F_{1|X}(t|x)\right)\right)\right\} 
\]
 for $s\geq t$. Take any $\{s_{n}\}$ and $\{t_{n}\}$ such that
$\rho(s_{n},t_{n})\rightarrow0$. Since the term $O\left(h_{n}^{d_{x}}\right)$
in (\ref{eq:FCLT_rho_n}) holds uniformly in $s$ and $t$ and the
conditional distribution is continuous, $\rho_{n}(s_{n},t_{n})\rightarrow0$,
and this establishes that condition (v) of Theorem 10.6 in \citet{pollard1990empirical}
is met. 

In all, all conditions of Theorem 10.6 in \citet{pollard1990empirical}
are satisfied, and thus applying Theorem 10.6 in \citet{pollard1990empirical}
results in that $\sqrt{nh_{n}^{d_{x}}}\left(\hat{F}_{1|X,n}(\cdot|x)-F_{1|X}(\cdot|x)\right)\Rightarrow\mathbb{G}_{1,x}(\cdot)$
in $l^{\infty}(\mathcal{Y})$, where $\mathbb{G}_{1,x}$ is a zero
mean Gaussian process with covariance kernel $H_{1,x}(\cdot,\cdot)$. 

Similarly, we can show that $\sqrt{nh_{n}^{d_{x}}}\left(\hat{F}_{0|X,n}(\cdot|x)-F_{0|X}(\cdot|x)\right)\Rightarrow\mathbb{G}_{0,x}(\cdot)$
in $l^{\infty}(\mathcal{Y})$ with covariance kernel 
\begin{align*}
H_{0,x}(s,t) & \equiv\int K^{2}(u)du\left\{ F_{0|X}(\min(s,t)|x)-F_{0|X}(s|x)\cdot F_{0|X}(t|x)\right\} \frac{1}{(1-p_{0}(x))\cdot f_{X}(x)}
\end{align*}
 Since each component weakly converges to a Gaussian process in $l^{\infty}(\mathcal{Y})$,
the vector consisting of these components also weakly converges to
$\mathbb{G}_{x}(\cdot)$ (cf. \citet[p.270]{van1998asymptotic}),
and this ends the proof. \end{proof}

\subsection{Proof of Theorem \ref{thm:Asymptotic_Dist of Derivative} }

\begin{proof} The weak convergence is a direct result of Theorem
2.1 of \citet{fang2019inference}. Specifically, Theorem \ref{thm:asymp_marginals}
implies that Assumptions 1 and 2 of \citet{fang2019inference} hold
with $\phi_{L}(\cdot)$ and $\phi_{U}(\cdot)$, and therefore this
ends the proof. \end{proof}

\subsection{Proof of Theorem \ref{thm:bootstrap}}

We verify the sufficient conditions for Theorem 4.3 of \citet{firpo2021uniform}
to prove Theorem \ref{thm:bootstrap}. To this end, we first provide
useful lemmas. 

\begin{lemma}\label{lem:uniform_weak} Suppose that all of the conditions
in Theorem \ref{thm:asymp_marginals} hold and let $x\in int(\mathcal{X})$
be fixed. Then, the weak convergence presented in Theorem \ref{thm:asymp_marginals}
is uniform in underlying probability measures. \end{lemma} 

\begin{proof} We have already shown that $\int_{0}^{1}\sqrt{\sup_{Q}\log N(\epsilon,\mathcal{F}_{n},L_{2}(Q))}d\epsilon<\infty$
in the proof Theorem \ref{thm:asymp_marginals}. Since the diameter
of $\mathcal{F}_{n}$ is finite for every $n$, this implies that
\[
\int_{0}^{\infty}\sqrt{\sup_{Q}\log N\left(\epsilon,\mathcal{F}_{n},L_{2}(Q)\right)}d\epsilon<\infty.
\]
In addition, the condition on the envelope function that is assumed
in Theorem 2.8.3 of \citet{vvw1996} holds. Therefore, the weak convergence
is uniform in the underlying probability measure by Theorem 2.8.3
of \citet{vvw1996}. \end{proof}

The following lemma establishes the validity of the bootstrap for
the kernel estimators of the conditional distributions of the potential
outcomes given $X$. 

\begin{lemma}\label{lem:Asymptotic_Marginal_Bootstrap} Suppose that
all of the conditions in Theorem \ref{thm:bootstrap} hold. Let $x\in int(\mathcal{X})$
be fixed. Then, we have, conditional on the data, 
\[
\sqrt{nh_{n}^{d_{x}}}\mathbb{\hat{F}}_{\mathbf{Y}|X,n}^{*}(\cdot|x)\Rightarrow\mathbb{G}_{x}(\cdot)
\]
 in $\left(l^{\infty}(\mathcal{Y})\right)^{4}$. \end{lemma} 

\begin{proof} I only consider $\frac{1}{\sqrt{nh_{n}^{d_{x}}}}\sum_{i}B_{i}\hat{\psi}_{1,i}(\cdot|x)\Rightarrow\mathbb{G}_{1,x}(\cdot)$
conditional on the data, and one can prove the results for other components
by using a similar argument. In this proof, we let $\mathbb{E}_{B}[\cdot]$
denote the conditional expectation operator on the data. Recall that
\begin{align*}
\frac{1}{\sqrt{nh_{n}^{d_{x}}}}\sum_{i}B_{i}\cdot\hat{\psi}_{1,i}(y|x) & =\frac{\frac{1}{\sqrt{nh_{n}^{d_{x}}}}\sum_{i}B_{i}\{\mathbf{1}(Y_{i}\leq y)-\hat{F}_{1|X,n}(y|x)\}\cdot D_{i}\cdot K(\frac{X_{i}-x}{h_{n}})}{\frac{1}{nh_{n}^{d_{x}}}\sum_{i}^{n}K(\frac{X_{i}-x}{h_{n}})}\\
 & =\frac{1}{\hat{p}_{n}(x)\cdot\hat{f}_{X}(x)}\frac{1}{\sqrt{nh_{n}^{d_{x}}}}\sum_{i}B_{i}\{\mathbf{1}(Y_{i}\leq y)-\hat{F}_{1|X,n}(y|x)\}\cdot D_{i}\cdot K(\frac{X_{i}-x}{h_{n}})\\
 & =\frac{1}{\hat{p}_{n}(x)\cdot\hat{f}_{X}(x)}\frac{1}{\sqrt{nh_{n}^{d_{x}}}}\sum_{i}B_{i}\{\mathbf{1}(Y_{i}\leq y)-F_{1|X}(y|x)\}\cdot D_{i}\cdot K(\frac{X_{i}-x}{h_{n}})\\
 & +\frac{1}{\hat{p}_{n}(x)\cdot\hat{f}_{X}(x)}\frac{1}{\sqrt{nh_{n}^{d_{x}}}}\sum_{i}B_{i}\{F_{1|X}(y|x)-\hat{F}_{1|X,n}(y|x)\}\cdot D_{i}\cdot K(\frac{X_{i}-x}{h_{n}})
\end{align*}
 where $\hat{f}_{X}(x)=\frac{1}{nh_{n}^{d_{x}}}\sum_{i}^{n}K(\frac{X_{i}-x}{h_{n}})$.
I first show that $\frac{1}{\hat{p}_{n}(x)\cdot\hat{f}_{X}(x)}\frac{1}{\sqrt{nh_{n}^{d_{x}}}}\sum_{i}B_{i}\{F_{1|X}(y|x)-\hat{F}_{1|X,n}(y|x)\}\cdot D_{i}\cdot K(\frac{X_{i}-x}{h_{n}})$
converges to a zero process conditional on sample path by verifying
conditions (i)--(v) of Theorem 10.6 in \citet{pollard1990empirical}.
Define 
\begin{align*}
f_{i}^{u}(y|x) & \equiv\frac{1}{\hat{p}_{n}(x)\hat{f}_{X}(x)\sqrt{nh_{n}^{d_{x}}}}B_{i}\{F_{1|X}(y|x)-\hat{F}_{1|X,n}(y|x)\}\cdot D_{i}\cdot K(\frac{X_{i}-x}{h_{n}}),\\
F_{i}^{u}(x) & \equiv\frac{B_{i}}{\hat{p}_{n}(x)\hat{f}_{X}(x)\sqrt{nh_{n}^{d_{x}}}}K(\frac{X_{i}-x}{h_{n}}).
\end{align*}
 Since $\hat{F}_{1|X,n}(\cdot|x)$ and $F_{1|X}(\cdot|x)$ are monotone
increasing functions and bounded by 1, $\{\hat{F}_{1|X,n}(y|x):y\in\mathcal{Y}\}$
and $\{F_{1|X}(y|x):y\in\mathcal{Y}\}$ satisfy the Pollard's entropy
condition by the same argument in the proof of Theorem \ref{thm:asymp_marginals}.
Therefore, these classes are manageable. 

For the second condition of Theorem 10.6 in \citet{pollard1990empirical},
note that $\mathbb{E}_{B}[f_{i}^{u}(y|x)]=0$ since $\mathbb{E}[B]=0$
and $B$ is independent of data. In consequence, we have, for any
$y_{1},y_{2}\in\mathcal{Y}$, 
\begin{align*}
 & \Bigg|\mathbb{E}_{B}\left[\sum_{i}^{n}f_{i}^{u}(y_{1}|x)\sum_{j}^{n}f_{j}^{u}(y_{2}|x)\right]\Bigg|\\
= & \Bigg|\mathbb{E}_{B}\left[\sum_{i}^{n}f_{i}^{u}(y_{1}|x)\cdot f_{i}^{u}(y_{2}|x)\right]\Bigg|\\
\leq & \frac{1}{\left(\hat{p}_{n}(x)\hat{f}_{X}(x)\right)^{2}}\sup_{y}\Big|F_{1|X}(y|x)-\hat{F}_{1|X,n}(y|x)\Big|^{2}\cdot\frac{1}{nh_{n}^{d_{x}}}\sum_{i}^{n}K^{2}(\frac{X_{i}-x}{h_{n}}),
\end{align*}
 conditional on sample path. We have $\sqrt{nh_{n}^{d_{x}}}\left(F_{1|X}(\cdot|x)-\hat{F}_{1|X,n}(\cdot|x)\right)\Rightarrow\mathbb{G}_{1,x}(\cdot)$
in $l^{\infty}(\mathcal{Y})$ from Theorem \ref{thm:asymp_marginals},
and therefore $\sup_{y}|F_{1|X}(y|x)-\hat{F}_{1|X,n}(y|x)|\overset{p}{\rightarrow}0$.
We also have $\hat{f}_{X}(x)\overset{p}{\rightarrow}f_{X}(x)>0$,
$\hat{p}_{n}(x)\overset{p}{\rightarrow}p_{0}(x)$, and $\frac{1}{nh_{n}^{d_{x}}}\sum_{i}^{n}K^{2}(\frac{X_{i}-x}{h_{n}})\overset{p}{\rightarrow}\int K^{2}(u)du$
under the set of assumptions. In all, we obtain that 
\[
H_{1,x}^{B}(y_{1},y_{2})\equiv\mathbb{E}_{B}\left[\sum_{i}^{n}f_{i}^{u}(y_{1}|x)\sum_{j}^{n}f_{j}^{u}(y_{2}|x)\right]\overset{p}{\Rightarrow}0,
\]
 and thus condition (ii) of Theorem 10.6 in \citet{pollard1990empirical}
is satisfied with $H_{1}^{B}(y_{1},y_{2})$. 

Recall that 
\[
\sum_{i}^{n}\mathbb{E}_{B}\left[F_{i}^{u}(x)^{2}\right]=\frac{1}{\left(\hat{p}_{n}(x)\hat{f}_{X}(x)\right)^{2}}\frac{1}{nh_{n}^{d_{x}}}\sum_{i}^{n}K^{2}(\frac{X_{i}-x}{h_{n}})
\]
 conditional on sample path. Therefore, $\sum_{i}^{n}\mathbb{E}_{B}[F_{i}^{u}(x)^{2}]\overset{p}{\rightarrow}\frac{1}{\left(p_{0}(x)f_{X}(x)\right)^{2}}\int K^{2}(u)du<\infty$,
which implies that condition (iii) is satisfied. In addition, it can
be easily shown that 
\begin{align*}
 & \sum_{i}^{n}\mathbb{E}_{B}\left[F_{i}^{u}(x)^{2}\cdot\mathbf{1}(F_{i}^{u}(x)>\epsilon)\right]\\
= & \sum_{i}^{n}\frac{1}{\left(\hat{p}_{n}(x)\hat{f}_{X}(x)\right)^{2}}\frac{1}{nh_{n}^{d_{x}}}K^{2}(\frac{X_{i}-x}{h_{n}})\cdot\mathbf{1}\left(\frac{1}{\left(\hat{p}_{n}(x)\hat{f}_{X}(x)\right)^{2}}\frac{1}{nh_{n}^{d_{x}}}K^{2}(\frac{X_{i}-x}{h_{n}})>\epsilon^{2}\right)\\
\overset{p}{\rightarrow} & 0
\end{align*}
 for each $\epsilon>0$, so condition (iv) is also satisfied. 

Lastly, note that for any $s,t\in\mathcal{Y}$ such that $s\geq t$,
\begin{align*}
\rho_{n}^{u}(s,t)^{2} & \equiv\sum_{i}\mathbb{E}_{B}|f_{i}^{u}(s|x)-f_{i}^{u}(t|x)|^{2}\\
 & =\left\{ (F_{1|X}(s|x)-\hat{F}_{1|X,n}(s|x))-(F_{1|X}(t|x)-\hat{F}_{1|X,n}(t|x))\right\} ^{2}\\
 & \ \ \ \times\frac{1}{\left(\hat{p}_{n}(x)\hat{f}_{X}(x)\right)^{2}}\frac{1}{nh_{n}^{d_{x}}}\sum_{i}^{n}D_{i}K^{2}(\frac{X_{i}-x}{h_{n}})\\
 & \overset{p}{\Rightarrow}\frac{\int K^{2}(u)du}{p_{0}(x)\cdot f_{X}(x)}\left\{ \left(F_{1|X}(s|x)-F_{1|X}(t|x)\right)\cdot\left(1-\left(F_{1|X}(s|x)-F_{1|X}(t|x)\right)\right)\right\} \\
 & \equiv\rho^{u}(s,t)
\end{align*}
 by the same argument above. Hence, $\rho^{u}(s,t)$ is well-defined,
and if $\rho^{u}(s_{n},t_{n})\rightarrow0$, then $\rho_{n}^{u}(s_{n},t_{n})\rightarrow0$.
Therefore, condition (v) of Theorem 10.6 in \citet{pollard1990empirical}
is met. Applying Theorem 10.6 in \citet{pollard1990empirical} leads
to that, conditional on data, $\frac{1}{\hat{p}_{n}(x)\cdot\hat{f}_{X}(x)}\frac{1}{\sqrt{nh_{n}^{d_{x}}}}\sum_{i}B_{i}\{F_{1|X}(\cdot|x)-\hat{F}_{1|X,n}(\cdot|x)\}\cdot D_{i}\cdot K(\frac{X_{i}-x}{h_{n}})$
converges to a zero process in $l^{\infty}(\mathcal{Y})$. 

By applying Theorem 10.4 in \citet{kosorok2008} and Theorem \ref{thm:asymp_marginals},
we have 
\[
\frac{1}{\hat{p}_{n}(x)\cdot\hat{f}_{X}(x)}\frac{1}{\sqrt{nh_{n}^{d_{x}}}}\sum_{i}B_{i}\{\mathbf{1}(Y_{i}\leq\cdot)-F_{1|X}(\cdot|x)\}\cdot D_{i}\cdot K(\frac{X_{i}-x}{h_{n}})\Rightarrow\mathbb{G}_{1,x}(\cdot)\ \text{in }l^{\infty}(\mathcal{Y})
\]
 conditional on sample path. Therefore, 
\[
\frac{1}{\sqrt{nh_{n}^{d_{x}}}}\sum_{i}B_{i}\hat{\psi}_{1,i}(\cdot|x)\Rightarrow\mathbb{G}_{1,x}(\cdot)\ \text{in }l^{\infty}(\mathcal{Y})
\]
 conditional on the data. Using a similar argument, we can show that
\[
\frac{1}{\sqrt{nh_{n}^{d_{x}}}}\sum_{i}B_{i}\hat{\psi}_{0,i}(\cdot|x)\Rightarrow\mathbb{G}_{0,x}(\cdot)\ \text{in }l^{\infty}(\mathcal{Y})
\]
 conditional on the data. In all, $\sqrt{nh_{n}^{d_{x}}}\mathbb{\hat{F}}_{\mathbf{Y}|X,n}^{*}(\cdot|x;B)\Rightarrow\mathbb{G}_{x}(\cdot)$
conditional on the data. \end{proof}

\paragraph{Proof of Theorem \ref{thm:bootstrap} }

\begin{proof} We verify the assumption in Theorem 4.3 of \citet{firpo2021uniform}
(i.e., Assumptions A1-A3). Theorem \ref{thm:asymp_marginals} implies
Assumption A1 with $r_{n}=\sqrt{nh_{n}^{d_{x}}}$. Assumption A2 in
\citet{firpo2021uniform} is satisfied because Lemma \ref{lem:uniform_weak}
shows that the weak convergence in Theorem \ref{thm:asymp_marginals}
is uniform in underlying probability measures. Lastly, Lemma \ref{lem:Asymptotic_Marginal_Bootstrap}
establishes the bootstrap validity. Lemma \ref{lem:Asymptotic_Marginal_Bootstrap},
together with Lemma A.2 of \citet{linton2010improved}, implies that
Assumption A3 in \citet{firpo2021uniform} is met. In all, applying
Theorem 4.3 of \citet{firpo2021uniform} ends the proof. \end{proof}

\subsection{Proof of Theorem \ref{thm:endo_marginal_weak_limit} }

We first present the forms of covariance kernels in the theorem: 
\begin{align*}
H_{1,x}^{e}(s,t) & \equiv\int K^{2}(u)du\left\{ F_{1|1X}(\min(s,t)|x)-F_{1|1X}(s|x)\cdot F_{1|1X}(t|x)\right\} \frac{1}{p_{0}(x)\cdot f_{X}(x)},\\
H_{0,x}^{e}(s,t) & \equiv\int K^{2}(u)du\left\{ F_{0|0X}(\min(s,t)|x)-F_{0|0X}(s|x)\cdot F_{0|0X}(t|x)\right\} \frac{1}{(1-p_{0}(x))\cdot f_{X}(x)},\\
H_{Y}^{e}(s,t) & \equiv\int K^{2}(u)du\left\{ F_{Y|X}(\min(s,t)|x)-F_{Y|X}(s|x)\cdot F_{Y|X}(t|x)\right\} \frac{1}{f_{X}(x)}.
\end{align*}

\begin{proof} Note that 
\begin{align*}
 & \sqrt{nh_{n}^{d_{x}}}\left(\hat{F}_{1|1X,n}(y|x)-F_{1|1X}(y|x)\right)\\
 & =\frac{\frac{1}{\sqrt{nh_{n}^{d_{x}}}}\sum_{i}^{n}\{\mathbf{1}(Y_{i}\leq y)-F_{1|1X}(y|x)\}D_{i}K(\frac{X_{i}-x}{h_{n}})}{\frac{1}{nh_{n}^{d_{x}}}\sum_{i}^{n}D_{i}K(\frac{X_{i}-x}{h_{n}})}\\
 & =\frac{1}{p_{0}(x)f_{X}(x)+o_{p}(1)}\frac{1}{\sqrt{nh_{n}^{d_{x}}}}\sum_{i}^{n}\{\mathbf{1}(Y_{i}\leq y)-F_{1|1X}(y|x)\}D_{i}K(\frac{X_{i}-x}{h_{n}})\\
 & =\frac{1}{p_{0}(x)f_{X}(x)+o_{p}(1)}\frac{1}{\sqrt{nh_{n}^{d_{x}}}}\sum_{i}^{n}\{\mathbf{1}(Y_{i}\leq y)-F_{1|1X}(y|X_{i})\}D_{i}K(\frac{X_{i}-x}{h_{n}})\\
 & \ +\frac{1}{p_{0}(x)f_{X}(x)+o_{p}(1)}\frac{1}{\sqrt{nh_{n}^{d_{x}}}}\sum_{i}^{n}\{F_{1|1X}(y|X_{i})-F_{1|1X}(y|x)\}D_{i}K(\frac{X_{i}-x}{h_{n}})
\end{align*}
 Using the same argument for the proof of Theorem \ref{thm:asymp_marginals},
\[
\sqrt{nh_{n}^{d_{x}}}\left(\hat{F}_{1|1X,n}(\cdot|x)-F_{1|1X}(\cdot|x)\right)\Rightarrow\mathbb{G}_{1}^{e}(\cdot)\ \text{in }l^{\infty}(\mathcal{Y}),
\]
 where $\mathbb{G}_{1,x}^{e}(\cdot)$ is a Gaussian process with zero
mean and covariance kernel 
\[
H_{1,x}^{e}(s,t)\equiv\int K^{2}(u)du\left\{ F_{1|1X}(\min(s,t)|x)-F_{1|1X}(s|x)\cdot F_{1|1X}(t|x)\right\} \frac{1}{p_{0}(x)\cdot f_{X}(x)}.
\]
 Similarly, 
\[
\sqrt{nh_{n}^{d_{x}}}\left(\hat{F}_{0|0X,n}(\cdot|x)-F_{0|0X}(\cdot|x)\right)\Rightarrow\mathbb{G}_{0,x}^{e}(\cdot)\ \text{in }l^{\infty}(\mathcal{Y}),
\]
 where $\mathbb{G}_{0}^{e}(\cdot)$ is a Gaussian process with zero
mean and covariance kernel 
\[
H_{0,x}^{e}(s,t)\equiv\int K^{2}(u)du\left\{ F_{0|0X}(\min(s,t)|x)-F_{0|0X}(s|x)\cdot F_{0|0X}(t|x)\right\} \frac{1}{(1-p_{0}(x))\cdot f_{X}(x)}
\]

Lastly, I prove the weak convergence of that $\sqrt{nh_{n}^{d_{x}}}\left(\hat{F}_{Y|X,n}(\cdot|x)-F_{Y|X}(\cdot|x)\right)$.
Note that 
\begin{align*}
\hat{F}_{Y|X,n}(y|x)-F_{Y|X}(y|x) & =\frac{\frac{1}{nh_{n}^{d_{x}}}\sum_{i}^{n}(\mathbf{1}(Y_{i}\leq y)-F_{Y|X}(y|x))K(\frac{X_{i}-x}{h_{n}})}{f_{X}(x)+o_{p}(1)}\\
 & =\frac{\frac{1}{nh_{n}^{d_{x}}}\sum_{i}^{n}(\mathbf{1}(Y_{i}\leq y)-F_{Y|X}(y|X_{i}))K(\frac{X_{i}-x}{h_{n}})}{f_{X}(x)+o_{p}(1)}\ \\
 & \ +\frac{\frac{1}{nh_{n}^{d_{x}}}\sum_{i}^{n}(F_{Y|X}(y|X_{i})-F_{Y|X}(y|x))K(\frac{X_{i}-x}{h_{n}})}{f_{X}(x)+o_{p}(1)}.
\end{align*}
 The latter term is $o_{p}\left((nh_{n})^{-1/2}\right)$ under the
conditions in the theorem. We then focus on the first term. The Pollard's
entropy condition is satisfied for 
\[
\mathcal{F}_{n}^{e}\equiv\left\{ \frac{1}{f_{X}(x)}(\mathbf{1}(Y\leq y)-F_{Y|X}(y|X))K(\frac{X-x}{h_{n}}):y\in\mathcal{Y}\right\} 
\]
 by the same argument in the proof of Theorem \ref{thm:asymp_marginals}.
For given $\omega\in\Omega$, let 
\[
f_{ni}(\omega,y;x)\equiv\frac{1}{\sqrt{nh_{n}^{d_{x}}}f_{X}(x)}\{\mathbf{1}(Y_{i}\leq y)-F_{Y|X}(y|X_{i})\}K(\frac{X_{i}-x}{h_{n}})
\]
 and $\mathcal{Z}_{n}(\omega,y;x)\equiv\sum_{i}^{n}f_{ni}(\omega,y;x)$.
Then, for any $s,t\in\mathcal{Y}$, 
\begin{align*}
 & \lim_{n\rightarrow\infty}\mathbb{E}[\mathcal{Z}_{n}(\omega,s;x)\mathcal{Z}_{n}(\omega,t;x)]\\
= & \lim_{n\rightarrow\infty}\mathbb{E}\left[\sum_{i}^{n}f_{ni}(\omega,s;x)f_{ni}(\omega,t;x)\right]\\
= & \lim_{n\rightarrow\infty}\int K^{2}(u)du\left\{ F_{Y|X}(\min(s,t)|x)-F_{Y|X}(s|x)\cdot F_{Y|X}(t|x)+O\left(h_{n}^{d_{x}}\right)\right\} \frac{1}{f_{X}(x)}\\
= & \int K^{2}(u)du\left\{ F_{Y|X}(\min(s,t)|x)-F_{Y|X}(s|x)\cdot F_{Y|X}(t|x)\right\} \frac{1}{f_{X}(x)}\equiv H_{Y}^{e}(s,t)
\end{align*}
 by using the same argument for the proof of Theorem \ref{thm:asymp_marginals}.
The remaining conditions can be verified by the same way as before,
and thus 
\[
\sqrt{nh_{n}^{d_{x}}}\left(\hat{F}_{Y|X,n}(\cdot|x)-F_{Y|X}(\cdot|x)\right)\Rightarrow\mathbb{G}_{Y,x}^{e}(\cdot)\ \text{in }l^{\infty}(\mathcal{Y}),
\]
 where $\mathbb{G}_{Y,x}^{e}(\cdot)$ is a Gaussian process with zero
mean and covariance kernel $H_{Y,x}^{e}(s,t).$ \end{proof}

\subsection{Proof of Theorem \ref{thm:endo_bootstrap} }

\begin{proof} The only thing to show is the bootstrap validity, but
this can be shown by using the same way of the proof of Theorem \ref{thm:bootstrap},
together with Theorem \ref{thm:endo_marginal_weak_limit}. Therefore,
we omit the proof of this theorem. \end{proof} 

\subsection{Proof of Theorem \ref{thm:bootstrap-1} }

Let $x_{1}\in int(\mathcal{X}_{1})$ and define $||K_{1}||_{2}^{2}\equiv\int K_{1}^{2}(u)du$,
$G_{1|X_{1}}(y|x_{1})\equiv\mathbb{E}\left[\frac{F_{1|X}(y|X)}{p_{0}(X)}\Big|X_{1}=x_{1}\right]$,
and $G_{0|X_{1}}(y|x_{1})\equiv\mathbb{E}\left[\frac{F_{0|X}(y|X)}{1-p_{0}(X)}\Big|X_{1}=x_{1}\right]$. 

\begin{lemma}\label{lem:weak_convergence_abrevaya} Suppose that
conditions in Theorem \ref{thm:bootstrap-1} hold. Then, for any given
$x_{1}\in int(\mathcal{X}_{1})$, 
\begin{align*}
\sqrt{nh_{n}^{d_{1}}}\left(\hat{F}_{1|X_{1},n}(\cdot|x_{1})-F_{1|X_{1}}(\cdot|x_{1})\right) & \Rightarrow\tilde{\mathbb{G}}_{1}(\cdot)\ \text{in }l^{\infty}(\mathcal{Y}),\\
\sqrt{nh_{n}^{d_{1}}}\left(\hat{F}_{0|X_{1},n}(\cdot|x_{1})-F_{0|X_{1}}(\cdot|x_{1})\right) & \Rightarrow\tilde{\mathbb{G}}_{0}(\cdot)\ \text{in }l^{\infty}(\mathcal{Y}),
\end{align*}
 where $\tilde{\mathbb{G}}_{1}$ and $\tilde{\mathbb{G}}_{0}$ are
Gaussian processes with mean zero and covariance kernels 
\begin{align*}
\tilde{H}_{1}(y_{1},y_{2}) & \equiv\left\{ \min\left(G_{1|X_{1}}(y_{1}|x_{1}),G_{1|X_{1}}(y_{2}|x_{1})\right)-F_{1|X_{1}}(y_{1}|x_{1})F_{1|X_{1}}(y_{2}|x_{1})\right\} \frac{||K_{1}||_{2}^{2}}{f_{X_{1}}(x_{1})},\\
\tilde{H}_{0}(y_{1},y_{2}) & \equiv\left\{ \min\left(G_{0|X_{1}}(y_{1}|x_{1}),G_{0|X_{1}}(y_{2}|x_{1})\right)-F_{0|X_{1}}(y_{1}|x_{1})F_{0|X_{1}}(y_{2}|x_{1})\right\} \frac{||K_{1}||_{2}^{2}}{f_{X_{1}}(x_{1})},
\end{align*}
 respectively. \end{lemma}

\begin{proof} Pick any $x_{1}\in int(\mathcal{X}_{1})$ and define
\begin{align*}
\Psi^{1}(d,y,t,p) & \equiv\frac{d\cdot\mathbf{1}(y\leq t)}{p},\\
\Psi^{0}(d,y,t,p) & \equiv\frac{(1-d)\cdot\mathbf{1}(y\leq t)}{1-p}.
\end{align*}
 We also denote the first-order partial derivative of $\Psi^{j}(d,y,t,p)$
with respect to $p$ by $\Psi_{p}^{j}(d,y,t,p)$ for given $j\in\{0,1\}$.
For simplicity of notation, let $p_{0}(x)\equiv p(x;\theta_{0})$
and $\hat{p}_{n}(x)\equiv p(x;\hat{\theta}_{n})$. Under Assumptions
\ref{assu:x1_smooth}, \ref{assu:kernel1}, and \ref{assu:bandwidth1},
we have $\frac{1}{nh_{1n}^{d_{1}}}\sum_{i}K_{1}(\frac{X_{1i}-x_{1}}{h_{1n}})=f_{X_{1}}(x_{1})+o_{p}(1)$.
Hence, 
\begin{align*}
 & \sqrt{nh_{n}^{d_{1}}}\left(\hat{F}_{1|X_{1},n}(y|x_{1})-F_{1|X_{1}}(y|x_{1})\right)\\
 & =\frac{1}{\sqrt{nh_{1n}^{d_{1}}}}\sum_{i}^{n}\left\{ \frac{D_{i}\cdot\mathbf{1}(Y_{i}\leq y)}{\hat{p}_{n}(X_{i})}-F_{1|X_{1}}(y|x_{1})\right\} \cdot K_{1}(\frac{X_{1i}-x_{1}}{h_{1n}})\Big/\frac{1}{nh_{1n}^{d_{1}}}\sum_{i}^{n}K_{1}(\frac{X_{1i}-x_{1}}{h_{1n}}),\\
 & =\frac{1}{\sqrt{nh_{1n}^{d_{1}}}}\sum_{i}^{n}\left\{ \Psi^{1}(D_{i},Y_{i},y,\hat{p}_{n}(X_{i}))-F_{1|X_{1}}(y|x_{1})\right\} K_{1}(\frac{X_{1i}-x_{1}}{h_{1n}})\Big/f_{X_{1}}(x_{1})+o_{p}(1)\\
 & =\frac{1}{\sqrt{nh_{1n}^{d_{1}}}}\sum_{i}^{n}\left\{ \Psi^{1}(D_{i},Y_{i},y,\hat{p}_{n}(X_{i}))-F_{1|X_{1}}(y|X_{1i})\right\} K_{1}(\frac{X_{1i}-x_{1}}{h_{1n}})\Big/f_{X_{1}}(x_{1})\\
 & \ +\frac{1}{\sqrt{nh_{1n}^{d_{1}}}}\sum_{i}^{n}\left\{ F_{1|X_{1}}(y|X_{1i})-F_{1|X_{1}}(y|x_{1})\right\} K_{1}(\frac{X_{1i}-x_{1}}{h_{1n}})\Big/f_{X_{1}}(x_{1})+o_{p}(1).
\end{align*}
 The latter term is $o_{p}(1)$, because, by the standard argument
in the literature on kernel estimation, one can show that 
\begin{align*}
 & \mathbb{E}\left[\frac{1}{\sqrt{nh_{1n}^{d_{1}}}}\sum_{i}^{n}\left\{ F_{1|X_{1}}(y|X_{1i})-F_{1|X_{1}}(y|x_{1})\right\} K_{1}(\frac{X_{1i}-x_{1}}{h_{1n}})\right]\\
= & \sqrt{nh_{1n}^{d_{1}}}\cdot\int\{F_{1|X}(y|x_{1}+uh_{1n})-F_{1|X}(y|x_{1})\}K_{1}(u)\cdot f_{X_{1}}(x_{1}+uh_{1n})du\\
= & O\left(\sqrt{nh_{1n}^{d_{1}+4}}\right)=o(1).
\end{align*}
 Now we consider the first term. By a Taylor approximation of $\Psi^{1}(d,y,t,p)$
around at $p(X_{i};\theta_{0})$, we have 
\begin{align}
 & \frac{1}{\sqrt{nh_{1n}^{d_{1}}}}\sum_{i}^{n}\left\{ \Psi^{1}(D_{i},Y_{i},y,\hat{p}_{n}(X_{i}))-F_{1|X_{1}}(y|X_{1i})\right\} K_{1}(\frac{X_{1i}-x_{1}}{h_{1n}})\nonumber \\
 & =\frac{1}{\sqrt{nh_{1n}^{d_{1}}}}\sum_{i}^{n}\left\{ \Psi^{1}(D_{i},Y_{i},y,p_{0}(X_{i}))-F_{1|X_{1}}(y|X_{1i})\right\} K_{1}(\frac{X_{1i}-x_{1}}{h_{1n}})\label{eq:F1_Limit-1}\\
 & +\frac{1}{\sqrt{nh_{1n}^{d_{1}}}}\sum_{i}^{n}\left\{ \Psi_{p}^{1}(D_{i},Y_{i},y,\tilde{p}(X_{i}))\cdot(\hat{p}_{n}(X_{i})-p_{0}(X_{i}))\right\} K_{1}(\frac{X_{1i}-x_{1}}{h_{1n}})\nonumber 
\end{align}
 where $\tilde{p}(x)\equiv p(x;\tilde{\theta}_{n})$ and $\tilde{\theta}_{n}$
lies between $\hat{\theta}_{n}$ and $\theta_{0}$. Note that the
second term in (\ref{eq:F1_Limit-1}) is $o_{p}(1)$ uniformly in
$y$ under Assumption \ref{assu:propensity_para} because 
\begin{align*}
 & \Bigg|\frac{1}{\sqrt{nh_{1n}^{d_{1}}}}\sum_{i}^{n}\left\{ \Psi_{p}^{1}(D_{i},Y_{i},y,\tilde{p}(X_{i}))\cdot(\hat{p}_{n}(X_{i})-p_{0}(X_{i}))\right\} K_{1}(\frac{X_{1i}-x_{1}}{h_{1n}})\Bigg|\\
 & \leq\sup_{x\in\mathcal{X}}\sqrt{nh_{1n}^{d_{1}}}(\hat{p}_{n}(x)-p_{0}(x))\cdot\frac{1}{nh_{1n}^{d_{1}}}\sum_{i}\Big|\Psi_{p}^{1}(D_{i},Y_{i},y,\tilde{p}(X_{i}))\cdot K_{1}(\frac{X_{1i}-x_{1}}{h_{1n}})\Big|\\
 & =o_{p}(1)\cdot O_{p}(1)=o_{p}(1).
\end{align*}
 We now establish the limiting process (with respect to $y$) of the
first term (equation (\ref{eq:F1_Limit-1})). To this end, we verify
the conditions of the functional central limit theorem in \citet[p.53, Theorem 10.6]{pollard1990empirical}. 

For given $\omega\in\Omega$, let $f_{ni}(\omega,y)\equiv\frac{1}{f_{X_{1}}(x_{1})\sqrt{nh_{1n}^{d_{1}}}}\left\{ \Psi^{1}(D_{i},Y_{i},y,p_{0}(X_{i}))-F_{1|X_{1}}(y|X_{1i})\right\} \cdot K_{1}(\frac{X_{1i}-x_{1}}{h_{1n}})$
and $\mathcal{F}_{n}^{1}\equiv\{f_{ni}(\omega,y):y\in\mathbb{R}\}$.
By using the same argument in the proof of Theorem \ref{thm:asymp_marginals},
we can show that $\mathcal{F}_{n}^{1}$ satisfies Pollard's entropy
condition. Since Pollard's entropy condition implies manageability
(see \citet[p.2284]{andrews1994empirical}), condition (i) of Theorem
10.6 in \citet{pollard1990empirical} is satisfied. 

Let $\mathbb{Z}_{n}(\omega,y)\equiv\sum_{i}^{n}f_{ni}(\omega,y)$
and $y_{1},y_{2}\in\mathbb{R}$ be given. By the law of iterated expectations,
we have 
\begin{align*}
 & \mathbb{E}[\mathbb{Z}_{n}(\omega,y_{1})\cdot\mathbb{Z}_{n}(\omega,y_{2})]\\
= & \mathbb{E}\left[\left\{ \Psi^{1}(D_{i},Y_{i},y_{1},p_{0}(X_{i}))-F_{1|X_{1}}(y_{1}|X_{1i})\right\} \cdot\left\{ \Psi^{1}(D_{i},Y_{i},y_{2},p_{0}(X_{i}))-F_{1|X_{1}}(y_{2}|X_{1i})\right\} \cdot K_{1}^{2}(\frac{X_{1i}-x_{1}}{h_{1n}})\right]\\
 & \times\frac{1}{f_{X_{1}}(x_{1})^{2}h_{1n}^{d_{1}}}\\
= & \mathbb{E}\left[K_{1}^{2}(\frac{X_{1i}-x_{1}}{h_{1n}})\cdot\frac{1}{f_{X_{1}}(x_{1})^{2}h_{1n}^{d_{1}}}\left\{ G_{1|X_{1}}(\min(y_{1},y_{2})|X_{1i})-F_{1|X_{1}}(y_{1}|X_{1i})F_{1|X_{1}}(y_{2}|X_{1i})\right\} \right].
\end{align*}
 Under Assumptions \ref{assu:fx1}, \ref{assu:x1_smooth}, and \ref{assu:bandwidth1},
it follows that 
\[
\mathbb{E}[\mathbb{Z}_{n}(\omega,y_{1})\cdot\mathbb{Z}_{n}(\omega,y_{2})]=\left\{ G_{1|X_{1}}(\min(y_{1},y_{2})|x_{1})-F_{1|X_{1}}(y_{1}|x_{1})F_{1|X_{1}}(y_{2}|x_{1})+O\left(h_{1n}^{d_{1}}\right)\right\} \frac{||K_{1}||_{2}^{2}}{f_{X_{1}}(x_{1})}
\]
 by using the standard arguments for kernel estimators. Therefore,
\begin{align*}
\tilde{H}_{1}(y_{1},y_{2}) & \equiv\lim_{n\rightarrow\infty}\mathbb{E}[\mathbb{Z}_{n}(\omega,y_{1})\cdot\mathbb{Z}_{n}(\omega,y_{2})]\\
 & =\left\{ G_{1|X_{1}}(\min(y_{1},y_{2})|x_{1})-F_{1|X_{1}}(y_{1}|x_{1})F_{1|X_{1}}(y_{2}|x_{1})\right\} \frac{||K_{1}||_{2}^{2}}{f_{X_{1}}(x_{1})}
\end{align*}
 is well-defined, and hence condition (ii) of Theorem 10.6 in \citet{pollard1990empirical}
is satisfied. 

Let $F_{ni}\equiv\frac{1}{\underline{f_{X_{1}}}\sqrt{nh_{1n}^{d_{1}}}}K_{1}(\frac{X_{1i}-x_{1}}{h_{1n}})$
be an envelope function, where $\underline{f_{X_{1}}}\equiv\inf_{x_{1}\in\mathcal{X}_{1}}f_{X_{1}}(x_{1})>0$.
Since the kernel function is uniformly bounded and symmetric around
zero, we have, for any $n\in\mathbb{N}$,
\begin{align*}
\lim_{n}\sum_{i}^{n}\mathbb{E}[F_{ni}^{2}] & =\lim_{n}\frac{1}{\underline{f_{X_{1}}}}\frac{1}{h_{1n}^{d_{1}}}\int K_{1}^{2}(\frac{t-x_{1}}{h_{1n}})\cdot f_{X_{1}}(t)dt\\
 & \lesssim\lim_{n}\frac{1}{h_{1n}^{d_{1}}}\int K_{1}(\frac{t-x_{1}}{h_{1n}})\cdot f_{X_{1}}(t)dt=\lim_{n}\left\{ f_{X_{1}}(x_{1})+O\left(h_{1n}^{2d_{1}}\right)\right\} =f_{X_{1}}(x_{1})<\infty
\end{align*}
 under Assumptions \ref{assu:kernel1} and \ref{assu:bandwidth1}.
Therefore, condition (iii) of Theorem 10.6 in \citet{pollard1990empirical}
is satisfied. Similarly, we can show that for any $\eta>0$, 
\begin{align*}
\sum_{i}^{n}\mathbb{E}\left[F_{ni}^{2+\eta}\right] & \lesssim\sum_{i}^{n}\frac{1}{\left(nh_{1n}^{d_{1}}\right)^{1+\eta/2}}\int K_{1}^{2+\eta}(\frac{t-x_{1}}{h_{1n}})f_{X_{1}}(t)dt\\
 & \lesssim\frac{1}{(nh_{1n})^{\eta/2}}\cdot\int u^{2}K_{1}(u)du=o(1)\cdot O(1)=o(1).
\end{align*}
 Hence, for any $\epsilon>0$, 
\[
\sum_{i}\mathbb{E}\left[F_{ni}^{2}\mathbf{1}(F_{ni}>\epsilon)\right]=\sum_{i}\mathbb{E}\left[F_{ni}^{2+\eta}F_{ni}^{-\eta}\mathbf{1}(F_{ni}>\epsilon)\right]\leq\frac{1}{\epsilon^{\eta}}\sum_{i}\mathbb{E}\left[F_{ni}^{2+\eta}\right]\lesssim\frac{1}{\epsilon^{\eta}}\cdot o(1)=o(1),
\]
 which implies condition (iv) of Theorem 10.6 in \citet{pollard1990empirical}
is met. 

For $y_{1},y_{2}\in\mathbb{R}$, define $\rho_{n}(y_{1},y_{2})\equiv\left(\sum_{i}^{n}\mathbb{E}[|f_{ni}(\cdot,y_{1})-f_{ni}(\cdot,y_{2})|^{2}]\right)^{1/2}$
and $\rho(y_{1},y_{2})\equiv\lim_{n}\rho_{n}(y_{1},y_{2})$. Then,
\begin{align*}
 & \rho_{n}(y_{1},y_{2})^{2}\\
= & \Big\{(G_{1|X_{1}}(y_{1}|x_{1})-F_{1|X_{1}}(y_{1}|x_{1})^{2})+(G_{1|X_{1}}(y_{2}|x_{1})-F_{1|X_{1}}(y_{2}|x_{1})^{2})\\
 & +G_{1|X_{1}}(\min(y_{1},y_{2})|x_{1})-F_{1|X_{1}}(y_{1}|x_{1})F_{1|X_{1}}(y_{2}|x_{1})+o(1)\Big\}\cdot\frac{||K_{1}||_{2}^{2}}{f_{X_{1}}(x_{1})}.
\end{align*}
 Therefore, $\rho(y_{1},y_{2})$ is well-defined for any $y_{1},y_{2}\in\mathbb{R}$.
Since the components in $\rho(y_{1},y_{2})$ are continuous in $(y_{1},y_{2})$,
this leads to condition (v) of Theorem 10.6 in \citet{pollard1990empirical}.

In all, by Theorem 10.6 in \citet{pollard1990empirical}, we have
\[
\frac{1}{f_{X_{1}}(x_{1})\sqrt{nh_{1n}^{d_{1}}}}\sum_{i}^{n}\left\{ \Psi^{1}(D_{i},Y_{i},\cdot,p_{0}(X_{i}))-F_{1|X_{1}}(\cdot|x_{1})\right\} \cdot K_{1}(\frac{X_{1i}-x_{1}}{h_{1n}})\Rightarrow\tilde{\mathbb{G}}_{1}(\cdot)\ \text{in }l^{\infty}(\mathcal{Y}),
\]
 where $\tilde{\mathbb{G}}_{1}(\cdot)$ is a mean zero Gaussian process
with covariance kernel $\tilde{H}_{1}(y_{1},y_{2})$. One can prove
the weak convergence of $\sqrt{nh_{1n}^{d_{1}}}\left(\hat{F}_{0|X_{1},n}(\cdot|x_{1})-F_{0|X_{1}}(\cdot|x_{1})\right)$
by a similar way, and the covariance kernel is 
\[
\tilde{H}_{0}(y_{1},y_{2})\equiv\left\{ G_{0|X_{1}}(\min(y_{1},y_{2})|x_{1})-F_{0|X_{1}}(y_{1}|x_{1})F_{0|X_{1}}(y_{2}|x_{1})\right\} \frac{||K_{1}||_{2}^{2}}{f_{X_{1}}(x_{1})}.
\]
Since the Cartesian product of two Donsker classes is Donsker by \citet[p.270]{van1998asymptotic},
we obtain the result in the theorem. \end{proof}

\begin{lemma}\label{lem:bootstrap_abrevaya} Let $x_{1}\in int(\mathcal{X}_{1})$
be given. Suppose that the conditions in Theorem \ref{thm:bootstrap-1}
hold. Then, 
\begin{align*}
\sqrt{nh_{1n}^{d_{1}}}F_{1|X_{1},n}^{*}(\cdot|x_{1}) & \Rightarrow\tilde{\mathbb{G}}_{1}(\cdot),\\
\sqrt{nh_{1n}^{d_{1}}}F_{0|X_{1},n}^{*}(\cdot|x_{1}) & \Rightarrow\tilde{\mathbb{G}}_{0}(\cdot),
\end{align*}
 conditional on data in $l^{\infty}(\mathcal{Y})$. \end{lemma}

\begin{proof} One can use the similar arguments of the proof of Lemma
\ref{lem:Asymptotic_Marginal_Bootstrap}, together with Lemma \ref{lem:weak_convergence_abrevaya},
and thus we omit the proof. \end{proof} 

\paragraph{Proof of Theorem \ref{thm:bootstrap-1} }

\begin{proof} Note that the functionals $\phi_{L}$ and $\phi_{U}$
are Hadamard directionally differentiable. Therefore, the weak convergence
result follows from, together with Lemma \ref{lem:weak_convergence_abrevaya},
the same argument for the proof of Theorem \ref{thm:Asymptotic_Dist of Derivative}.
The bootstrap validity can be shown by using the same argument of
the proof of Theorem \ref{thm:bootstrap} with letting $r_{n}=\sqrt{nh_{1n}^{d_{1}}}$,
the weak convergence result, and Lemma \ref{lem:bootstrap_abrevaya}.
\end{proof}

\subsection{Proof of Theorem \ref{thm:test_equality} }

\begin{proof} Under $H_{0}$, $\theta_{L,p,e}=0$. Therefore, extending
part 3 of Corollary 4.2 in \citet{firpo2021uniform} and using the
chain rule establishes the result. The proof of the validity of the
bootstrap is almost identical to the proof of Theorem \ref{thm:bootstrap}.
\end{proof}

\section{\label{sec:Figures} Figures}

\begin{figure}[H]
\begin{centering}
\caption{\label{fig:endo_treat} Heterogeneity in the treatment effects across
income groups}
\par\end{centering}
\bigskip{}

\begin{centering}
\begin{tabular}{c||c||c||c}
\multicolumn{4}{c}{\includegraphics[scale=0.65]{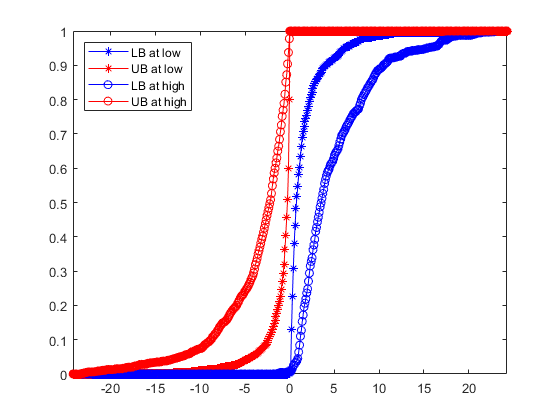}}\tabularnewline
\end{tabular}
\par\end{centering}
\bigskip{}

Note: The blue and red lines are estimated lower and upper bounds
on the conditional distributions of treatment effects, respectively.
The star-marked lines are the bounds on the conditional distribution
of treatment effects given the 0.2 quantile of the income. The circle-marked
lines are the bounds on the conditional distribution of treatment
effects given the 0.8 quantile of income. All of these bounds are
obtained under Assumptions \ref{assu:fsd1} and \ref{assu:fsd2}.
When the assumptions are not imposed, the resulting bounds are the
logical ones. 
\end{figure}

\begin{figure}[H]
\caption{\label{fig:comparison_all_quantiles} The effect of 401k on net financial
assets - Comparison of bounds across different income groups}

\bigskip{}

\begin{centering}
\begin{tabular}{ccc}
\includegraphics[scale=0.35]{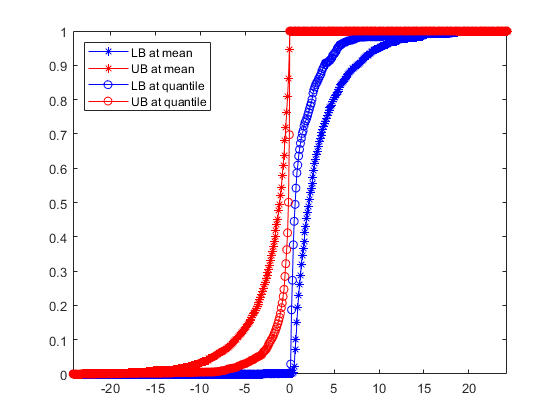} & \includegraphics[scale=0.35]{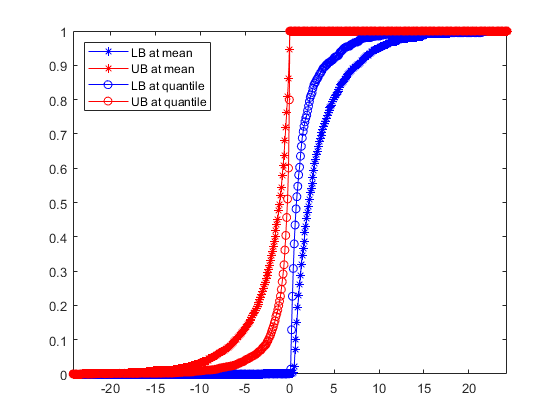} & \includegraphics[scale=0.35]{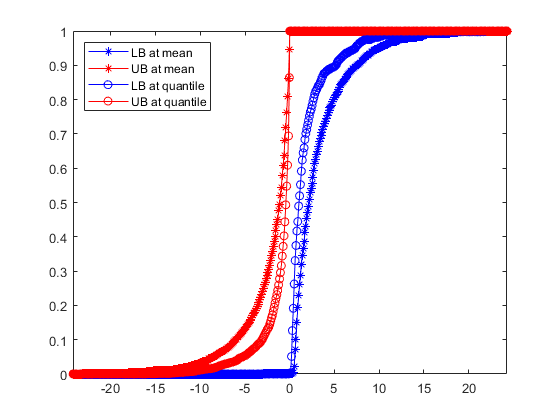}\tabularnewline
0.1 quantile of income  & 0.2 quantile of income  & 0.3 quantile of income \tabularnewline
\includegraphics[scale=0.35]{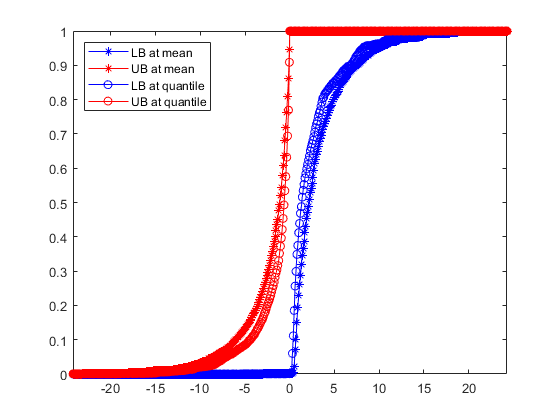} & \includegraphics[scale=0.35]{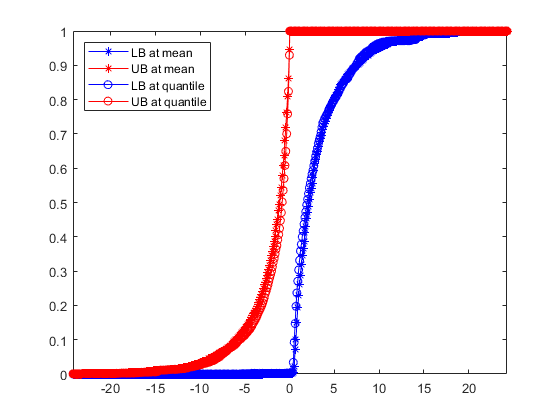} & \includegraphics[scale=0.35]{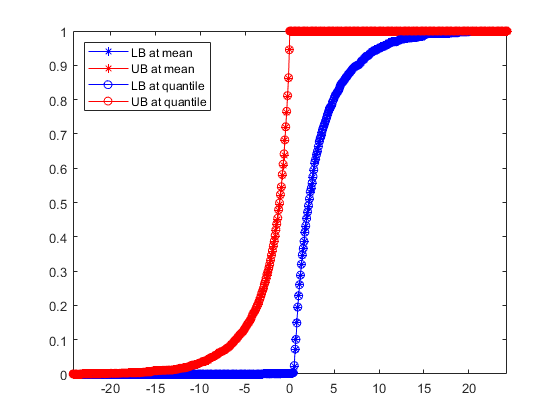}\tabularnewline
0.4 quantile of income  & 0.5 quantile of income  & 0.6 quantile of income \tabularnewline
\includegraphics[scale=0.35]{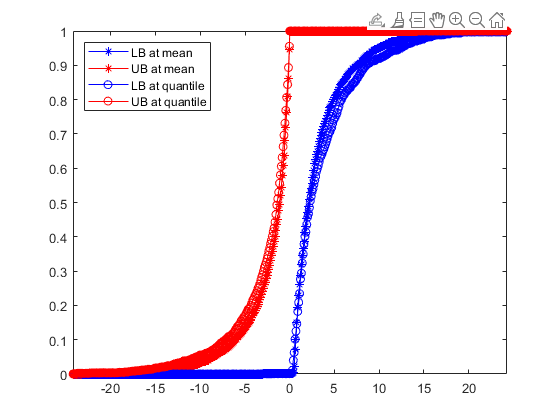} & \includegraphics[scale=0.35]{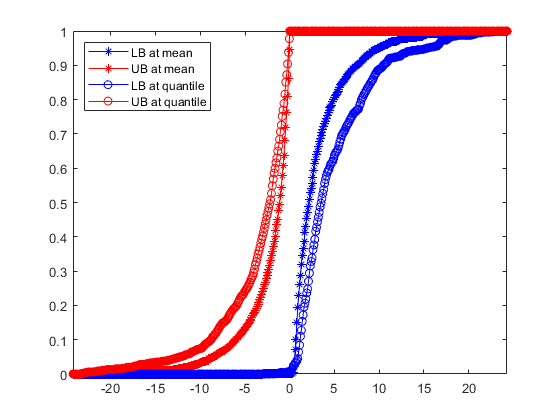} & \includegraphics[scale=0.35]{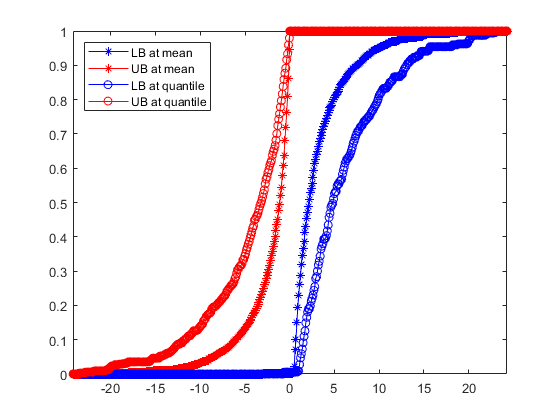}\tabularnewline
0.7 quantile of income  & 0.8 quantile of income  & 0.9 quantile of income \tabularnewline
\end{tabular}
\par\end{centering}
\bigskip{}

Note: Estimated bounds at various quantile levels and at the mean
of income are reported. These bounds are obtained under Assumptions
\ref{assu:fsd1} and \ref{assu:fsd2}. When the assumptions are not
imposed, the resulting bounds are the logical ones. 
\end{figure}

\end{document}